\documentclass[12pt,hidelinks]{article}
\usepackage{amsfonts,amsmath,amsthm,amssymb,graphicx,float,subcaption,
  color,xcolor,natbib,url,hyperref,tikz,enumerate}
\usepackage[shortlabels]{enumitem}
\usepackage{mathtools}
\DeclarePairedDelimiter\ceil{\lceil}{\rceil}
\DeclarePairedDelimiter\floor{\lfloor}{\rfloor}
\graphicspath{{figures/}}
\allowdisplaybreaks
  
\newcommand{\drawline}[2]{\protect\tikz[baseline={(a.base)}]{
    \protect\node[inner sep=4pt,outer sep=4pt] (a) {#1};
    \protect\draw[#2] {([yshift=3.5pt]a.base west) --
      ([yshift=3.5pt]a.base east)};}}
\definecolor{aqua}{RGB}{0,255,255}
\definecolor{fuchsia}{RGB}{255,0,255}

\newcommand{\bbeta}{\boldsymbol{\beta}}

\newcommand{\bdelta}{\boldsymbol{\delta}}
\newcommand{\bmu}{\boldsymbol{\mu}}
\newcommand{\bOmega}{\boldsymbol{\Omega}}
\newcommand{\bSigma}{\boldsymbol{\Sigma}}
\newcommand{\bLambda}{\boldsymbol{\Lambda}}
\newcommand{\brho}{\scalebox{1.4}{$\boldsymbol{\rho}$}}
\newcommand{\bsigma}{\scalebox{1.1}{$\boldsymbol{\Phi}$}}
\newcommand{\onen}{\frac{1}{n}}
\newcommand{\tp}{^{\rm T}}
\newcommand{\bz}{\mathbf{0}}
\newcommand{\x}{\mathbf{x}}
\newcommand{\y}{\mathbf{y}}
\newcommand{\z}{\mathbf{z}}
\newcommand{\A}{\mathbf{A}}
\newcommand{\B}{\mathbf{B}}
\newcommand{\I}{\mathbf{I}}
\newcommand{\M}{\mathbf{M}}
\newcommand{\W}{\mathbf{W}}
\newcommand{\X}{\mathbf{X}}
\newcommand{\Z}{\mathbf{Z}}
\newcommand{\sumn}{\sum_{i=1}^{n}}
\newcommand{\sumk}{\sum_{i=1}^{k}}
\newcommand{\sumr}{\sum_{i=1}^{r}}
\newcommand{\sumrn}{\sum_{i=n-r+1}^{n}}
\newcommand{\sump}{\sum_{j=1}^{p}}
\newcommand{\onek}{\frac{1}{k}}

\newcommand{\op}{o_P(1)}
\newcommand{\Op}{O_P(1)}

\newcommand{\vr}{\mathrm{var}}
\newcommand{\Exp}{\mathrm{E}}
\newcommand{\Var}{\mathrm{V}}
\newcommand{\diag}{\mathrm{diag}}

\newtheorem{thm}{Theorem}

\newtheorem{lemma}{Lemma}
\newtheorem{alg}{Algorithm}
\theoremstyle{remark}
\newtheorem{remark}{\bf Remark}

\newcommand{\blind}{1}
\def\spacingset#1{\renewcommand{\baselinestretch}
  {#1}\small\normalsize} \spacingset{1}

\begin{document}
\title{Information-Based Optimal Subdata Selection for Big Data Linear Regression}
\if0\blind{\date{}}
\else{\author{ HaiYing Wang, Min
    Yang, and John Stufken \footnote{ HaiYing Wang is Assistant
      Professor, Department of Statistics, University
      of Connecticut, Storrs, Mansfield, CT 06269
      (haiying.wang@uconn.edu).  Min
      Yang is Professor, Department of Mathematics, Statistics, and
      Computer Science, University of Illinois at Chicago, Chicago, IL
      60607 (myang2@uic.edu).  John Stufken is Charles Wexler
      Professor, School of Mathematical and Statistical Sciences,
      Arizona State University, Tempe, AZ 85287 (jstufken@asu.edu) }}
  \date{\today}}
\fi
\maketitle
\begin{abstract}
  Extraordinary amounts of data are being produced in many branches of
  science.  Proven statistical methods are no longer applicable with
  extraordinary large data sets due to computational limitations.  A
  critical step in big data analysis is data reduction. Existing
  investigations in the context of linear regression focus on
  subsampling-based methods. However, not only is this approach prone
  to sampling errors, it also leads to a covariance matrix of the
  estimators that is typically bounded from below by a term that is of the order of the inverse of the subdata size.
  We propose a novel approach, termed
  information-based optimal subdata selection (IBOSS). Compared to
  leading existing subdata methods, the IBOSS approach has the following advantages:
  (i) it is significantly faster; (ii) it is suitable for distributed
  parallel computing; (iii) the variances of the slope parameter
  estimators converge to 0 as the full data size increases even if the
  subdata size is fixed, i.e., the convergence rate depends on the
  full data size; (iv) data analysis for IBOSS subdata is
  straightforward and the sampling distribution of an IBOSS estimator
  is easy to assess. Theoretical results and extensive simulations
  demonstrate that the IBOSS approach is superior to subsampling-based
  methods, sometimes by orders of magnitude.  The advantages of the
  new approach are also illustrated through analysis of real data.
\end{abstract}
\noindent%
\vfill
\newpage
\spacingset{1.45} 

\section{Introduction}
Technological advances have enabled an exponential growth in data
collection and the size of data sets. Although computational resources
have also been growing rapidly, this pales in comparison to the
astonishing growth in data volume. This presents the challenge of drawing useful information and converting data into knowledge with available computational resources. We meet this challenge in the context of linear regression by identifying informative subdata, which can be fully analyzed.

For linear regression with a $n\times1$ response vector and
$n\times p$ covariate matrix, the full data volume is $n(p+1)$. In the
setting of ``big data'', this prevents computation of parameter estimates in a traditional way due to insufficient computational resources. The scenario with $p\gg n$ is usually referred to as high-dimensional data. Multiple methods for analyzing high-dimensional data have been proposed and studied, such as LASSO \citep{tibshirani1996regression,meinshausen2009lasso},
 Dantzig selector \citep{candes2007dantzig}, and sure
independence screening \citep{fan2008sure}, among others.
  We focus on the scenario with $n\gg p$ for an extremely large $n$.
 This is an important problem that arises in practice. For example, for the chemical sensors data in Section~\ref{sec:real-data}, $n=4,208,261$ and $p=15$. As another example, the airline on-time data set from the 2009 ASA Data Expo contains $n=123,534,969$ observations on 29 variables about flight arrival and departure information for all commercial flights within the USA, from October 1987 to April 2008. 
Clearly, not everyone has the computing resources to fully analyze the whole data in the aforementioned examples, and thus data reduction is a crucial step to extract useful information from the data. For the case of $n\gg p$, the required computing time for ordinary least squares (OLS) is
$O(np^2)$. This time complexity is too long for big data, and may even be beyond the computational capacity of available computing 
facilities. To address this computational limitation, data reduction is
important. 
Existing investigations in this direction focus on taking
random subsamples from the full data. Interesting 
studies include
\cite{drineas2006sampling,Drineas:11,Ma2014,PingMa2014-ICML,
  PingMa2014-JMLR}, among others. Parameter estimates are obtained based on a small random
subsample of the full data. Normalized statistical leverage scores are often used for
nonuniform subsampling probabilities, and the method is known as {\it algorithmic leveraging} \citep{PingMa2014-ICML}.

With exact statistical leverage scores, the computing time of this method is $O(np^2)$, just as for OLS on the full data. \cite{Drineas:12} developed a randomized algorithm to approximate leverage scores; it is $O(np\log n/\epsilon^2)$, and $o(np^2)$ if $\log n=o(p)$, where $\epsilon\in(0, 0.5]$. Thus the computing time for the algorithmic leveraging method is at least $O(np\log n/\epsilon^2)$. A subsampling-based method also induces sampling error and the information in the resultant subdata is typically proportional to its size.  We find that (details to be shown in
Section~\ref{sec:framework}), for a subsampling-based estimator, the
covariance matrix is bounded from below by a term that is typically of
order $1/k$, where $k$ is the subsample size. This order is only a
function of $k$, so that the variance does not go to 0 with increasing full data size $n$.

In this paper, we propose an alternative subdata selection approach, which we call {\bf i}nformation-{\bf b}ased {\bf o}ptimal {\bf s}ubdata {\bf s}election (IBOSS) method from big
data. The basic idea is to select the most informative data points deterministically so that subdata of a small size preserves most of the information contained in the full data. It is akin to the basic motivation of optimal experimental design \citep{kiefer1959}, which aims at obtaining the maximum information with a fixed budget. Traditionally, optimal design is not a data analysis tool, but focuses on data collection. The idea of ``maximizing'' an information matrix, however, can be borrowed to establish a framework to identify the most informative subdata from the full data for estimating unknown parameters. Using this framework, we will also gain more insight into the popular subsampling-based methods.

As we will show, the IBOSS approach has the following advantages compared to existing methods: 1) the computing time for the IBOSS algorithm is $O(np)$, which is significantly faster than existing methods; 2) the IBOSS algorithm is very suitable for distributed parallel computing platforms; it identifies informative data points by examining each covariate
individually; 3) the IBOSS method does not induce sampling error and the variance of estimators can go to 0 as the full data size $n$ becomes large even if the subdata size $k$ is fixed; and 4) in terms of distributional properties, IBOSS estimators inherit properties of estimators based on the full data. 

The remainder of the paper is organized as follows. In
Section~\ref{sec:framework}, we present the IBOSS framework and use it
to analyze the popular subsampling-based methods. A lower bound for
covariance matrices of subsampling-based estimators will be given. In
Section~\ref{sec:optim-crit-main}, we characterize IBOSS subdata
under the D-optimality criterion and use it to develop a computationally efficient algorithm. In Section~\ref{sec:asymptotic-analysis},
we evaluate the IBOSS algorithm by deriving its asymptotic
properties. In Section~\ref{sec:numer-exper}, the performance of the IBOSS method is examined through extensive simulations and two real data applications. We offer concluding remarks in Section~\ref{sec:concluding-remarks} and show all technical details in
the appendix. Additional numerical results are provided in the Supplementary Material.

\section{The framework}\label{sec:framework}
Let $(\z_1, y_1), ..., (\z_n, y_n)$ denote the full data, and assume the
linear regression model:
\begin{align}\label{eq:1}
  y_i&=\beta_0+\z_i\tp\bbeta_1+\varepsilon_i
       =\beta_0+\sum_{j=1}^{p}z_{ij}\beta_j+\varepsilon_i,\
       \quad i=1,...,n,
\end{align}
where $\beta_0$ is the scalar intercept parameter,
$\bbeta_1=(\beta_1,\beta_2, ..., \beta_p)\tp$ is a $p$-dimensional
vector of unknown slope parameters, $\z_i=(z_{i1},...,z_{ip})\tp$ is a
covariate vector, $y_i$ is a response, and $\varepsilon_i$ is an error
term. We write $\x_i=(1, \z\tp_i)\tp$,
$\bbeta=(\beta_0, \bbeta_1\tp)\tp$, $\Z=(\z_1, ..., \z_n)\tp$,
$\X=(\x_1, ..., \x_n)\tp$, $\y=(y_1, ..., y_n)\tp$, 
assume that the $y_i$'s are uncorrelated given the covariate matrix $\Z$, and that the error terms $\varepsilon_i$'s satisfy $\Exp(\varepsilon_i)=0$ and $\Var(\varepsilon_i)=\sigma^2$.

The intercept parameter is often not of interest and could be eliminated by centralizing the full data. However, this does not work for streaming data, is not practical if the focus is on building a predictive model, and requires a computing time of $O(np)$.

When using the full data and model~\eqref{eq:1}, the least-squares estimator of $\bbeta$, which is also its best linear unbiased estimator (BLUE), is
\begin{equation*}
  \hat\bbeta_f=\left(\sumn \x_i\x_i\tp\right)^{-1}\sumn \x_iy_i.
\end{equation*}
The covariance matrix of this unbiased estimator is equal to the inverse of
\begin{equation*}
  \M_{f}=\frac{1}{\sigma^2}\sumn \x_i\x_i\tp.
\end{equation*}
This is the observed Fisher information matrix for $\bbeta$ from the full data if the $\varepsilon_i$'s are normally distributed. While we do not require the normality assumption, for simplicity we will still call $\M_{f}$ the information matrix. 
Let $(\z_1^*, y_1^*), ..., (\z_k^*, y_k^*)$ be subdata of size $k$ selected deterministically from the full data, in which the rule to determine whether a data point is included or not depends on $\Z$ only. Then the subdata also follow the linear regression model
\begin{equation}\label{eq:12}
  y_i^*=\beta_0+{\z_i^*}\tp\bbeta_1+\varepsilon^*_i,\qquad i=1, ...,k,
\end{equation}
with the same assumptions and unknown parameters as for model~\eqref{eq:1}.  
 The LS estimator
\begin{equation*}
  \hat\bbeta_s=\left(\sumk \x_i^*{\x_i^*}\tp\right)^{-1}\sumk \x_i^*y_i^*,
\end{equation*}
is the BLUE of $\bbeta$ for model~\eqref{eq:12} based on the
subdata, where $\x_i^*=(1,{\z_i^*}\tp)\tp$.  The observed
information matrix for $\bbeta$ based on the subdata is
\begin{equation*}
  \M_{s}=\frac{1}{\sigma^2}\sumk \x_i^*{\x_i^*}\tp,
\end{equation*}
which is the inverse of the covariance matrix of
$\hat\bbeta_s$, namely, 
\begin{equation}\label{eq:13}
  \Var(\hat\bbeta_s|\Z)=\M_{s}^{-1}=
  \sigma^2\left(\sumk \x_i^*{\x_i^*}\tp\right)^{-1}.
\end{equation}

To estimate a linear function of $\bbeta$ using subdata, plugging in the LS
estimator yields the estimator with the minimum variance among all linear unbiased 
estimators. Since this minimum variance is a function of the covariate
values in the subdata, one can judiciously select the subdata to minimize the minimum variance. This is akin to the basic idea behind optimal experimental design \citep{kiefer1959}. We implement this here by seeking subdata that, in some sense, maximize $\M_s$.

Some additional notation will help to formulate this idea as an optimization problem. Let $\delta_i$ be the indicator variable that signifies whether
$(\z_i, y_i)$ is included in the subdata, i.e., $\delta_i=1$ if $(\z_i, y_i)$ is included and $\delta_i=0$ otherwise. The information matrix with subdata of size $k$ can then be written as
\begin{equation}\label{eq:2}
  \M(\bdelta)=\frac{1}{\sigma^2}\sumn\delta_i\x_i\x_i\tp,
\end{equation}
where $\bdelta=\{\delta_1, \delta_2, ..., \delta_n\}$ such that
$\sumn\delta_i=k$.  To have an optimal estimator based on a subdata,
one can choose a $\bdelta$ that ``maximizes'' the information
matrix \eqref{eq:2}. Since $\M(\bdelta)$ is a matrix, in optimal experimental design \citep{kiefer1959}, this is typically done by maximizing a univariate {\em optimality criterion function} of the matrix.

Let $\psi$ denote an optimality criterion function. The problem is presented as the following optimization problem given the observed big data:
\begin{equation}\label{eq:16}
  \bdelta^{opt}=\arg\max_{\bdelta}\psi\{\M(\bdelta)\},
  \quad\text{subject to}\quad\sumn\delta_i=k.
\end{equation}
A popular optimality criterion is the D-optimality criterion. It maximizes the determinant of $\M(\bdelta)$, which has the interpretation of minimizing the
expected volume of the joint confidence ellipsoid for $\bbeta$. We will come back to this criterion in greater detail in Section~\ref{sec:optim-crit-main}. 

\subsection{Analysis of existing subsampling-based methods}
\label{sec:analys-exist-subs}
Based on our information-based subdata selection framework, we can also gain insights into popular random subsampling-based methods. To see this, let
$\boldsymbol{\eta}_L$ be the $n$-dimensional count-vector whose $i$th entry denotes the number of times that the $i$th data point is included in a subsample of size $k$, which is obtained using a random subsampling method with probabilities proportional to $\pi_i$, $i=1,...,n$, such that
$\sumn\pi_i=1$. 
A subsampling-based estimator has the general form
\begin{equation}\label{eq:23}
  \tilde\bbeta_L=
  \left(\sumn w_i\eta_{Li}\x_i\x_i\tp\right)^{-1}
  \sumn w_i\eta_{Li}\x_iy_i,
\end{equation}
where the weight $w_i$ is often taken to be $1/\pi_i$. Corresponding to
different choices of $\pi_i$ and $w_i$, some popular subsampling-based
methods \citep{PingMa2014-JMLR} include: uniform subsampling (UNI)
in which $\pi_i=1/n$ and $w_i=1$; leverage-based subsampling (LEV) in
which $\pi_i=h_{ii}/(p+1)$, $w_i=1/\pi_i$ and
$h_{ii}=\x_i\tp(\X\X\tp)^{-1}\x_i$; shrinked leveraging estimator
(SLEV) in which $\pi_i=\alpha h_{ii}/(p+1)+(1-\alpha)/n$,
$w_i=1/\pi_i$ and $\alpha\in[0,1]$; unweighted leveraging estimator
(LEVUNW) in which $\pi_i=h_{ii}/(p+1)$ and $w_i=1$.
 
The distribution of $\tilde\bbeta_L$ is complicated, but we can study its performance using the proposed information-based framework. The ``information matrix''
$\M(\boldsymbol{\eta}_L)$ given $\Z$ is random because
$\boldsymbol{\eta}_L$ is random. While the $\eta_{Li}$'s are correlated, we will only need to use the marginal distribution of each $\eta_{Li}$. If subsampling is with replacement, then each $\eta_{Li}$ has a binomial distribution
Bin$(k,\pi_i)$; if the subsampling is without replacement, the marginal distribution of each $\eta_{Li}$ is Bin$(1,k\pi_i)$ under the condition that $k\pi_i\le1$. Either way, $\Exp(\eta_{Li})=k\pi_i$. Hence, taking
expectations with respect to $\boldsymbol{\eta}_L$, the expected
observed information matrix for a subsampling-based method is
\begin{align}\label{eq:22}
  \M_{EL}=\Exp\{\M(\boldsymbol{\eta}_L)|\Z\}
  &=\frac{1}{\sigma^2}\sumn\Exp(\eta_{Li})\x_i\x_i\tp
    =\frac{k}{\sigma^2}\sumn\pi_i\x_i\x_i\tp. 
\end{align}
Unlike the IBOSS approach, the inverse of $\M_{EL}$ is not the variance
covariance matrix of $\Var(\tilde\bbeta_L|\Z)$. In fact, for subsampling with replacement there is a small probability that 
$\Var(\tilde\bbeta_L|\Z)$ is not well defined because $\bbeta_L$ is not estimable. To solve this issue, we consider only subsamples with
full-rank covariate matrices to define the covariance matrix of $\tilde\bbeta_L$. 
The following theorem states a relationship between $\M_{EL}$ and the covariance matrix of $\tilde\bbeta_L$. 

\begin{thm}\label{thm:1} 
  Suppose that a subsample of size $k$ is taken using a random
  subsampling procedure with probabilities proportional to $\pi_i$,
  $i=1,...,n$, such that $\sumn\pi_i=1$. Consider the set 
  $\Delta=\{\boldsymbol{\eta}_L:\sumn\eta_{Li}\x_i\x_i\tp$ is 
  non-singular\}, where $\boldsymbol{\eta}_L$ is the $n$-dimensional 
  vector that counts how often each data point is 
  included. Let $I_{\Delta}(\boldsymbol{\eta}_L)=1$ if and only if $\boldsymbol{\eta}_L\in\Delta$. Given $I_{\Delta}(\boldsymbol{\eta}_L)=1$, $\tilde\bbeta_L$ is 
  unbiased for $\bbeta$, and 
  \begin{equation}
    \Var\{\tilde\bbeta_L|\Z,I_{\Delta}(\boldsymbol{\eta}_L)=1\} \ge
    P\{I_{\Delta}(\boldsymbol{\eta}_L)=1|\Z\}\M_{EL}^{-1}
    =\frac{\sigma^2P\{I_{\Delta}(\boldsymbol{\eta}_L)=1|\Z\}}{k}
    \left\{\sumn\pi_i\x_i\x_i\tp\right\}^{-1}
  \end{equation}
  in the Loewner ordering.
\end{thm}

\begin{remark}
Theorem~\ref{thm:1} is true regardless of the choice for the weights $w_i$'s in $\tilde\bbeta_L$ or whether subsampling is with or without replacement. It provides a lower bound for covariance matrices of  subsampling-based estimators, which provides a feasible way to evaluate the best performance of a subsampling-based estimator.   
Note that when $k\gg p+1$, $P\{{ I_{\Delta}(\boldsymbol{\eta}_L)=1}|\Z\}$ is often close to 1, so that the lower bound is close to the inverse of the expected observed information $\M_{EL}$. 
\end{remark}

\begin{remark} 
  Existing investigations on the random subsampling approach focus on sampling with replacement where the subsample is independent given the full data. For subsampling without replacement with fixed sample size, the subsample is no longer independent given the full data and the properties of the resultant estimator are more complicated. As a result, to the best of our knowledge, subsampling without replacement for a fixed subsample size from big data has never been investigated. Our IBOSS framework, however, is applicable here to assess the performance of an estimator based on subsampling without replacement. 
\end{remark}

Applying Theorem~\ref{thm:1} to the popular sampling-based methods, we obtain the following results. For the UNI method,
  \begin{equation}\label{eq:24}
    \Var\left\{\tilde\bbeta_L^{\rm UNI}\Big|\Z,{ I_{\Delta}(\boldsymbol{\eta}_L)=1}\right\}
    \ge \frac{\sigma^2P\{I_{\Delta}(\boldsymbol{\eta}_L)=1|\Z\}}{k}
    \left(\onen\sumn\x_i\x_i\tp\right)^{-1};
  \end{equation}
for the LEV (or LEVUNW) method, 
  \begin{align}\label{eq:25}
    \Var\left\{\tilde\bbeta_L^{\rm LEV}\Big|\Z, I_{\Delta}(\boldsymbol{\eta}_L)=1\right\}
    \ge \frac{(p+1)\sigma^2P\{I_{\Delta}(\boldsymbol{\eta}_L)=1|\Z\}}{k}
    \left\{\sumn\x_i\x_i\tp\left(\X\tp\X\right)^{-1}
    \x_i\x_i\tp\right\}^{-1};
  \end{align}
for the SLEV method,
  \begin{align}\label{eq:26}
    \Var\left\{\tilde\bbeta_L^{\rm SLEV}\Big|\Z, I_{\Delta}(\boldsymbol{\eta}_L)=1\right\}
    \ge &\frac{\sigma^2P\{I_{\Delta}(\boldsymbol{\eta}_L)=1|\Z\}}{k}
    \left\{\frac{\alpha}{p+1}\sumn\x_i\x_i\tp\left(\X\tp\X\right)^{-1}
    \x_i\x_i\tp\right.\notag\\
    &\hspace{5.9cm}\left.+\frac{1-\alpha}{n}\sumn\x_i\x_i\tp\right\}^{-1}.
  \end{align}
  
If $\x_i$, $i=1, ..., n$, are generated independently from the same distribution as random vector $\x$ with finite second moment, i.e., $\Exp\|\x\|^2<\infty$, then from the strong law of large numbers,
\begin{equation}\label{eq:27}
  \onen\sumn\x_i\x_i\tp \rightarrow\Exp(\x\x\tp),
\end{equation}
almost surely as $n\rightarrow\infty$. If we further assume that the fourth moment of $\x$ is finite, i.e., $\Exp\|\x\|^4<\infty$, then we have
\begin{equation}\label{eq:28}
  \sumn\x_i\x_i\tp\left(\X\tp\X\right)^{-1}\x_i\x_i\tp
  \rightarrow \Exp\big[\x\x\tp\{\Exp(\x\x\tp)\}^{-1}\x\x\tp\big]
\end{equation}
almost surely as $n\rightarrow\infty$. 
 Let $P_{\eta}= \underset{n\rightarrow\infty}{\lim\inf}P\{I_{\Delta}(\boldsymbol{\eta}_L)=1|\Z\}$. Note that $P_{\eta}=1$ under some mild condition, e.g., the covariate distribution is continuous. 
From \eqref{eq:24} -~\eqref{eq:28}, we have
  \begin{equation}\label{eq:29}
    \Var\left\{\tilde\bbeta_L^{\rm UNI}\Big|\Z, I_{\Delta}(\boldsymbol{\eta}_L)=1\right\}
    \ge \frac{\sigma^2P_{\eta}}{k}
    \left\{\Exp(\x\x\tp)\right\}^{-1},
  \end{equation}
  \begin{align}\label{eq:30}
    \Var\left\{\tilde\bbeta_L^{\rm LEV}\Big|\Z, I_{\Delta}(\boldsymbol{\eta}_L)=1\right\}
    \ge \frac{(p+1)\sigma^2P_{\eta}}{k}
    \left(E[\x\x\tp\{\Exp(\x\x\tp)\}^{-1}\x\x\tp]\right)^{-1},
  \end{align}
  \begin{align}\label{eq:31}
    \Var\left\{\tilde\bbeta_L^{\rm SLEV}\Big|\Z, I_{\Delta}(\boldsymbol{\eta}_L)=1\right\}
    \ge &\frac{\sigma^2P_{\eta}}{k}
          \Bigg(\frac{\alpha}{p+1}E[\x\x\tp\{\Exp(\x\x\tp)\}^{-1}\x\x\tp]
      +(1-\alpha)\Exp(\x\x\tp)\Bigg)^{-1},
  \end{align}
  almost surely as $n\rightarrow\infty$. 
From \eqref{eq:29}, \eqref{eq:30} and \eqref{eq:31}, one sees that covariance matrices of these commonly used subsampling-based estimators are bounded from below in the Loewner ordering by finite quantities that are at the order of $1/k$. These quantities do not go to 0 as the full data sample
size $n$ goes to $\infty$. 

\section{The D-optimality criterion and an IBOSS algorithm}
\label{sec:optim-crit-main}
In this section, we study the commonly used D-optimality criterion and
develop IBOSS algorithms based on theoretical characterizations of IBOSS
subdata under this criterion. 

In our framework, for given full data of size $n$, the
D-optimality criterion suggests the selection of subdata of size $k$ so that
\begin{equation*}
  \bdelta^{opt}_{\mathrm{D}}=
  \arg\max_{\bdelta}\left|\sumn\delta_i\x_i\x_i\tp\right|,\quad
  \sumn\delta_i=k.
\end{equation*}
Obtaining an exact solution is computationally far too expensive. In working towards an approximate solution, we first derive an upper bound for $|\M(\bdelta)|$ which, while only attainable for very special cases, will guide our later algorithm.

\begin{thm}[D-optimality]\label{thm:2} For subdata of size $k$
  represented by $\bdelta$,
  \begin{align}\label{eq:5}
    |\M(\bdelta)|\le\frac{k^{p+1}}{4^p\sigma^{2(p+1)}}
    \prod_{j=1}^p(z_{(n)j}-z_{(1)j})^2,
  \end{align}
  where $z_{(n)j}=\max\{z_{1j}, z_{2j}, ..., z_{nj}\}$ and
  $z_{(1)j}=\min\{z_{1j}, z_{2j}, ..., z_{nj}\}$ are the
  $n$th and first order statistics of $z_{1j}, z_{2j}, ..., z_{nj}$.
  If the subdata consists of the $2^p$ points
    $(a_{1},\ldots,a_{p})\tp$ where $a_{j}=z_{(n)j}$ or $z_{(1)j}$,
    $j=1, 2, ..., p$, each occurring equally often, then equality holds in~\eqref{eq:5}.
\end{thm}

\begin{remark}
  Often $k$ is much smaller than $2^p$, so that subdata with equality in 
  Theorem~\ref{thm:2} will not exist. However, just as for Hadamard's determinant bound \citep{hadamard1893}, the result suggests to collect subdata with extreme covariate values, both small and large, occurring with the same frequency. 
  This agrees with the common statistical knowledge that larger variation in covariates is more informative and results in better parameter estimation.
\end{remark}

The following algorithm is motivated by the result in Theorem~\ref{thm:2}.

\begin{alg}[Algorithm motivated by D-optimality]\label{alg:1}
  Suppose that $r=k/(2p)$ is an integer. Using a partition-based
  selection algorithm \citep{Martinez2004}, perform the following steps:
  \begin{enumerate}[(1)]
  \item For $z_{i1}$, $1\le i\le n$, include $r$ data points with the
    $r$ smallest $z_{i1}$ values and $r$ data points with the $r$
    largest $z_{i1}$ values;
  \item For $j=2, ...,p$, exclude data points that were previously
    selected, and from the remainder select $r$ data points with the
    smallest $z_{ij}$ values and $r$ data points with the largest
    $z_{ij}$ values.
  \item Return
    $\hat\bbeta^{\mathrm{D}}=\{(\X^*_{\mathrm{D}})\tp\X^*_{\mathrm{D}}\}^{-1}
    (\X^*_{\mathrm{D}})\tp\y^*_{\mathrm{D}}$
    and the estimated covariance matrix $\hat\sigma^2_{\mathrm{D}}\{(\X^*_{\mathrm{D}})\tp\X^*_{\mathrm{D}}\}^{-1}$, where
    $\X^*_{\mathrm{D}}=(\mathbf{1},\Z^*_{\mathrm{D}})$,
    $\Z^*_{\mathrm{D}}$ is the covariate matrix of the subdata
    selected in the previous steps, $\y^*_{\mathrm{D}}$ is the
    response vector of the subdata and
    $\hat\sigma^2_{\mathrm{D}} = \big\|\y^*_{\mathrm{D}} -
    \X^*_{\mathrm{D}} \hat\bbeta^{\mathrm{D}} \big\|^2/(k-p-1)$.
  \end{enumerate}
\end{alg}

\begin{remark}
  For each covariate, a partition-based selection algorithm has an
  average time complexity of $O(n)$ to find the $r$th largest or smallest value
  \citep{musser1997introspective, Martinez2004}.
  Thus the time to obtain
  the subdata is $O(np)$. Using the subdata, the computing time for $\hat\bbeta^{\mathrm{D}}$ and $\hat\sigma^2_{\mathrm{D}}$ is $O(kp^2+p^3)$ and $O(kp)$, respectively. Thus, the time complexity of
  Algorithm~\ref{alg:1} is $O(np+kp^2+p^3+kp)=O(np+kp^2)$. For the
  scenario that $n>kp$, this reduces to $O(np)$. This algorithm is
  faster than algorithmic leveraging, which has a computing time of
  $O(np\log n)$ \citep{Drineas:12}.
\end{remark}
\begin{remark}
  Algorithm~\ref{alg:1} gives the covariance matrix of the resultant
  estimator, which is very crucial for statistical inference. This is
  the exact covariance matrix of $\hat\bbeta^{\mathrm{D}}$ if the
  variance of the error term, $\sigma^2$, is known. With
  an additional assumption of normality of $\varepsilon_i$,
  $\hat\bbeta^{\mathrm{D}}$ has an exact normal distribution.
\end{remark}
\begin{remark}
  Algorithm~\ref{alg:1} is naturally suited for distributed storage
  and processing facilities for parallel computing. One can simultaneously process each
  covariate and find the indexes of its extreme values. These indexes can then
 be combined to obtain the subdata. While this approach may result in a subdata size that is smaller than $k$ if there is duplication of indexes, the resultant estimator will still have the same convergence rate.
\end{remark}

\begin{remark}
  Algorithm~\ref{alg:1} selects subdata according to extreme values of
  each covariate, which may include outliers. However, the
  selection rule is ancillary, and the resultant subdata follow the same underlying regression model as the full data. We can thus use outlier
  diagnostic methods to identify outliers in the subdata. If there are
  outliers in the full data, it is very likely that these data points will be
  identified as outliers in the subdata. On the other hand, if there are data points that are far from others but still follow the underlying model, then
  these data points actually contain more information about the model and should be used for parameter estimation.
\end{remark}

\begin{remark}
  The restriction that the subdata sample size $k$ is chosen to make $r=k/(2p)$ an integer is mostly for convenience. In the case that $r=k/(2p)$ is not an integer, one can either adjust $k$ by using the floor $\floor{r}$ or the ceiling $\ceil{r}$, or use a
  combination of $\floor{r}$ and $\ceil{r}$ to keep the subdata sample
  size as $k$.
\end{remark}

 The following theorem gives some insight on the quality of using Algorithm~\ref{alg:1} to approximate the upper bound of $|\M(\bdelta)|$ in Theorem~\ref{thm:2}. 

\begin{thm}\label{thm:3}
 Let $\Z^*_{\mathrm{D}}$ be the
  covariate matrix for the subdata of size $k=2pr$ selected using
  Algorithm~\ref{alg:1} and $\X^*_{\mathrm{D}}=(\mathbf{1},\Z^*_{\mathrm{D}})$. The determinant $|(\X^*_{\mathrm{D}})\tp\X^*_{\mathrm{D}}|$ satisfies
\begin{equation}\label{eq:36}
  \frac{|(\X^*_{\mathrm{D}})\tp\X^*_{\mathrm{D}}|}
  {\frac{k^{p+1}}{4^p}\prod_{j=1}^p(z_{(n)j}-z_{(1)j})^2}
  \ge\frac{\lambda_{\min}^p(\mathbf{R})}{p^p}
  \prod_{j=1}^p\left(\frac{z_{(n-r+1)j}-z_{(r)j}}{z_{(n)j}-z_{(1)j}}\right)^2,
\end{equation}
where $\lambda_{\min}(\mathbf{R})$ is the smallest eigenvalue of $\mathbf{R}$, the sample correlation matrix of $\Z^*_{\mathrm{D}}$.
\end{thm}

From this theorem, it is seen that although Algorithm~\ref{alg:1} may not achieve the unachievable upper bound in Theorem~\ref{thm:2}, it may achieve the same order. For example, if $p$ is fixed and
$\underset{n\rightarrow\infty}{\lim\inf}\
\lambda_{\min}(\mathbf{R})>0$,
then under reasonable assumptions, the lower bound in Equation~\eqref{eq:36} will not converge to 0 as
$n\rightarrow\infty$. This means that $|(\X^*_{\mathrm{D}})\tp\X^*_{\mathrm{D}}|$ is of the
same order as the upper bound for $|\M(\bdelta)|$ in Theorem~\ref{thm:2}, even though the latter is typically not attainable. 

\section{Properties of parameter estimator}
\label{sec:asymptotic-analysis}
In this section, we investigate the theoretical properties of the D-optimality motivated IBOSS algorithm, and provide both finite sample assessment and asymptotic results. These results are provided to evaluate the performance of the proposed method and show more insights about the IBOSS approach. The application of the D-optimality motivated IBOSS algorithm does not depend on asymptotic properties of the approach. 

Since $\hat\bbeta^{\mathrm{D}}$ is unbiased for $\bbeta$, we focus on its variance. The next theorem gives bounds on variances of estimators of the intercept and slope parameters from the D-optimality motivated algorithm.

\begin{thm}\label{thm:4}
If $\lambda_{\min}(\mathbf{R})>0$, then, the following results hold for the estimator, $\hat\bbeta^{\mathrm{D}}$, obtained from Algorithm~\ref{alg:1}:
\begin{align}
 &\Var(\hat\beta^{\mathrm{D}}_0|\Z)\ge\frac{\sigma^2}{k},\label{eq:63}\\
 \frac{4\sigma^2}{k\lambda_{\max}(\mathbf{R})(z_{(n)j}-z_{(1)j})^2}
  \le&\Var(\hat\beta^{\mathrm{D}}_j|\Z)
  \le\frac{4p\sigma^2}
   {k\lambda_{\min}(\mathbf{R})(z_{(n-r+1)j}-z_{(r)j})^2},
     \quad j=1, ..., p.\label{eq:64}
\end{align}
\end{thm}
Theorem~\ref{thm:4} describes finite sample properties of the proposed estimator and does not require any quantity to go to $\infty$. It shows that the variance of the intercept estimator is bounded from below by a term proportional to the inverse subdata size. This is similar to the results for existing subsampling methods. However, for the slope estimator, the variance is bounded from above by a term that is proportional to $\frac{p}{k(z_{(n-r+1)j}-z_{(r)j})^2}$, which may converge to 0 as $n$ increases even when the subdata size $k$ is fixed. We present this asymptotic result in the following theorem. 

\begin{thm}\label{thm:4-2}
Assume that covariate distributions are in the domain of attraction of the
generalized extreme value distribution, and  $\underset{n\rightarrow\infty}{\lim\inf}\lambda_{\min}(\mathbf{R})>0$. For large enough $n$, the following results hold for the estimator, $\hat\bbeta^{\mathrm{D}}$, obtained from Algorithm~\ref{alg:1}: 
\begin{align}
 &\Var(\hat\beta^{\mathrm{D}}_j|\Z)=O_P\left\{
   \frac{p}{k(z_{(n-r+1)j}-z_{(r)j})^2}\right\},
   \quad j=1, ..., p.\label{eq:9}
\end{align}
Furthermore, 
\begin{align}
 &\Var(\hat\beta^{\mathrm{D}}_j|\Z)\asymp_P 
   \frac{p}{k(z_{(n)j}-z_{(1)j})^2},
   \quad j=1, ..., p,\label{eq:18}
\end{align}
if one of the following conditions holds: 1)
$r$ is fixed; 2) the support of $F_j$ is bounded,  $r\rightarrow\infty$, and $r/n\rightarrow0$, where $F_j$ is the marginal distribution function of the $j$th component of $\z$; 3) the upper endpoint for the support of $F_j$ is $\infty$ and the lower endpoint for the support of $F_j$ is finite, and $r\rightarrow\infty$ slow enough such that
\begin{equation}\label{eq:8}
  \frac{r}{n[1-F_j\{(1-\epsilon)F_j^{-1}(1-n^{-1})\}]}\rightarrow0,
\end{equation}
for all $\epsilon>0$; 4) the upper endpoint for the support of $F_j$ is finite and the lower endpoint for the support of $F_j$ is $-\infty$, and $r\rightarrow\infty$ slow enough such that
\begin{equation}\label{eq:33}
  \frac{r}{nF_j\{(1-\epsilon)F_j^{-1}(n^{-1})\}}\rightarrow0,
\end{equation}
for all $\epsilon>0$; 5) the upper endpoint and the lower endpoint for the support of $F_j$ are $\infty$ and $-\infty$, respectively, and \eqref{eq:8} and \eqref{eq:33} hold. 
\end{thm}

Equation~\eqref{eq:9} gives a general result on the variance of a slope estimator. It holds for any values of $n$, $r$ and $p$, so that it can also be used to obtain asymptotic results when one or more of $n$, $r$ and $p$ go to infinity. The expression shows that if ${p}/{(z_{(n-r+1)j}-z_{(r)j})^2}=\op$, then the convergence of the variance would be faster than $1/k$, the typical convergence rate for a subsampling method \citep{PingMa2014-JMLR, WangZhuMa2017}. Note that the results are derived from the upper bound in \eqref{eq:64}, and thus the real convergence of the variance can be faster than $\frac{p}{k(z_{(n-r+1)j}-z_{(r)j})^2}$.

For the condition in \eqref{eq:8}, it can be satisfied by many commonly seen distributions, such as exponential distribution, double exponential distribution, lognormal distribution, normal distribution, and gamma distribution \citep{hall1979relative}. For different distributions, the required rate at which $r\rightarrow\infty$ is different. For example, if $F_j$ is a normal distribution function, then \eqref{eq:8} holds if and only if $\log r/\log\log n\rightarrow0$; if $F_j$ is an exponential distribution function, then \eqref{eq:8} holds if and only if $\log r/\log n\rightarrow0$. The condition in \eqref{eq:33} is the same as that in~\eqref{eq:8} if one take $\z=-\z$. 

For the result in \eqref{eq:18}, it can be shown from the proof that 
\begin{align*}
 \Var(\hat\beta^{\mathrm{D}}_j|\Z)\asymp_P 
  \frac{p}{k\big\{F_j^{-1}(1-n^{-1})-F_j^{-1}(n^{-1})\big\}^2}, \quad
  j=1, ..., p.
\end{align*}

What we find more interesting is the fact that, when $k$ is fixed, from Theorems 2.8.1 and 2.8.2 in \cite{galambos1987asymptotic}, $z_{(n-r+1)j}-z_{(r)j}$ goes to infinity with the same rate as that of $z_{(n)j}-z_{(1)j}$. Thus the order of the variance of a slope estimator is the inverse of the squared full data sample range for the corresponding covariate. If the sample range goes to $\infty$ as $n\rightarrow\infty$, then the variance converges to 0 even when the subdata size $k$ is fixed. This suggests that subdata may preserve information at a scale related to the full data size. We will return to this for specific cases with more details. In the remainder of this section, we focus on the case that both $p$ and $k$ are fixed.

That the variance $\Var(\hat{\beta}^{\mathrm{D}}_0|\Z)$ does not go to 0
for a fixed subdata size $k$ is not a concern if inference for the slope parameters is of primary interest, as is often the case.  However, if the focus is on building a predictive model, the intercept
needs to be estimated more precisely. This can be done by using the full data
means, $\bar y$ and $\bar\z$, say. After obtaining the slope estimator $\hat{\bbeta}^{\mathrm{D}}_1$, compute  the following adjusted estimator of the intercept
\begin{equation}\label{eq:66}
  \hat\beta_0^{Da}=\bar y-\bar\z\tp\hat{\bbeta}^{\mathrm{D}}_1.
\end{equation}
The estimator 
$\hat\beta_0^{Da}$ has a convergence rate similar to that of the slope
parameter estimators, because 
$\hat\beta_0^{Da}-\beta_0 =(\hat\beta_0^{\mathrm{full}}-\beta_0)
+\bar\z\tp(\hat\bbeta_1^{\mathrm{full}}-\bbeta_1)
-\bar\z\tp(\hat{\bbeta}^{\mathrm{D}}_1-\bbeta_1)$
and the last term is the dominating term if
$\Exp(\z)\neq\bz$. The rate may be faster than
that of the slope parameter estimators if $\Exp(\z)=\bz$. The additional computing time for this approach is $O(np)$, but the estimation
efficiency for $\beta_0$ will be substantially improved. We will demonstrate this numerically in Section~\ref{sec:numer-exper}.

Whereas Theorem~\ref{thm:4-2} provides a general result for the variance of individual parameter estimators, more can be said for special cases. The next theorem studies the structure of the covariance matrix for estimators based on Algorithm~\ref{alg:1} under various assumptions. 
\begin{thm}\label{thm:5}
  Let $\bmu=(\mu_1, ..., \mu_p)\tp$ and let $\bSigma=\bsigma\brho\bsigma$ be a full rank covariance matrix, 
  where $\bsigma=\diag(\sigma_1, ..., \sigma_p)$ is a diagonal matrix of standard deviations and $\brho$ is a correlation matrix. Assume that $\z_i$'s, $i=1, ..., n$, are i.i.d. with a distribution specified below. The following results hold for $\hat\bbeta^{\mathrm{D}}$, the estimator from Algorithm~\ref{alg:1}.\\
  (i) For multivariate normal covariates, i.e., $\z_i\sim N(\bmu,\bSigma)$,
  \begin{equation}\label{eq:17}
    \Var(\A_n\hat\bbeta^{\mathrm{D}}|\Z)=\frac{\sigma^2}{2k}
    \begin{bmatrix}
      2 & \bz \\
      \bz & p(\bsigma\brho^2\bsigma)^{-1}
    \end{bmatrix} +O_P\left(\frac{1}{\sqrt{\log n}}\right),
  \end{equation}
  where $\A_n=\diag\big(1, \sqrt{\log n}, ..., \sqrt{\log n}\big)$.\\
  (ii) For multivariate lognormal covariates, i.e.,
  $\z_i\sim\mathrm{LN}(\bmu,\bSigma)$,
  \begin{align}\label{eq:14}
    \Var(\A_n\hat\bbeta^{\mathrm{D}}|\Z)=\frac{2\sigma^2}{k}
    \begin{bmatrix}
      1 & -\mathbf{u}\tp \\
      -\mathbf{u} & p\bLambda+\mathbf{u}\mathbf{u}\tp,
    \end{bmatrix}+\op
  \end{align}
  where $\A_n=\diag\Big\{1, \exp\big(\sigma_1\sqrt{2\log n}\big), ..., \exp\big(\sigma_p\sqrt{2\log n}\big)\Big\}$, $\mathbf{u}=(e^{-\mu_1}, ..., e^{-\mu_p})\tp$ and $\bLambda=\diag(e^{-2\mu_1}, ..., e^{-2\mu_p})$.
\end{thm}

For the distributions in Theorem~\ref{thm:5},
$\Var(\hat{\beta}^{\mathrm{D}}_0|\Z)$ is proportional to $1/k$ and
never converges to 0 with a fixed $k$. Based on
Theorem~\ref{thm:5}, $\Var(\hat{\bbeta}^{\mathrm{D}}_1|\Z)$ converges
to 0 at different rates for different distributions. When $\z$ has a
normal distribution, the convergence rate of components of
$\Var(\hat{\bbeta}^{\mathrm{D}}_1|\Z)$ is $1/\log n$. When $\z$ has a lognormal distribution,
the component $\Var_{j_1j_2}(\hat{\bbeta}^{\mathrm{D}}_1|\Z)$ has a
convergence rate
$\exp\big\{-(\sigma_{j_1}+\sigma_{j_2})\sqrt{2\log n}\big\}$,
$j_1,j_2=1,...,p$. In comparison, for most popular subsampling-based
methods, from the results in Section~\ref{sec:analys-exist-subs}, for
the normal and lognormal distributions, variances of the slope parameter estimators never
converge to 0 because they are bounded from below by terms that are proportional to $1/k$.

For a subsampling-based method, if the covariate distribution is sufficiently  heavy-tailed, then some components of the lower bound in Theorem \ref{thm:1} can go to 0. However, even then the convergence rate is much slower than that for the IBOSS approach, which may then produce an
estimator with a convergence rate close to that of the full data
estimator. For example, Table~\ref{tab:1} summarizes the orders of
variances for parameter estimators when the only covariate $z$ in a simple linear regression model has a $t$ distribution with degrees of freedom $\nu$. Three approaches are compared: the D-optimality motivated IBOSS approach (D-OPT), the UNI approach and the full data approach (FULL).

For $\beta_0$, neither subdata approach produces a variance that goes
to 0. For $\beta_1$, the D-OPT IBOSS approach results in a variance
that goes to 0 at a rate of $n^{-2/\nu}$.  When $\nu\le2$, the
variance of the estimator based on the full data goes to 0 at a rate
that is slower than $n^{-(2/\nu+\alpha)}$ for any $\alpha>0$, so
that the D-OPT IBOSS approach reaches a rate that is very close to that of the full data.
For the UNI approach,
the lower bound of the variance goes to 0 at a much slower rate. Note
that the convergence to 0 does not contradict the conclusion in
(\ref{eq:29}), which assumes that the $\x_i$ are i.i.d. with
finite second moment.

\begin{table}
  \begin{center}\def\arraystretch{1.1}
    \caption{Orders of variances and orders of lower bounds of
      variances when the covariate has a $t_\nu$ distribution.  The
      orders are in probability.}
    \label{tab:1}
    \begin{tabular}{lccccc}\hline
      Methods& \multicolumn{4}{c}{Covariates are $t_\nu$}\\
      \cline{2-5}
             & $\beta_0$  &  & \multicolumn{2}{c}{$\beta_1$} \\
      \cline{4-5}
             &   &  & $\nu>2$ & $\nu\le2$ \\
      \cline{2-5}
      D-OPT & $1/k$ &  & $1/(kn^{2/\nu})$ & $1/(kn^{2/\nu})$ \\
      UNI   & $1/k$ &  & $1/k$           &
                                           slower than $1/\{kn^{(2/\nu-1+\alpha)}\}$ for any $\alpha>0$ \\
      FULL  & $1/n$ &  & $1/n$           & 
      slower than $1/\{kn^{(2/\nu+\alpha)}\}$ for any $\alpha>0$ \\
      \hline
    \end{tabular}
  \end{center}
\end{table}

\section{Numerical experiments}\label{sec:numer-exper}
Using simulated and real data, we will now evaluate the performance of the IBOSS method. 
\subsection{Simulation studies}\label{sec:simulation}
Data are generated from the linear model~\eqref{eq:1} with the true value of $\bbeta$ being a 51 dimensional vector of unity and $\sigma^2=9$. An intercept is included so $p=50$. Let $\bSigma$ be a covariance matrix with $\Sigma_{ij}=0.5^{I(i\neq j)}$, for $i, j=1, ..., 50$,  where $I()$ is the indicator function. Covariates $\z_i$'s are generated according to the following scenarios.\\[-10mm]
\begin{enumerate}[{Case} 1., leftmargin=*]
  \addtolength{\itemsep}{-0.5\baselineskip}
\item $\z_i$'s have a multivariate normal distribution, i.e., 
  $\z_i\sim N(\mathbf{0}, \bSigma)$.
\item $\z_i$'s have a multivariate lognormal distribution, i.e., 
  $\z_i\sim \mathrm{LN}(\mathbf{0}, \bSigma)$.
\item $\z_i$'s have a multivariate $t$ distribution with degrees of freedom $\nu=2$, i.e., $\z_i\sim t_2(\mathbf{0}, \bSigma)$.
\item $\z_i$'s have a mixture distribution of four different distributions,
  $N(\mathbf{1}, \bSigma)$, $t_2(\mathbf{1}, \bSigma)$,
  $t_3(\mathbf{1}, \bSigma)$, U$[\mathbf{0},\mathbf{2}]$ and
  $\mathrm{LN}(\mathbf{0}, \bSigma)$ with equal proportions, where
  U$[\mathbf{0},\mathbf{2}]$ means its components are independent
  uniform distributions between 0 and 2.
\item $\z_i$'s consist of multivariate normal random variables with
  interactions and quadratic terms. To be specific, denote
  $\mathbf{v}=(v_1, ..., v_{20})\tp\sim N(\mathbf{0},
  \bSigma_{20\times20})$,
  where $\bSigma_{20\times20}$ is the 20 by 20 upper diagonal
  sub-matrix of $\bSigma$. Let
  $\z=(\mathbf{v}\tp, v_1\mathbf{v}\tp, v_2v_{11},v_2v_{12}, ...,
  v_2v_{20})\tp$
  and $\z_i$'s are generated from the distribution of $\z$.
\end{enumerate}

The simulation is repeated $S=1000$ times and empirical mean squared
errors (MSE) are calculated using
$\text{MSE}_{\beta_0}=S^{-1}\sum_{s=1}^S(\hat{\beta}_0^{(s)}-\beta_0)^2$
and
$\text{MSE}_{\bbeta_1}=S^{-1}\sum_{s=1}^S\|\hat{\bbeta}_1^{(s)}-\bbeta_1\|^2$
for intercept and slope estimators from different approaches, where
$\hat{\beta}_0^{(s)}$, and $\hat{\bbeta}_1^{(s)}$ are estimates in the
$s$th repetition. We compare four different approaches:
D-OPT, the D-optimality motivated IBOSS algorithm described in
Algorithm~\ref{alg:1} (black solid line
\drawline{{\color{black}{\large$\bf\circ$}}}{black, line width=
  1.6pt}), 
UNI (green short dotted line
\drawline{{\color{green}$\boldsymbol{\large +}$}}{{green,line width=
    1.6pt,dash pattern=on 2pt off 2pt}}), LEV (blue dashed line
\drawline{{\color{blue}$\boldsymbol{\large\times}$}}{{blue,line width=
    1.6pt, dash pattern=on 1pt off 3pt on 4pt off 3pt}}), and FULL,
the full data approach (aqua long dashed line
\drawline{{\color{aqua}$\large\boldsymbol{\diamond}$}}{{aqua,line
    width= 1.6pt,dash pattern=on 7pt off 3pt}}). To get the best
performance of the LEV method in parameter estimation, exact
statistical leverage scores are used to calculate the subsampling
probabilities. Note that for linear regression, the divide-and-conquer method produces results that are identical to these from the FULL \citep{LinXie2011,schifano2016online}, while the computational cost is not lower than the FULL. Thus the comparisons between the
  IBOSS approach and the full data approach reflect the relative
  performance of the IBOSS and the divide-and-conquer method in the
  context of linear regression.

For full data sizes $n=5\times10^3$, $10^4$, $10^5$ and $10^6$ and fixed subdata size $k=10^3$, Figures~\ref{fig:2} and \ref{fig:3} present plots of the log$_{10}$ of the MSEs against log$_{10}(n)$. Figure~\ref{fig:2} gives the log$_{10}$ of the MSEs for estimating the slope parameter $\bbeta_1$ using different methods. As seen in the plots, the D-OPT IBOSS method uniformly dominates the
subsampling-based methods UNI and LEV, and its advantage is more
significant if the tail of the covariate distribution is heavier.
More importantly, the MSEs from the D-OPT IBOSS method for estimating
$\bbeta_1$ decrease as the full data sample size $n$ increases, even
though the subdata size is fixed at $k=10^3$. For the
  normal covariate distribution in Figure~\ref{fig:2}(a), the
  decrease in the MSE for the D-OPT IBOSS estimator is not as evident because, as shown in Theorem~\ref{thm:5}, the convergence rate of variances for this case is as slow as $p/k/\log n$. To show the
  relative performance of the D-OPT IBOSS approach compared to that
  of the subsampling-based approaches, in the right panel of
  Figure~\ref{fig:2}(a), we scale all the MSEs so that the MSEs for
  the UNI method are one. From this figure, the MSEs for the D-OPT
  IBOSS approach are about 80\% of those for the subsampling-based
  approaches.

Unlike the IBOSS method, the random subsampling-based methods yield
MSEs that show very little change with increasing $n$ except for the $t_2$ and mixture covariate distributions. This agrees with the conclusion in Theorem~\ref{thm:1} that the covariance matrix for a random subsampling-based estimator is bounded from below by a matrix that depends only on the subdata size $k$ if the fourth moment of the covariate distribution is finite. For the $t_2$ and mixture covariate distributions, the second moments of the covariate distributions are not finite and we see that the MSEs decrease as $n$ becomes larger. However, the convergence rates are much slower than for the IBOSS method.

For the full data approach, where all data points are used, the MSEs
decrease as the size $n$ increases. It is noteworthy that
the performance of the D-OPT IBOSS method can be comparable to that of
using the full data for estimating $\bbeta_1$. For example, as shown
in Figure~\ref{fig:2} (d) for the mixture of distributions, an analysis using the D-OPT IBOSS method with subdata size $k=10^3$ from full data of size
$n = 10^6$ outperforms a full data analysis with data of size
$n = 10^5$; the MSE from the full data analysis is 2.4 times as large
as the MSE from using subdata of size $k=10^3$.

Figure~\ref{fig:3} gives results for estimating the intercept
parameter $\beta_0$. In general, the D-OPT
IBOSS method is superior to other subdata based methods, but its MSE
does not decrease as the full data size increases. This agrees with the
result in Theorem~\ref{thm:4}. We also calculate the MSE of the adjusted estimator in \eqref{eq:66}, $\hat\beta_0^{Da}=\bar y-\bar\z\tp\hat{\bbeta}^{\mathrm{D}}_1$, which is labeled with D-OPTa (red dashed line
\drawline{{\color{red}{\large$\bf\circ$}}}{{red,line
    width= 1.6pt,dash pattern=on 5pt off 5pt}}) in Figure~\ref{fig:3}. It is seen that the relative performances of $\hat{\beta}_0^{Da}$ for Cases 2, 4 and 5 are similar to those of the slope estimator, which agrees with the asymptotic properties discussed below \eqref{eq:66} in Section~\ref{sec:asymptotic-analysis}. For Cases 1 and 3, the results are very interesting in that $\hat{\beta}_0^{Da}$ performs as good as the full data approach, which means that the convergence rate of $\hat{\beta}_0^{Da}$ is much faster than that of the slope estimator. This seems surprising, but it agrees with the asymptotic properties discussed below \eqref{eq:66} in Section~\ref{sec:asymptotic-analysis}. For these two cases,  $\Exp(\z)=\bz$ and this is the reason why the convergence rate can be faster than that of the slope parameter estimators.

To see the effect of the subdata size $k$ for estimating the slope parameter $\bbeta_1$, Figure~\ref{fig:4}
presents plots of the log$_{10}$ of the MSEs against the subdata size $k$, with choices
$k=200$, $400$, $500$, $10^3$, $2\times10^3$, $3\times10^3$ and
$5\times10^3$, for fixed full data size $n=10^6$. The MSE for
using the full data is a constant with respect to $k$ and is plotted
for comparison. Clearly, all subdata-based methods improve as the subdata
size $k$ increases, with the D-OPT IBOSS method again being the best performer. For example, when $n=10^6$ for the mixture
covariate distribution, the analysis based on the D-OPT IBOSS method
with $k=200$ is about 10 times as accurate as that of the LEV method
with $k=5\times10^3$ as measured by the MSE value.
      
\begin{figure}
  \centering
  \begin{subfigure}{\textwidth}
    \includegraphics[width=0.49\textwidth]{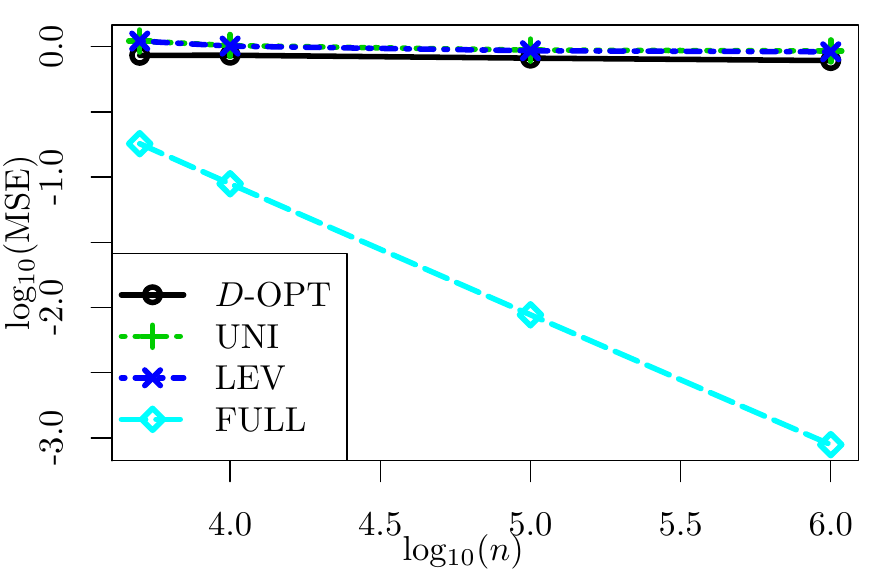}
    \includegraphics[width=0.49\textwidth]{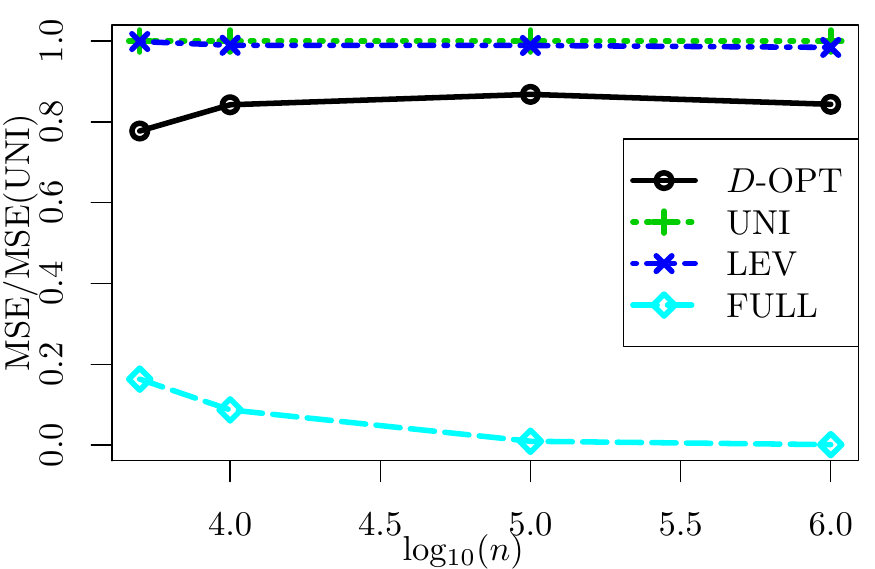}\\[-10mm]
    \caption{Case 1: $\z_i$'s are normal.}
  \end{subfigure}
  \\[3mm]
  \begin{subfigure}{0.49\textwidth}
    \includegraphics[width=\textwidth]{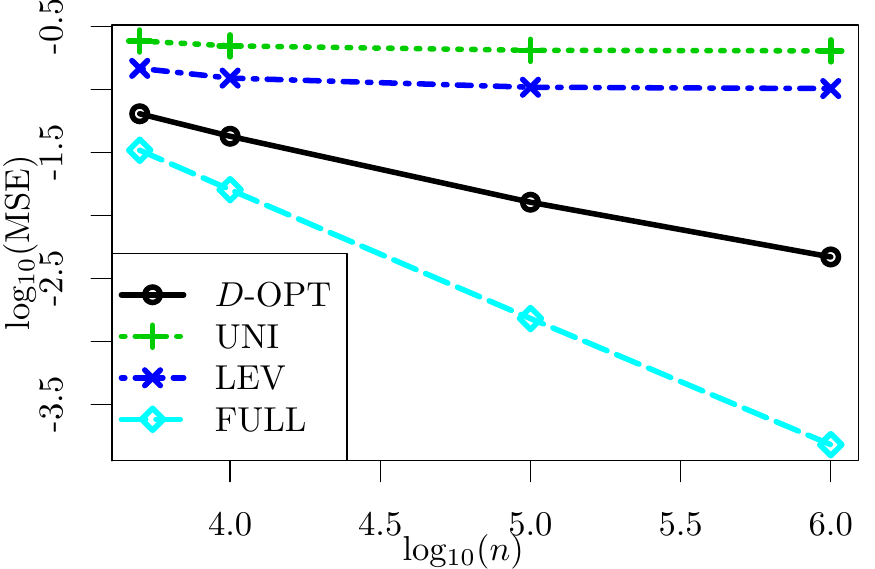}\\[-10mm]
    \caption{Case 2: $\z_i$'s are lognormal.}
  \end{subfigure}
  \begin{subfigure}{0.49\textwidth}
    \includegraphics[width=\textwidth]{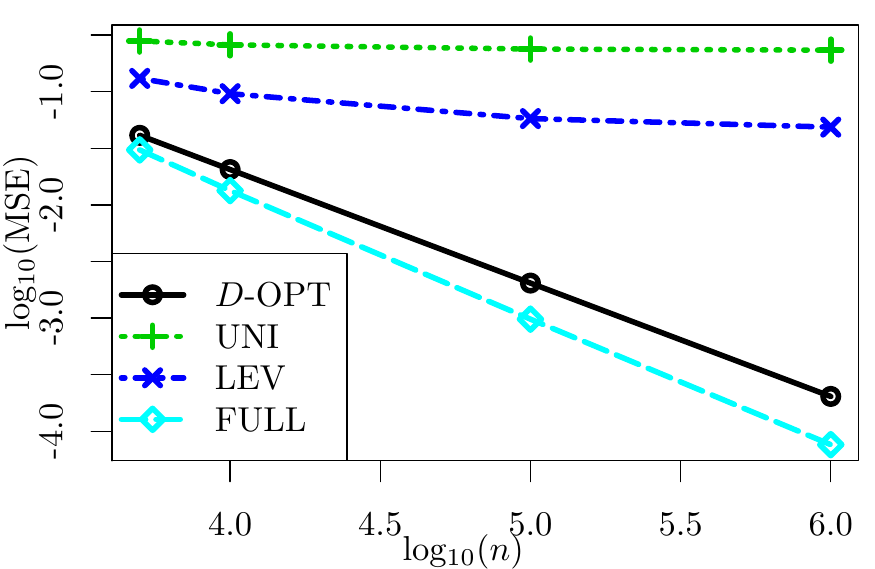}\\[-10mm]
    \caption{Case 3: $\z_i$'s are $t_2$.}
  \end{subfigure}\\[3mm]
  \begin{subfigure}{0.49\textwidth}
    \includegraphics[width=\textwidth]{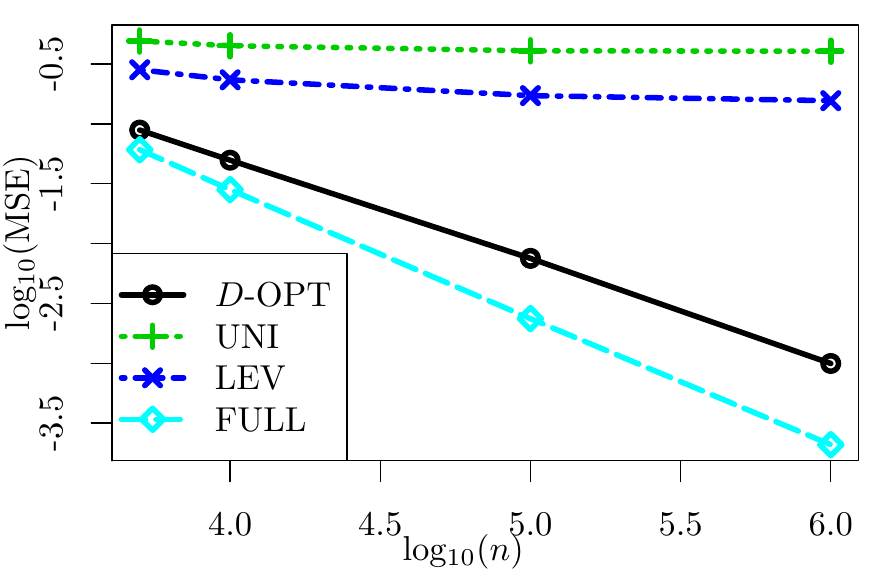}\\[-10mm]
    \caption{Case 4: $\z_i$'s are a mixture.}
  \end{subfigure}
  \begin{subfigure}{0.49\textwidth}
    \includegraphics[width=\textwidth]{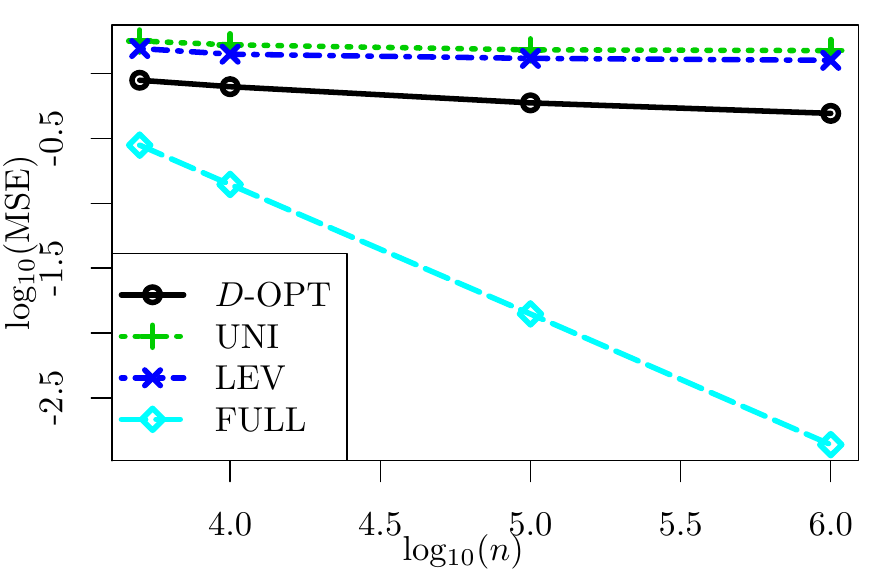}\\[-10mm]
    \caption{Case 5: $\z_i$'s include interaction terms.}
  \end{subfigure}
  \caption{MSEs for estimating the slope parameter for five
    different distributions for the covariates $\z_i$. The subdata
    size $k$ is fixed at $k=1000$ and the full data size $n$
    changes. Logarithm with base 10 is taken of $n$ and MSEs for
    better presentation of the figures except for the right panel of
    (a) in which MSEs are scaled so that MSEs for the UNI method are
    1.}
  \label{fig:2}
\end{figure}

\begin{figure}
  \centering
  \begin{subfigure}{\textwidth}
    \includegraphics[width=0.49\textwidth]{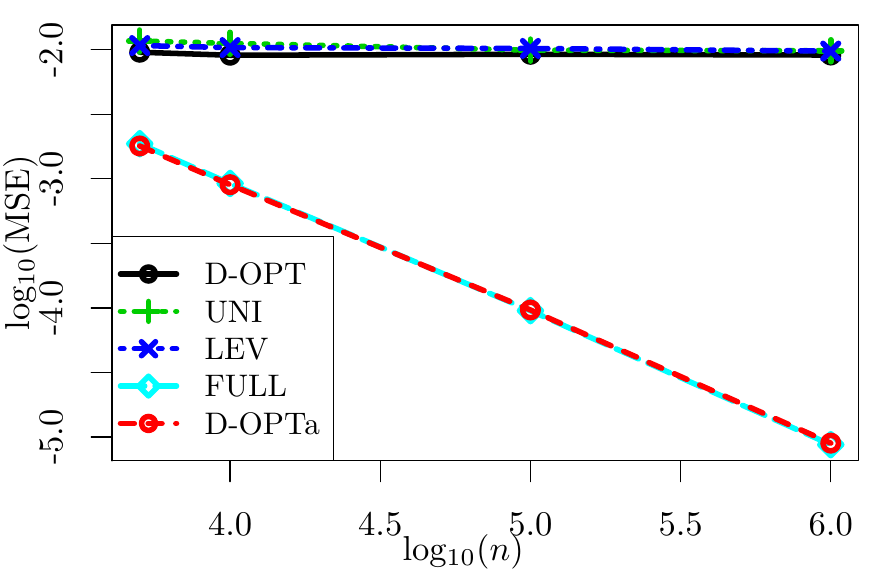}
    \includegraphics[width=0.49\textwidth]{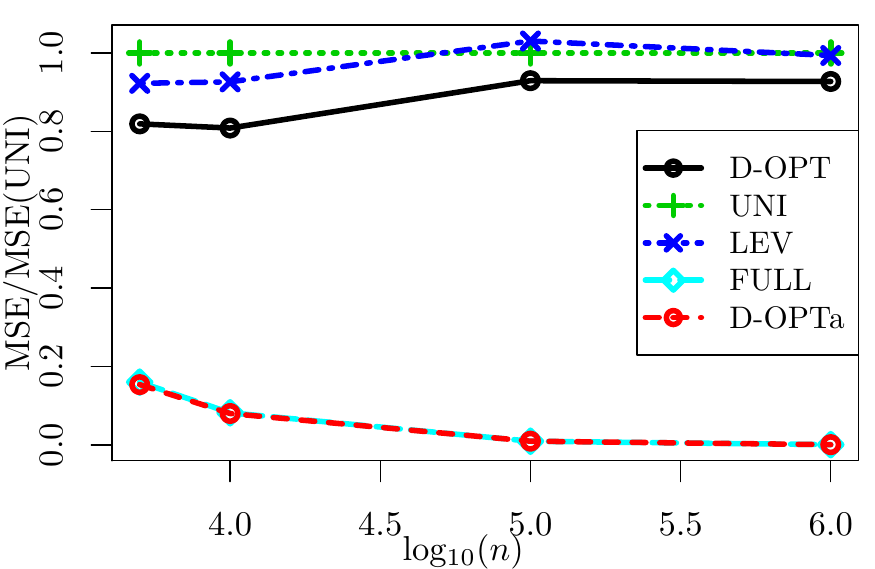}\\[-10mm]
    \caption{Case 1: $\z_i$'s are normal.}
  \end{subfigure}
  \\[3mm]
  \begin{subfigure}{0.49\textwidth}
    \includegraphics[width=\textwidth]{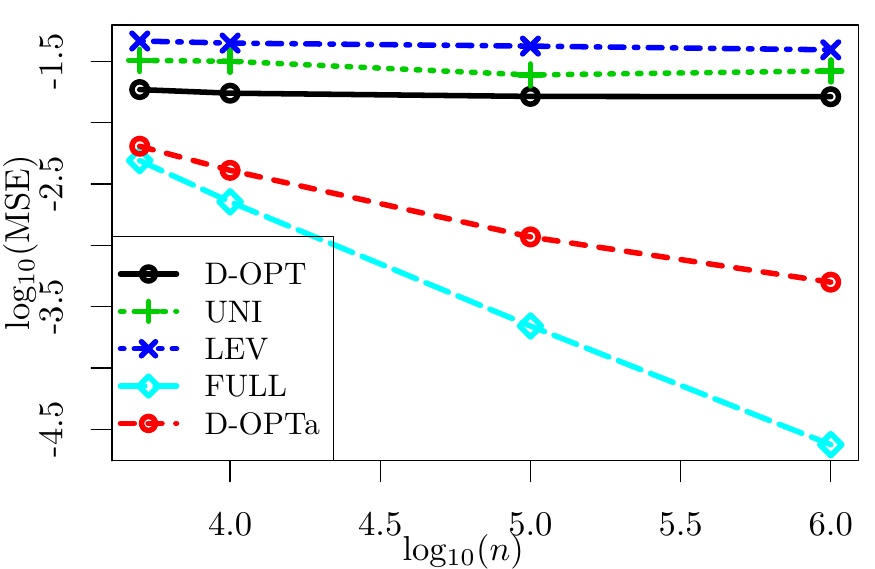}\\[-10mm]
    \caption{Case 2: $\z_i$'s are lognormal.}
  \end{subfigure}
  \begin{subfigure}{0.49\textwidth}
    \includegraphics[width=\textwidth]{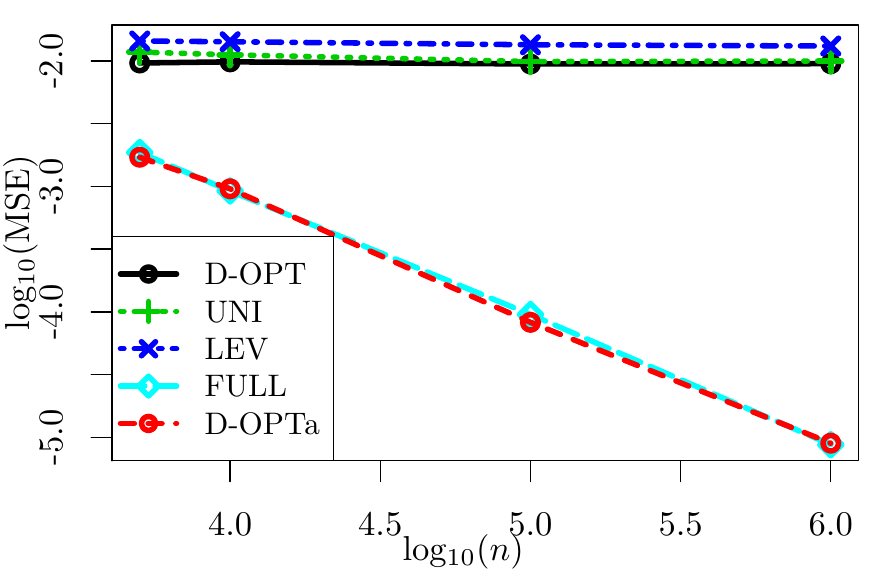}\\[-10mm]
    \caption{Case 3: $\z_i$'s are $t_2$.}
  \end{subfigure}\\[3mm]
  \begin{subfigure}{0.49\textwidth}
    \includegraphics[width=\textwidth]{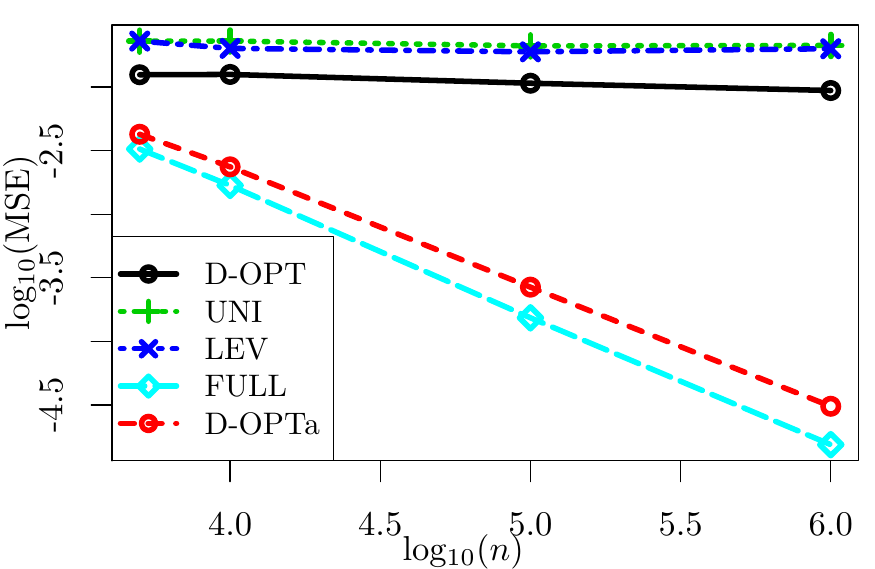}\\[-10mm]
    \caption{Case 4: $\z_i$'s are a mixture.}
  \end{subfigure}
  \begin{subfigure}{0.49\textwidth}
    \includegraphics[width=\textwidth]{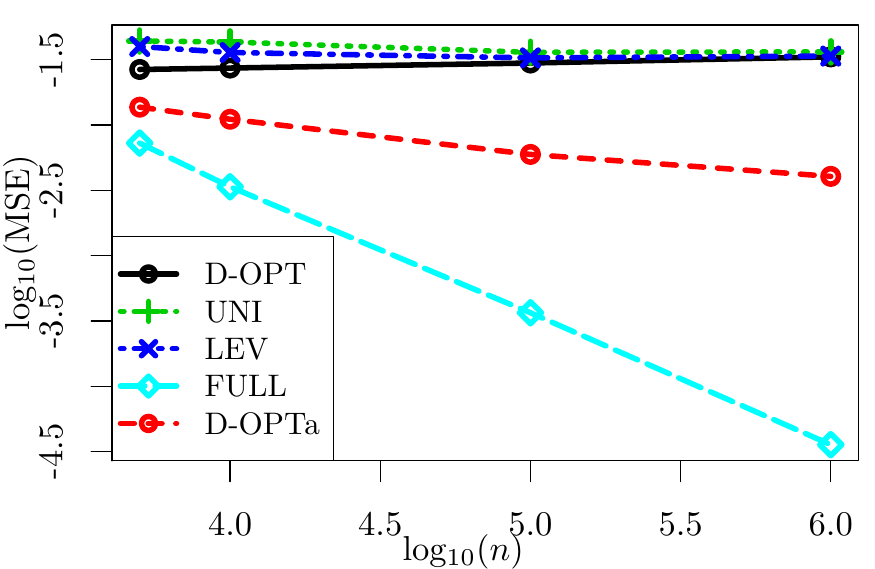}\\[-10mm]
    \caption{Case 5: $\z_i$'s include interaction terms.}
  \end{subfigure}
  \caption{MSEs for estimating the intercept parameter for five
    different distributions for the covariates $\z_i$. The subdata
    size $k$ is fixed at $k=1000$ and the full data size $n$
    changes. Logarithm with base 10 is taken of $n$ and MSEs for
    better presentation of the figures except for the right panel of
    (a) in which MSEs are scaled so that MSEs for the UNI method are
    1.}
  \label{fig:3}
\end{figure}

\begin{figure}
  \centering
  \begin{subfigure}{\textwidth}
    \includegraphics[width=0.49\textwidth]{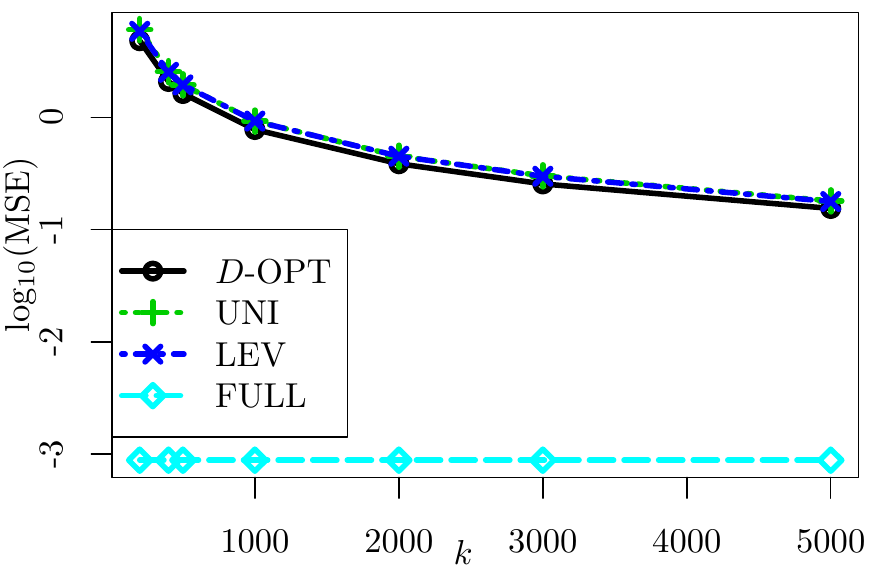}
    \includegraphics[width=0.49\textwidth]{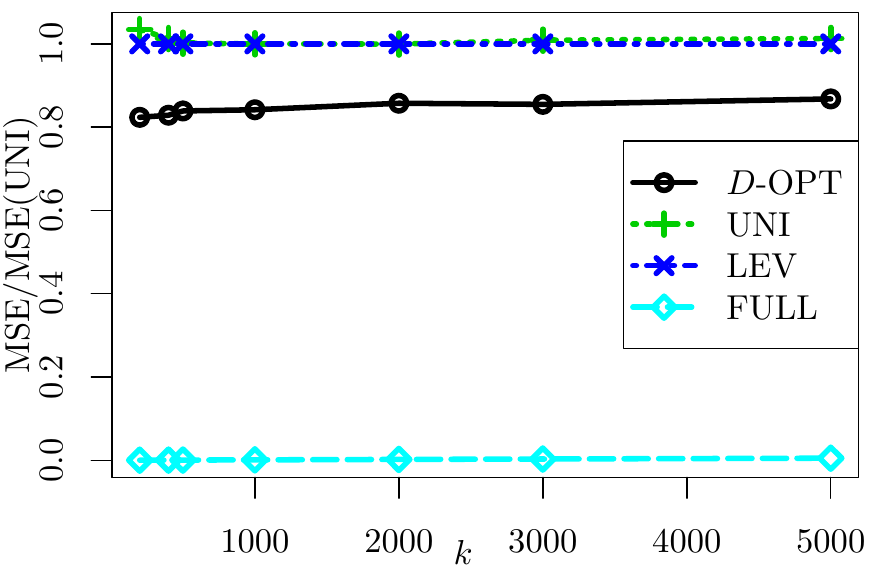}\\[-10mm]
    \caption{Case 1: $\z_i$'s are normal.}
  \end{subfigure}
  \\[3mm]
  \begin{subfigure}{0.49\textwidth}
    \includegraphics[width=\textwidth]{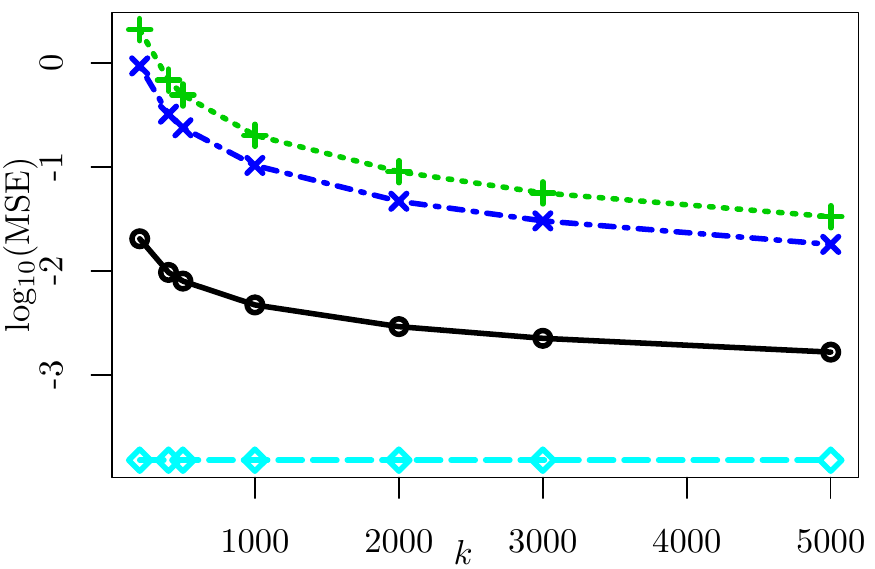}\\[-10mm]
    \caption{Case 2: $\z_i$'s are lognormal.}
  \end{subfigure}
  \begin{subfigure}{0.49\textwidth}
    \includegraphics[width=\textwidth]{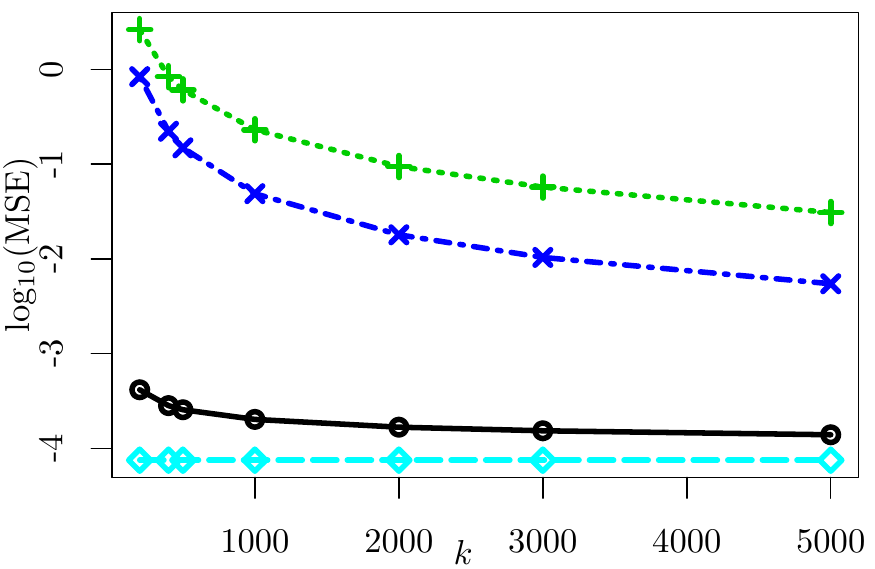}\\[-10mm]
    \caption{Case 3: $\z_i$'s are $t_2$.}
  \end{subfigure}\\[3mm]
  \begin{subfigure}{0.49\textwidth}
    \includegraphics[width=\textwidth]{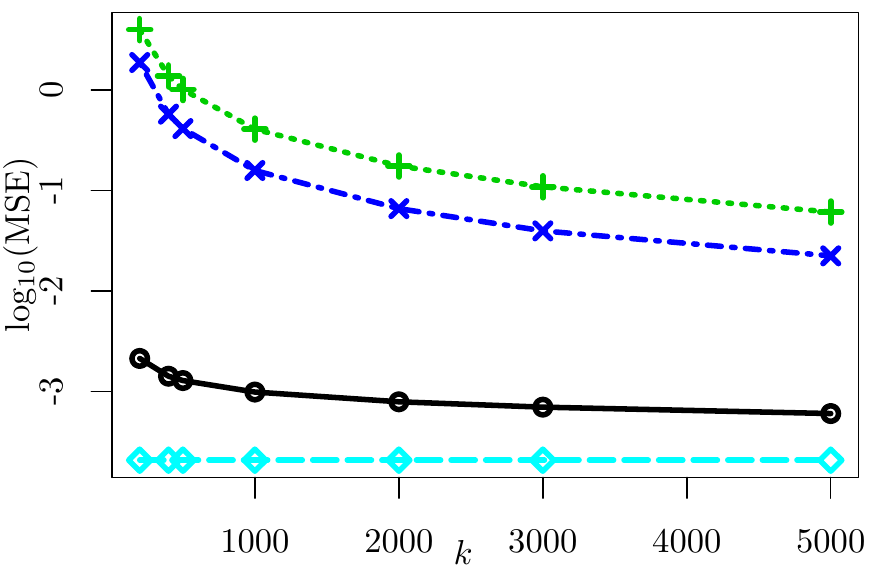}\\[-10mm]
    \caption{Case 4: $\z_i$'s are a mixture.}
  \end{subfigure}
  \begin{subfigure}{0.49\textwidth}
    \includegraphics[width=\textwidth]{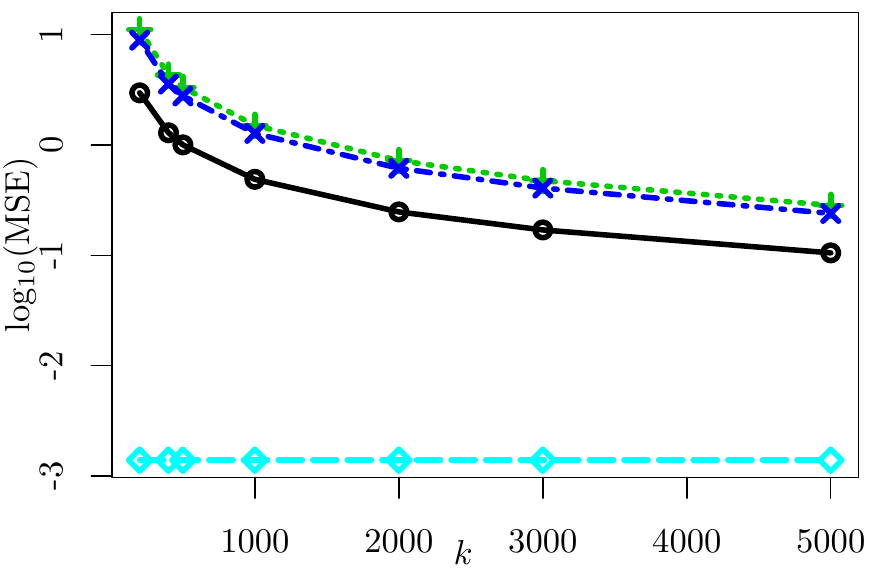}\\[-10mm]
    \caption{Case 5: $\z_i$'s include interaction terms.}
  \end{subfigure}
  \caption{MSEs for estimating the slope parameter for five
    different distributions for the covariates $\z_i$. The full data
    size is fixed at $n=10^6$ and the subdata size $k$
    changes. Logarithm with base 10 is taken of MSEs for
    better presentation of the figures except for the right panel of
    (a) in which MSEs are scaled so that MSEs for the UNI method are
    1.}
  \label{fig:4}
\end{figure}
To evaluate the performance of the D-OPT IBOSS approach for statistical
inference, we calculate the empirical coverage probabilities and
average lengths of the 95\% confidence intervals from this method. Results for the full data analysis are also computed for comparison. Figure~\ref{fig:5} gives results for the normal and mixture covariate distributions. The
estimated parameter is the first slope parameter $\beta_1$. Confidence intervals are constructed using
$\hat{\beta}_1^{(s)}\pm Z_{0.975}SE_1^{(s)}$, where
$\hat{\beta}_1^{(s)}$ and $SE_1^{(s)}$ are the estimate and its
standard error of $\beta_1$ in the $s$th repetition, and $Z_{0.975}$
is the 97.5th percentile of the standard normal distribution. It is seen
that all empirical coverage levels are close to the nominal level of
0.95, which shows that the inference based on IBOSS subdata is
valid. We do not compare this to subsampling-based approaches because we are not aware of theoretically justified methods for constructing
  confidence intervals under these approaches.

\begin{figure}
  \centering
  \begin{subfigure}{\textwidth}
    \includegraphics[width=\textwidth]{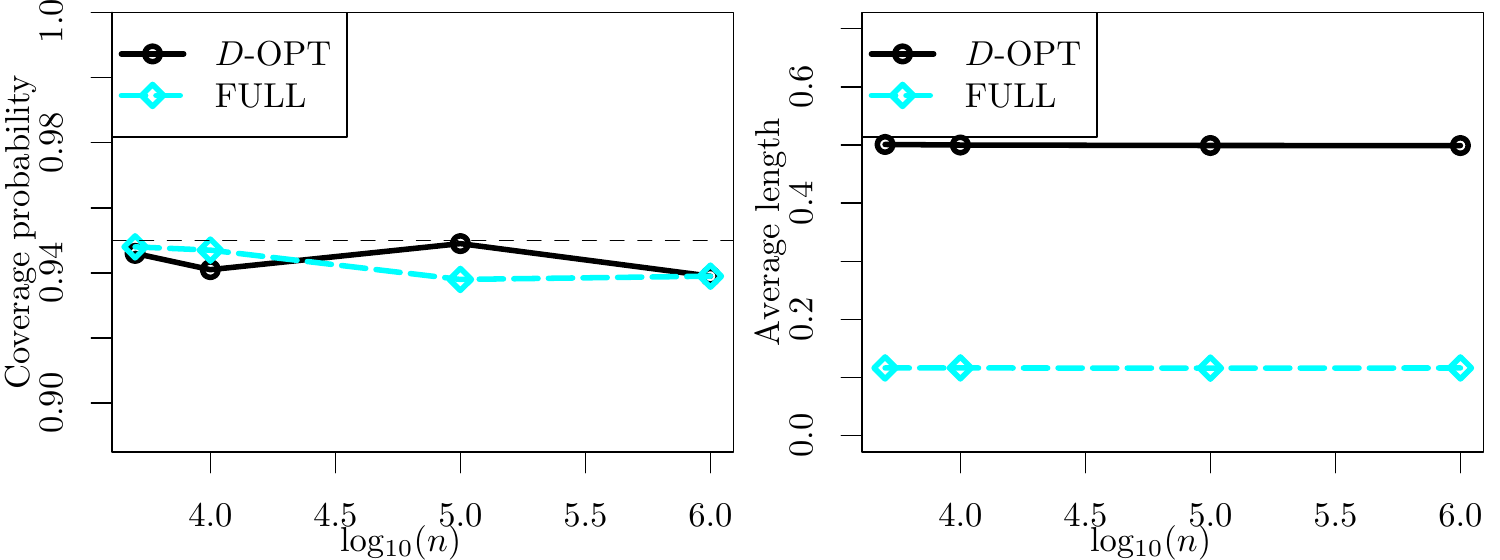}\\[-10mm]
    \caption{Case 1: $\z_i$'s are normal.}
  \end{subfigure}
  \begin{subfigure}{\textwidth}
    \includegraphics[width=\textwidth]{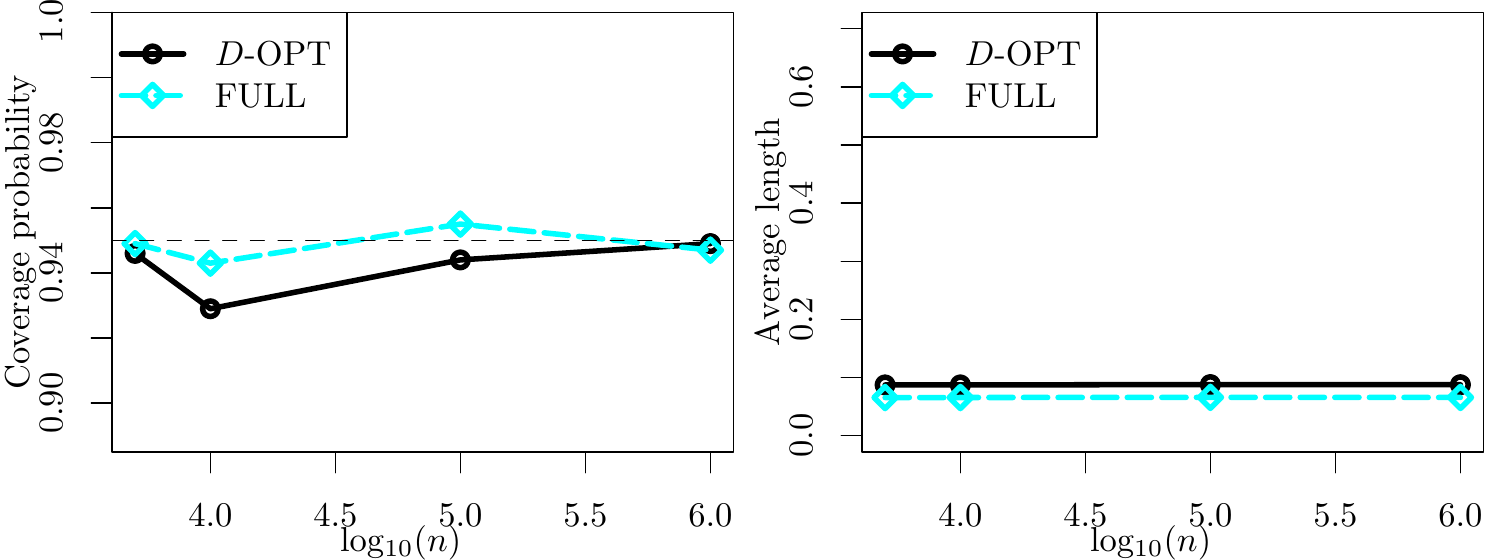}\\[-10mm]
    \caption{Case 4: $\z_i$'s are a mixture.}
  \end{subfigure}
  \caption{Empirical coverage probabilities and average lengths of
    $95\%$ confidence intervals from the D-OPT IBOSS method and the full
    data method. The gray horizontal dashed line in the left panel is
    the intended coverage probability 0.95. The subdata size is fixed
    at $k=10^3$.}
  \label{fig:5}
\end{figure}

Results on computational efficiency of the D-OPT IBOSS approach are presented in Table~\ref{tab:2}, which shows CPU times (in seconds) for different
combinations of the full data size $n$ and the number of covariates
$p$ for a fixed subdata size of $k=10^3$ and normal distribution for the $\z_i$'s. The R programming language \citep{R} is used
to implement each method. For the IBOSS approach, it requires a
partition-based partial sort algorithm which is not available in R, so
the standard C++ function {``\it nth\underline{ }element''} \citep{stroustrup1986c}, is called from R for
partial sorting. In order to get good performance in terms of CPU
times for the LEV method, the leverage scores are approximated using
the fast algorithm in \cite{Drineas:12}. The CPU times for using the
full data are also presented for comparison. All computations are
carried out on a desktop running Windows 10 with an Intel I7 processor
and 16GB memory.

It is seen from Table~\ref{tab:2} that the D-OPT IBOSS method compares favorably to the LEV method, both being more efficient than the full data method. Results for other cases are similar and thus are omitted.

\begin{table}
  \caption{CPU times for different combinations of $n$ and $p$ with a fixed $k=10^3$.}
  \label{tab:2}
  \centering
  \begin{subtable}{0.51\textwidth}
    \centering
    \caption{CPU times for different $n$ with $p=500$}
    \begin{tabular}{r|cccc}\hline
  $n$          & D-OPT   & UNI  & LEV   & FULL   \\ \hline
  $5\times10^3$ & 1.19    & 0.33 & 0.88  & 1.44   \\ 
  $5\times10^4$ & 1.36    & 0.29 & 2.20  & 13.39  \\ 
  $5\times10^5$ & 8.89    & 0.31 & 21.23 & 132.04 \\ \hline
    \end{tabular}
  \end{subtable}
  \begin{subtable}{0.51\textwidth}
    \centering
    \caption{CPU times for different $p$ with $n=5\times10^5$}
    \begin{tabular}{r|cccc}\hline
      $p$ & D-OPT & UNI & LEV   & FULL   \\ \hline
      10  & 0.19      & 0.00 & 1.94  & 0.21   \\ 
      100 & 1.74      & 0.02 & 4.66  & 6.55   \\ 
      500 & 9.30      & 0.31 & 21.94 & 132.47 \\ \hline
    \end{tabular}
  \end{subtable}
\end{table}

\subsection{Real data}\label{sec:real-data}
In this section, we evaluate the performance of the proposed IBOSS approach on two real data examples.

\subsubsection{Example 1: food intakes data}
The first example is a data set obtained from
the {\it Continuing Survey of Food Intakes by Individuals} (CSFII)
that was published by the Human Nutrition Research Center, U.S. Department of Agriculture, Beltsville, Maryland (CSFII Reports
No. 85-4 and No. 86-3). Part of the data set has been used in \citep{Thompson1992}. 
 It contains dietary intake and
related information for $n=1,827$ individuals, such as the intakes of calorie, fat, protein, and carbohydrate, as well as body mass index, age, etc. The size of this data 
set is not too big, and we can compare the IBOSS method to the full analysis. With this size of the data, we are also able to plot the full data in order to compare its pattern with that of the subdata selected by the IBOSS method. Interest is in examining the effects of the average intake levels of fat ($z_1$), protein ($z_2$), carbohydrate (carb, $z_3$), as well as body mass index
(BMI, $z_4$) and age ($z_5$) on calorie intake, $y$. 
Thus $p=5$. We fit the model
\begin{equation*}
  y=\beta_0+\beta_1z_1+\beta_2z_2+\beta_3z_3
  +\beta_4z_4+\beta_5z_5+\varepsilon,
\end{equation*}
using both the D-OPT IBOSS method with $k=10p=50$ and the full data. Results are summarized in Table~\ref{tab:3}. { The D-OPT IBOSS estimates for the slope parameters are not very different from those from the full data, and the signs of the estimates from the IBOSS method and from the full data are consistent. 
 The standard errors for the IBOSS method, while larger than for the full data, are reasonably good in view of the small subdata size.} The estimates for the intercept parameter show a larger difference, and the standard error for the IBOSS method is large. This agrees with the theoretical result that the intercept cannot be estimated precisely without a large subdata size. The D-OPT IBOSS method identifies the significant effects of fat, protein and carbohydrate intake levels on calorie intake. Based on the full data, the effect of BMI is near the boundary of significance at the 5\% level, and is not identified as significant by the D-OPT IBOSS method.

Figure~\ref{fig:6} gives scatter plots of calorie intake against each
covariate for the full data of $n=1,827$ with the D-OPT subdata of
$k=50$ marked. It is seen that the relationship between the response
and each covariate is similar for the subdata and the full data, especially for covariates fat, protein and carbohydrate. Also, there do not seem to be
  any extreme outliers in this data set.

\begin{table}
  \caption{Estimation results for the CSFII data. For the D-OPT IBOSS method, the subdata size is $k=10p=50$.}\label{tab:3}
  \centering
  \begin{tabular}{cccccccccc}
    \hline
    Parameter &  & \multicolumn{2}{c}{D-OPT} & & & \multicolumn{2}{c}{FULL} \\
    \cline{3-4}\cline{6-8}
              &  & Estimate & Std. Error &  &  & Estimate & Std. Error \\\hline
    Intercept &  & 33.545 & 46.833 &   &  & 45.489 & 11.883 \\
    Age       &  & -0.496 & 1.015  &   &  & -0.200 & 0.234  \\
    BMI       &  & -0.153 & 0.343  &   &  & -0.521 & 0.224  \\
    Fat       &  & 8.459  & 0.405  &   &  & 9.302  & 0.115  \\
    Protein   &  & 5.080  & 0.386  &   &  & 4.254  & 0.127  \\
    Carb      &  & 3.761  & 0.106  &   &  & 3.710  & 0.035  \\ \hline
  \end{tabular}
\end{table}

\begin{figure}
  \centering
    \includegraphics[width=\textwidth]{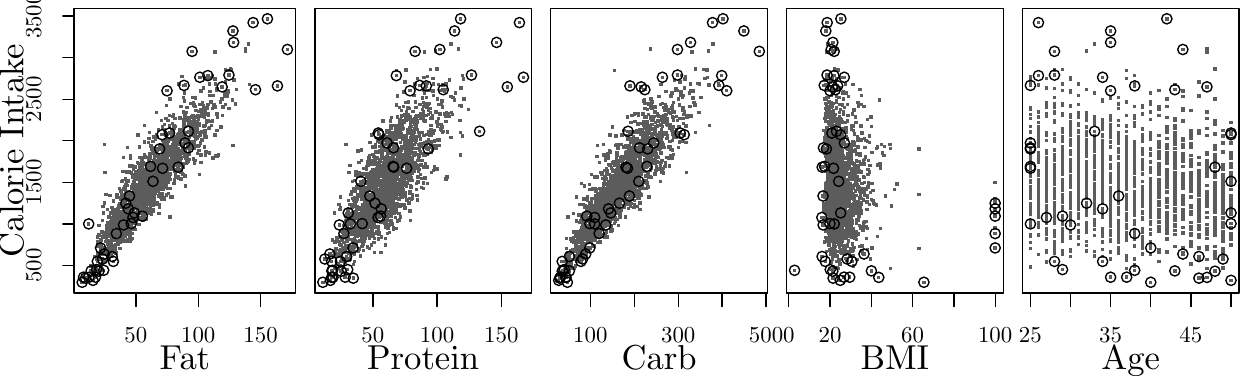}
  \caption{Scatter plots of calorie intake against each covariate for
    the full CSFII data (grey dots). The D-OPT subdata is labeled by
    {\large$\bf\circ$}.}
  \label{fig:6}
\end{figure}
To compare the IBOSS performance to that of the
subsampling approaches, we compute the MSE for the vector of slope parameters for each method by using one thousand bootstrap samples. 
Each bootstrap sample is a random sample of size $n$ from the full data using uniform sampling with replacement. For a bootstrap sample, we implement each subdata method to obtain the subdata estimate or implement the full data approach to obtain the full data estimate. The bootstrap MSEs are the empirical MSEs corresponding to the 1,000 estimates. 
We do this for $k=4p$, $6p$,
$10p$ and $20p$. Figure~\ref{fig:7} shows that
the D-OPT IBOSS method dominates random subsampling-based methods. 
The full data approach is shown for comparison.

\begin{figure}
  \centering
  \includegraphics[width=0.7\textwidth]{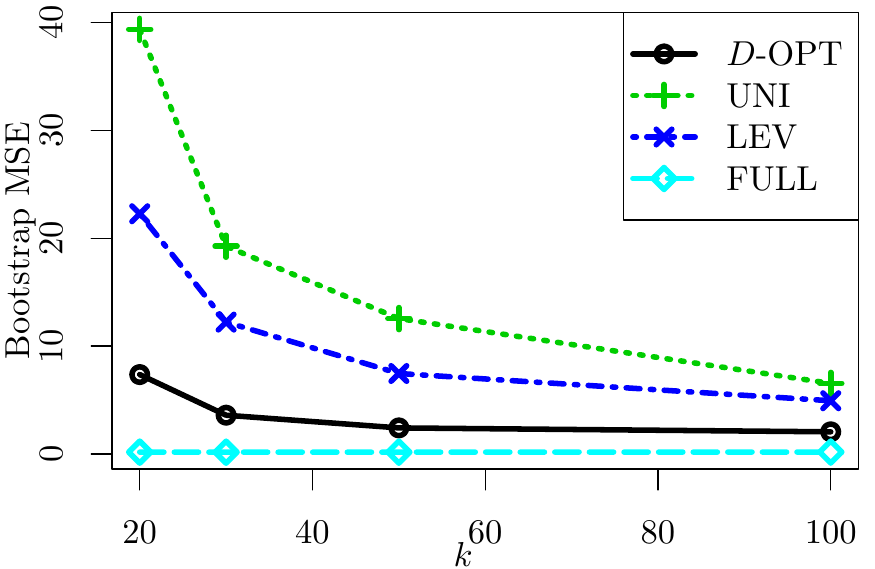}
  \caption{MSEs for estimating slope parameters
    for the CSFII data. They are computed from 1000
    bootstrap samples.}
    \label{fig:7}
\end{figure} 

\subsubsection{Example 2: chemical sensors data}
In this example, we consider chemical sensors data collected to develop and test strategies to solve a wide variety of tasks, e.g., to develop algorithms for continuously monitoring or improving response time of sensory systems \citep{fonollosa2015reservoir}. 
The data were collected at the ChemoSignals Laboratory in the BioCircuits Institute, University of California San Diego. It contains the readings of 16 chemical sensors exposed to the mixture of Ethylene and CO at varying concentrations in air. 
Each measurement was constructed by the continuous acquisition of the sixteen-sensor array signals for a duration of about 12 hours without interruption. 
The concentration transitions were set at random times and to random concentration levels. Further information about the data set can be found in \cite{fonollosa2015reservoir}.

For illustration, we use the
reading from the last sensor as the response and readings from
other sensors as covariates. 
Since trace concentrations often have a lognormal distribution \citep{goodson2011mathematical}, we take a log-transformation of the sensors readings. Readings from the second sensor are not used in the analysis because about 20\% of the values are negative for reasons unknown to us. Thus, there are $p=14$ covariates in this example. In addition, we exclude the first 20,000 data points corresponding to less than 4 minutes of system run-in time. Thus, the full data used contain $n=4,188,261$ data points.

Figure~\ref{fig:9} gives scatter plots of the response variable against each covariate for a simple random sample of size $10,000$, with D-OPT subdata of $k=280$ overlaid. Due to the size of the data, we cannot plot the full data in Figure~\ref{fig:9}. However, a simple random sample with a large sample size should be able to represent the overall pattern of the full data. It is seen that a linear model seems appropriate for the log-transformed readings and the relationship between the response and each covariate is similar for the subdata and the full data.
\begin{figure}
  \centering
    \includegraphics[width=\textwidth]{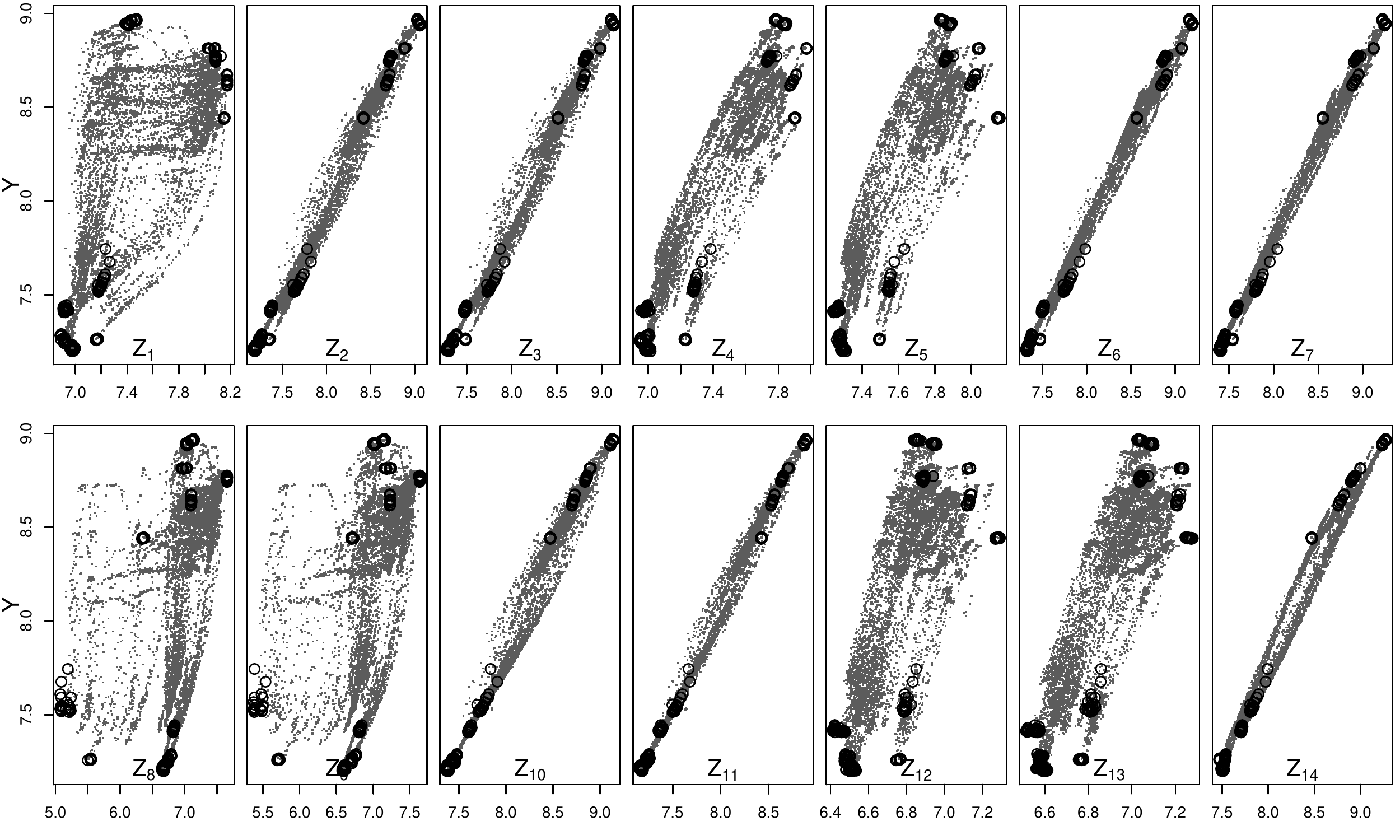}
  \caption{Scatter plots of a size $10,000$ simple random sample of the chemical sensors data (grey dots). The D-OPT subdata of size $k=280$ is plotted as {\large$\bf\circ$}. }
  \label{fig:9}
\end{figure}

We also use bootstrap to calculate the MSEs of different estimators for estimating the slope parameters. As for the first example, we considered $k=4p$, $6p$, $10p$ and $20p$ as subdata size for each method.  Results computed from $100$ bootstrap samples are plotted in Figure~\ref{fig:8}. The performance of the D-OPT IBOSS method lies between the full data approach and the subsampling based methods. 

\begin{figure}
  \centering
    \includegraphics[width=0.7\textwidth]{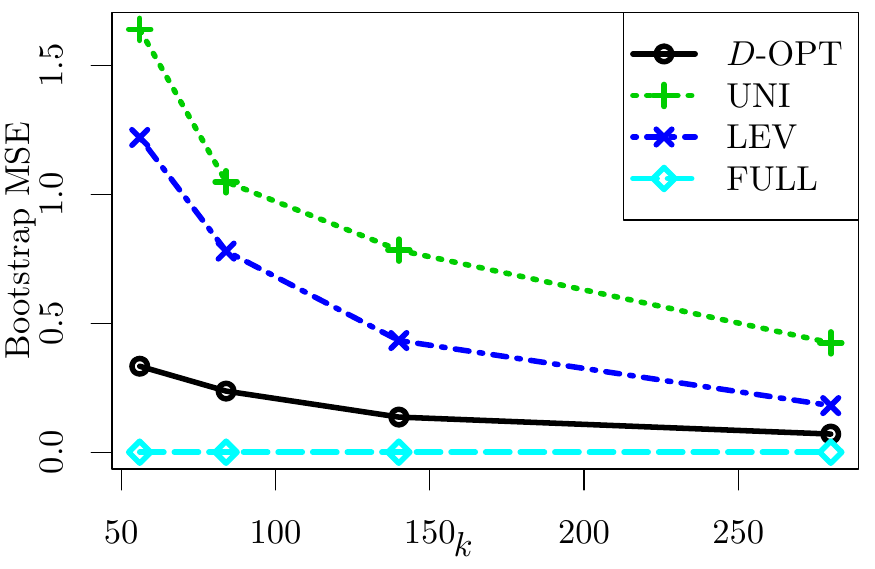}
    \caption{MSEs from 100 bootstrap samples for estimating slope parameters for the chemical sensors data.}
    \label{fig:8}
\end{figure}

\section{Concluding remarks}\label{sec:concluding-remarks}
In this paper, we have developed a subdata selection method, IBOSS, in the context of big data linear regression problems. Using the framework for the IBOSS method, we have analyzed existing subsampling-based methods and derived a lower bound for covariance matrices of the resultant estimators. For the IBOSS method, we focused on D-optimality. After a theoretical characterization of the IBOSS subdata under D-optimality, we developed a computationally efficient
algorithm to approximate the optimal subdata. Theoretical properties
of the D-OPT IBOSS method have been examined in detail through asymptotic analysis, and its performance has been demonstrated by using simulated and real data. 

There are important and unsolved questions that require future
study. 
For example, while we only considered the D-optimality criterion, there are other optimality criteria with meaningful statistical interpretations and different inferential purposes. This includes A-optimality, which seeks to minimize the average variance of estimators of regression
coefficients, and c-optimality, which minimizes the variance of the
best estimator of a pre-specified function of the model parameters. These optimality criteria may also be useful to develop efficient IBOSS methods. 

Identifying informative subdata is important for extracting useful information from big data and more research is needed. We hope that this work will stimulate additional research in the direction suggested in this paper.

\appendix
\section{Appendix}
\subsection{Proof of Theorem~\ref{thm:1}}
We will use the following convexity result \citep[cf.][]{Nordstrom2011} in the proof of Theorem~\ref{thm:1}.
\begin{lemma} \label{lemma1}
For any positive definite matrices $\B_1$ and $\B_2$ of the same dimension,
\begin{equation}\label{1:1}
  \left\{\alpha \B_1+ (1-\alpha)\B_2 \right\}^{-1}
  \leq \alpha \B_1^{-1}+(1-\alpha)\B_2^{-1} 
\end{equation}
in the Loewner ordering, where $0\leq \alpha \leq 1$.
\end{lemma}

\begin{proof}[\bf Proof of Theorem \ref{thm:1}]
  The unbiasedness can be verified by direct calculation,
  \begin{equation*}
    \Exp\{\tilde\bbeta_L|\Z, I_{\Delta}(\boldsymbol{\eta}_L)=1\}
    =\Exp_{\boldsymbol{\eta}_L}[\Exp_{\y}\{\tilde\bbeta_L|\Z, I_{\Delta}(\boldsymbol{\eta}_L)=1\}]
    =\Exp_{\boldsymbol{\eta}_L}(\bbeta)=\bbeta.
  \end{equation*}
Let $\W=\diag(w_1\eta_{L1}, ..., w_n\eta_{Ln})$. 
The variance-covariance matrix of the sampling-based estimators can be written as
\begin{align}
\Var\{\tilde\bbeta_L|\Z, I_{\Delta}(\boldsymbol{\eta}_L)=1\}
=&\Exp_{\boldsymbol{\eta}_L}[\Var_{\y}\{\tilde\bbeta_L|\Z, I_{\Delta}(\boldsymbol{\eta}_L)=1\}]
+\Var_{\boldsymbol{\eta}_L}[\Exp_{\y}\{\tilde\bbeta_L|\Z, I_{\Delta}(\boldsymbol{\eta}_L)=1\}]\notag\\
=&\sigma^2\Exp_{\boldsymbol{\eta}_L}\left\{\left(\X\tp\W\X\right)^{-1}
  \left(\X\tp\W^2\X\right)
  \left(\X\tp\W\X\right)^{-1}\right\}
+\Var_{\boldsymbol{\eta}_L}(\bbeta)\notag\\
=&\sigma^2\Exp_{\boldsymbol{\eta}_L}\left[\left\{\left(\X\tp\W\X\right)
  \left(\X\tp\W^2\X\right)^{-1}
  \left(\X\tp\W\X\right)\right\}^{-1}\right]\notag\\
\ge&\sigma^2\left[\Exp_{\boldsymbol{\eta}_L}\left\{\left(\X\tp\W\X\right)
  \left(\X\tp\W^2\X\right)^{-1}
  \left(\X\tp\W\X\right)\right\}\right]^{-1}.\label{var1}
\end{align}
The last inequality is due to Lemma 1. Notice that $\W\X\left(\X\tp \W^2\X\right)^{-1}\X\tp \W=\text{pr}(\W\X)$, the orthogonal projection matrix onto the column space of $\W\X$. 
 Define 
\begin{equation*}
\B_{WX}=
\begin{bmatrix}
w_1\eta_{L1}\x_1\tp& & \\
& \ddots &\\
& & w_n\eta_{Ln}\x_n\tp
\end{bmatrix}.
\end{equation*}
Notice that the column-space of
$\W\X=(w_1\eta_{L1}\x_1, ..., w_n\eta_{Ln}\x_n)\tp$ 
is contained in the column-space of $\B_{WX}$. Hence 
we have $\text{pr}(\W\X)\leq \text{pr}(\B_{WX})$ in the Loewner ordering, i.e.,
\begin{equation*}
\W\X\left(\X\tp \W^2\X\right)^{-1}\X\tp \W
\leq
\begin{bmatrix}
\x_1\tp\left(\x_1\x_1\tp\right)^{-}\x_1I(\eta_{L1}>0)& & \\
& \ddots &\\
& & \x_n\tp\left(\x_n\x_n\tp\right)^{-}\x_nI(\eta_{Ln}>0)
\end{bmatrix}.
\end{equation*}
where $I()$ is the indicator function. From this result, it can be shown that 
\begin{equation}\label{var2}
\X\tp\W\X\left(\X\tp\W^2\X\right)^{-1}\X\tp\W\X
\leq\sumn\x_i\x_i\tp I(\eta_{Li}>0). 
\end{equation}
For sampling with replacement,
\begin{equation*}
  P(\eta_{Li}>0|\Z)=1-(1-\pi_i)^k=\pi_i\sumk(1-\pi_i)^{i-1}\leq k\pi_i.
\end{equation*}
For sampling without replacement,
\begin{equation*}
  P(\eta_{Li}>0|\Z)=P(\eta_{Li}=1|\Z)=k\pi_i.
\end{equation*}
Thus, in either case, $P(\eta_{Li}>0|\Z)\leq k\pi_i$. Therefore,
\begin{align}
P\{\eta_{Li}>0|\Z, I_{\Delta}(\boldsymbol{\eta}_L)=1\}
=\frac{P\{\eta_{Li}>0, I_{\Delta}(\boldsymbol{\eta}_L)=1|\Z\}}{P\{I_{\Delta}(\boldsymbol{\eta}_L)=1|\Z\}}
\le\frac{P(\eta_{Li}>0|\Z)}{P\{I_{\Delta}(\boldsymbol{\eta}_L)=1|\Z\}}
\le\frac{k\pi_i}{P\{I_{\Delta}(\boldsymbol{\eta}_L)=1|\Z\}}.
  \label{lemma2:1}
\end{align}

Combining (\ref{var1}), (\ref{var2}) and (\ref{lemma2:1}), we have
\begin{align*}
\Var\{\tilde\bbeta_L|\Z, I_{\Delta}(\boldsymbol{\eta}_L)=1\}
&\ge\sigma^2\left[\Exp_{\boldsymbol{\eta}_L}\left\{
\sumn\x_i\x_i\tp I(\eta_{Li}>0)
\right\}\right]^{-1}\\
&=\sigma^2\left[\sumn\x_i\x_i\tp
P\{\eta_{Li}>0|\Z, I_{\Delta}(\boldsymbol{\eta}_L)=1\}\right]^{-1}\\
&\ge\frac{\sigma^2P\{I_{\Delta}(\boldsymbol{\eta}_L)=1|\Z\}}{k}
\left\{\sumn\pi_{i}\x_i\x_i\tp\right\}^{-1}.
\end{align*}
\end{proof}

\subsection{Proof of Theorem~\ref{thm:2}}
\begin{proof}
  Let
  $\breve{z}_{ij}=\{2z_{ij}-(z_{(n)j}+z_{(1)j})\}/(z_{(n)j}-z_{(1)j})$. Then
  we have,
  \begin{align}\label{eq:4}
    \sumn\delta_i\x_i\x_i\tp=k\B_3^{-1}\breve{\M}(\bdelta)(\B_3\tp)^{-1},
  \end{align}
  where
  \begin{equation*}
    \breve{\M}(\bdelta)=
    \begin{bmatrix}
      1  & k^{-1}\sumn\delta_i\breve{z}_{i1}  &\ldots &k^{-1}\sumn\delta_i\breve{z}_{id} \\
      k^{-1}\sumn\delta_i\breve{z}_{i1} & k^{-1}\sumn\delta_i\breve{z}^2_{i1}  & \ldots &k^{-1}\sumn\delta_i\breve{z}_{i1}\breve{z}_{ip} \\
      \vdots &\vdots & \ddots & \vdots \\
      k^{-1}\sumn\delta_i\breve{z}_{ip} &
      k^{-1}\sumn\delta_i\breve{z}_{i1}\breve{z}_{ip} & \ldots &
      k^{-1}\sumn\delta_i\breve{z}^2_{ip}
    \end{bmatrix},
  \end{equation*}
  and
  \begin{equation}\label{eq:38}
    \B_3=\begin{bmatrix}
      1 &  &  &  \\
      -\frac{z_{(n)1}+z_{(1)1}}{z_{(n)1}-z_{(1)1}}
      & \frac{2}{z_{(n)1}-z_{(1)1}} &  &  \\
      \vdots &  & \ddots &  \\
      -\frac{z_{(n)p}+z_{(1)p}}{z_{(n)p}-z_{(1)p}}
      &  &  & \frac{2}{z_{(n)p}-z_{(1)p}} \\
    \end{bmatrix}
  \end{equation}
  Note that $\breve{z}_{ij}\in[-1,1]$ for all
  $i=1, ..., n$ and $j=1, ..., p$, which implies
  $k^{-1}\sumn\delta_i\breve{z}_{ij}^2\le1$ for all $1\le j\le p$. Thus,
  \begin{align}\label{eq:3}
    |\breve{\M}(\bdelta)|=\prod_{j=0}^{p}\lambda_j
    \le\left(\frac{\sum_{j=0}^{p}\lambda_j}{p+1}\right)^{p+1}
    =\left(\frac{1+\sump k^{-1}\sumn\delta_i\breve{z}_{ij}^2}{p+1}\right)^{p+1}
    \le1,
  \end{align}
  where $\lambda_j$, $j=0,1, ..., p$ are eigenvalues of
  $\breve{\M}(\bdelta)$.  
  From \eqref{eq:4}, \eqref{eq:38} and \eqref{eq:3},
  \begin{align*}
    \left|\sumn\delta_i\x_i\x_i\tp\right|
    =k^{p+1}|\B_3|^{-2}|\breve{\M}(\bdelta)|
    \le k^{p+1}\left|\prod_{j=1}^p\frac{2}{z_{(n)j}-z_{(1)j}}\right|^{-2}
    =\frac{k^{p+1}}{4^p}\prod_{j=1}^p(z_{(n)j}-z_{(1)j})^2.
  \end{align*}
    If the subdata consists of the $2^p$ points
    $(a_{1},\ldots,a_{p})\tp$ where $a_{j}=z_{(n)j}$ or $z_{(1)j}$,
    $j=1, 2, ..., p$, each occurring equally often, then the
  $\bdelta^{opt}$ corresponding to this subdata satisfies 
  $\breve{\M}(\bdelta)=\mathbf{I}$.
  This $\bdelta^{opt}$ attains equality in
  \eqref{eq:3} and corresponds therefore to D-optimal subdata.
\end{proof}

\subsection{Proof of Theorem~\ref{thm:3}}
\begin{proof} 
As before, for $i=1,...,n$, $j=1,...,p$, let $z_{(i)j}$ be the $i$th order statistic for $z_{1j}, ..., z_{nj}$. For $l\neq j$, let $z_j^{(i)l}$ be the concomitant of $z_{(i)l}$ for $z_j$, i.e., if
  $z_{(i)l}=z_{sl}$ then $z_j^{(i)l}=z_{sj}$, $i=1, ..., n$. For the subdata obtained from Algorithm~\ref{alg:1}, let $\bar{z}_{j}^*$ and $\vr(z_{j}^*)$ be the sample mean and sample variance for covariate $z_j$. { From Algorithm~\ref{alg:1}, the values $z_j$, $j=1,...,p$, in the subdata consist of $z_{(m)j}$, and
$z_j^{(m)l}$, $l=1,...j-1,j+1,...,p$, $m=1,...,r$, $n-r+1,...,n$. Note that the subdata may not contain exactly the $r$ smallest and $r$ largest values for each covariate since some data points may be removed in processing each covariate. However, since $r$ is fixed when $n$ goes to infinity, this will not affect the final result. Therefore, for easy of presentation, we abuse the notation and write the range of values of $m$ as $1,...,r$, $n-r+1,...,n$.} The information matrix based on the subdata can be written as
\begin{align} \label{eq:6}
 (\X^*_{\mathrm{D}})\tp\X^*_{\mathrm{D}}=\B_4^{-1}
  \begin{bmatrix} k  & \bz\tp\\ \bz & (k-1)\mathbf{R} \end{bmatrix}
(\B_4\tp)^{-1},
\end{align} 
where
\begin{align}\label{eq:7}
  \B_4=
  \begin{bmatrix}
    1 &  &  &  \\
    -\frac{\bar{z}_{1}^*}{\sqrt{\vr(z_{1}^*)}} &
    \frac{1}{\sqrt{\vr(z_{1}^*)}} &  & \\
    \vdots &  & \ddots &  \\
    -\frac{\bar{z}_{p}^*}{\sqrt{\vr(z_{p}^*)}}
    &  &  & \frac{1}{\sqrt{\vr(z_{p}^*)}} \\
  \end{bmatrix}.
\end{align}
From \eqref{eq:6} and \eqref{eq:7}, 
\begin{align}\label{eq:34}
  |(\X^*_{\mathrm{D}})\tp\X^*_{\mathrm{D}}|=k|(k-1)\mathbf{R}|\prod_{j=1}^p\vr(z_{j}^*)
  \ge k(k-1)^p\lambda_{\min}^p(\mathbf{R})\prod_{j=1}^p\vr(z_{j}^*).
\end{align} 
For each sample variance,
{
\begin{align}
  (k-1)\vr(z_{j}^*)
  =&\sumk\left(z_{ij}^*-\bar{z}_j^*\right)^2\notag\\
  =&\left(\sumr+\sumrn\right)
       \left(z_{(i)j}-\bar{z}_j^*\right)^2\notag
  +\sum_{l\neq j}\left(\sumr+\sumrn\right)\left(z_j^{(i)l}-\bar{z}_j^*\right)^2\\
  \ge&\left(\sumr+\sumrn\right)
       \left(z_{(i)j}-\bar{z}_j^{**}\right)^2\notag\\
  =&\sumr\left(z_{(i)j}-\bar{z}_j^{*l}\right)^2
     +\sumrn\left(z_{(i)j}-\bar{z}_j^{*u}\right)^2
     +\frac{r}{2}\left(\bar{z}_j^{*u}-\bar{z}_j^{*l}\right)^2\notag\\
  \ge&\frac{r}{2}\left(\bar{z}_j^{*u}-\bar{z}_j^{*l}\right)^2\notag\\
  \ge&\frac{r}{2}\left(z_{(n-r+1)j}-z_{(r)j}\right)^2\label{eq:70}
\end{align}
where $\bar{z}_j^{**}=\left(\sumr+\sumrn\right)z_{(i)j}/(2r)$, $\bar{z}_j^{*l}=\sumr z_{(i)j}/r$, and $\bar{z}_j^{*u}=\sumrn z_{(i)j}/r$.
From~\eqref{eq:70},
\begin{align}
  \vr(z_{j}^*)
  \ge&\frac{r(z_{(n)j}-z_{(1)j})^2}{2(k-1)}
       \left(\frac{z_{(n-r+1)j}-z_{(r)j}}{z_{(n)j}-z_{(1)j}}\right)^2  .
       \label{eq:71}
\end{align}
Thus,
\begin{align*}
  |(\X^*_{\mathrm{D}})\tp\X^*_{\mathrm{D}}|
  \ge&k(k-1)^p\lambda_{\min}^p(\mathbf{R})\prod_{j=1}^p
       \frac{r(z_{(n)j}-z_{(1)j})^2}{2(k-1)}
     \left(\frac{z_{(n-r+1)j}-z_{(r)j}}{z_{(n)j}-z_{(1)j}}\right)^2\\
  =&\frac{r^p}{2^p}
     k\lambda_{\min}^p(\mathbf{R})\prod_{j=1}^p(z_{(n)j}-z_{(1)j})^2
     \times\prod_{j=1}^p
     \left(\frac{z_{(n-r+1)j}-z_{(r)j}}{z_{(n)j}-z_{(1)j}}\right)^2.
\end{align*}
This shows that
\begin{align*}
  \frac{|(\X^*_{\mathrm{D}})\tp\X^*_{\mathrm{D}}|}
    {\frac{k^{p+1}}{4^p}\prod_{j=1}^p(z_{(n)j}-z_{(1)j})^2}
  \ge&\frac{\lambda_{\min}^p(\mathbf{R})}{p^p}
  \times\prod_{j=1}^p
     \left(\frac{z_{(n-r+1)j}-z_{(r)j}}{z_{(n)j}-z_{(1)j}}\right)^2.
\end{align*}
}
\end{proof}

\subsection{Proof of Theorem~\ref{thm:4}}
\begin{proof}
  From \eqref{eq:6} and \eqref{eq:7}, 
\begin{align*}
  \Var(\hat\bbeta^{\mathrm{D}}|\Z)
  &=\sigma^2\{(\X^*_{\mathrm{D}})\tp\X^*_{\mathrm{D}}\}^{-1}
  =\sigma^2\B_4\tp
  \begin{bmatrix} \onek  & \bz\tp\\ \bz & \frac{1}{k-1}\mathbf{R}^{-1} \end{bmatrix}\B_4.
\end{align*}
Thus
\begin{equation}\label{eq:59}
  \Var(\hat\beta^{\mathrm{D}}_0|\Z)
  =\sigma^2\left(\frac{1}{k}+\frac{1}{k-1}\mathbf{u}\tp\mathbf{R}^{-1}\mathbf{u}\right),
\end{equation}
and 
\begin{equation}\label{eq:62}
  \Var(\hat\beta^{\mathrm{D}}_j|\Z)
  =\frac{\sigma^2}{k-1}\frac{(\mathbf{R}^{-1})_{jj}}{\vr(z_{j}^*)},
\end{equation}
where $\mathbf{u}=\Big\{-\bar{z}_{1}^*/\sqrt{\vr(z_{1}^*)}, ..., -\bar{z}_{p}^*/\sqrt{\vr(z_{p}^*)}\Big\}\tp$ and $(\mathbf{R}^{-1})_{jj}$ is the $j$th diagonal element of $\mathbf{R}^{-1}$.

{
From \eqref{eq:59}, $\Var(\hat\beta^{\mathrm{D}}_0|\Z)\ge\sigma^2/k$ because $\mathbf{u}\tp\mathbf{R}^{-1}\mathbf{u}\ge0$.

Denote the spectral decomposition of $\mathbf{R}$ as $\mathbf{R}=\mathbf{V}\bLambda\mathbf{V}\tp$. Since $\bLambda^{-1}\le\lambda_{\min}^{-1}(\mathbf{R})\I_p$, $\mathbf{R}^{-1}=\mathbf{V}\bLambda^{-1}\mathbf{V}\tp
\le\mathbf{V}\lambda_{\min}^{-1}(\mathbf{R})\I_p\mathbf{V}\tp
=\lambda_{\min}^{-1}(\mathbf{R})\I_p\tp$. Thus $\mathbf{R}^{-1}_{jj}\le\lambda_{\min}^{-1}(\mathbf{R})$ for all $j$. 
 From this fact, and \eqref{eq:62} and \eqref{eq:71}, we have
\begin{align}
  \Var(\hat\beta^{\mathrm{D}}_j|\Z)
  &=\frac{\sigma^2}{k-1}\frac{(\mathbf{R}^{-1})_{jj}}{\vr(z_{j}^*)}
  \le\frac{4p\sigma^2}
    {k\lambda_{\min}(\mathbf{R})(z_{(n-r+1)j}-z_{(r)j})^2}.
\end{align}
Similarly, we have 
\begin{align}
  \Var(\hat\beta^{\mathrm{D}}_j|\Z)
  &=\frac{\sigma^2}{k-1}\frac{(\mathbf{R}^{-1})_{jj}}{\vr(z_{j}^*)}
  \ge\frac{4\sigma^2}
    {k\lambda_{\max}(\mathbf{R})(z_{(n)j}-z_{(1)j})^2}.
\end{align}
Here we utilize the following inequality
\begin{align}
  \vr(z_{j}^*)\leq&\frac{1}{k-1}\sum_{i=1}^k\left(z_{ij}^*-\frac{z_{(n)j}+z_{(1)j}}{2}\right)^2
  \le\frac{k}{4(k-1)}\left(z_{(n)j}-z_{(1)j}\right)^2,
\end{align}
where the last inequality is due to the fact $|z_{ij}^*-\frac{z_{(n)j}+z_{(1)j}}{2}|\leq \frac{z_{(n)j}-z_{(1)j}}{2}$ for all $i=1,\ldots, k$.}

\end{proof}

{
\subsection{Proof of Theorem~\ref{thm:4-2}}
\begin{proof}
  For \eqref{eq:9}, it is a direct result from \eqref{eq:64}. 

    For \eqref{eq:18}, we consider the five cases in the following. For the first case that $r$ is fixed, from results in Theorems 2.8.1 and 2.8.2 in \cite{galambos1987asymptotic}, we have that
  \begin{equation}\label{eq:69}
    \frac{z_{(n-r+1)j}-z_{(r)j}}{z_{(n)j}-z_{(1)j}}=\Op
    \quad\text{ and }\quad
    \frac{z_{(n)j}-z_{(1)j}}{z_{(n-r+1)j}-z_{(r)j}}=\Op.
  \end{equation}
  Combining \eqref{eq:9} and \eqref{eq:69}, \eqref{eq:18} follows.

For the second case when $r\rightarrow \infty$, $r/n\rightarrow0$, and the support of $F_j$ is bounded, \eqref{eq:69} can be easily verified.

For the third case when the upper endpoint for the support of $F_j$ is $\infty$ and the lower endpoint for the support of $F_j$ is finite, and $r\rightarrow\infty$ slow enough such that \eqref{eq:8} holds, if we can show that $z_{(n-r+1)j}/z_{(n)j}=1+\op$, then the result in \eqref{eq:18} follows. Let $b_{n,j}=F_j^{-1}(1-n^{-1})$. From \cite{hall1979relative}, we only need to show that $z_{(n-r+1)j}/b_{n,j}=1+\op$ in order to show that $z_{(n-r+1)j}/z_{(n)j}=1+\op$. For this, from the proof of Theorem 1 of \cite{hall1979relative}, it suffices to show that 
\begin{equation*}
  \left[\frac{1-F_j(b_{n,j})}{1-F_j\{(1-\epsilon)b_{n,j}\}}\right]^{-1/2}
  \left[1-\frac{r\{1-F_j(b_{n,j})\}}{1-F_j\{(1-\epsilon)b_{n,j}\}}\right]
  \rightarrow\infty,
\end{equation*}
which holds by directly applying the assumption in \eqref{eq:8} and the fact that $r\rightarrow\infty$. 

For the fourth case, it can be proved by using an approach similar to the one used for the third case. It can also be proved by noting that $z_{(r)j}=-(-z)_{(n-r+1)j}$, $z_{(1)j}=-(-z)_{(n)j}$, and the fact that the condition in \eqref{eq:33} on $\z$ becomes the condition in \eqref{eq:8} on $-\z$.

For the fifth case, it can be proved by combining the results in the third case and the fourth case.
\end{proof}}

\subsection{Proof of Theorem~\ref{thm:5}}
Let $\sigma_j$ and $\rho_{j_1j_2}$ be the $j$th diagonal element of $\bsigma$ and
entry $(j_1, j_2)$ of $\brho$, respectively, for $j, j_1, j_2=1,...,p$. { As described in the proof of Theorem~\ref{thm:3},} from Algorithm~\ref{alg:1}, the values $z_j$, $j=1,...,p$, in the subdata consist of $z_{(i)j}$, and
$z_j^{(i)l}$, $l=1,...j-1,j+1,...,p$, $i=1,...,r$, $n-r+1,...,n$,
where $z_j^{(i)l}$ are the concomitants for $z_j$. 

Let $\mathbf{v}=(\Z^*_{\mathrm{D}})\tp\mathbf{1}$ and $\bOmega=(\Z^*_{\mathrm{D}})\tp\Z^*_{\mathrm{D}}$. Then
\begin{equation}\label{eq:19}
  (\X^*_{\mathrm{D}})\tp\X^*_{\mathrm{D}}=
  \begin{bmatrix}
    k & \mathbf{v}\tp \\
    \mathbf{v} & \bOmega
  \end{bmatrix}.
\end{equation}
The $j$th diagonal element of $\bOmega$ is
  \begin{equation}\label{eq:39}
    \Omega_{jj}=\left(\sumr+\sumrn\right)z_{(i)j}^2
      +\sum_{l\neq j}\left(\sumr+\sumrn\right)\left(z_j^{(i)l}\right)^2,
  \end{equation}
 while entry $(j_1,j_2)$, $j_1 \neq j_2$, is
  \begin{equation}\label{eq:40}
    \Omega_{j_1j_2}=\left(\sumr+\sumrn\right)
      \left(z_{(i)j_1}z_{j_2}^{(i)j_1}+z_{(i)j_2}z_{j_1}^{(i)j_2}\right)
      +\sum_{l\neq j_1j_2}\left(\sumr+\sumrn\right)
      z_{j_1}^{(i)l}z_{j_2}^{(i)l}.
    \end{equation}  
The $j$th element of $\mathbf{v}$ is
\begin{equation}\label{eq:41}
  v_j=\left(\sumr+\sumrn\right)z_{(i)j}
    +\sum_{l\neq j}\left(\sumr+\sumrn\right)z_j^{(i)l}.
\end{equation}

Now we consider the two specific distributions in Theorem~\ref{thm:5} and prove the corresponding results in \eqref{eq:17} and \eqref{eq:14}.

\subsubsection{Proof of equation \eqref{eq:17} in Theorem~\ref{thm:5}}
\begin{proof}
 When $\z_i\sim N(\bmu,\bSigma)$, using the results in Example 2.8.1 of \cite{galambos1987asymptotic}, we obtain
 \begin{equation}\label{eq:42}
  \begin{split}
    z_{(i)j}&=\mu_j-\sigma_j\sqrt{2\log n}+\op,\quad
              i=1,...,r,\\
    z_{(i)j}&=\mu_j+\sigma_j\sqrt{2\log n}+\op,\quad
              i=n-r+1,...,n.
  \end{split}
\end{equation}
  Using an approach similar to Example 5.5.1 of
  \cite{galambos1987asymptotic}, we obtain
  \begin{equation}\label{eq:43}
  \begin{split}
    z_j^{(i)l}&=\mu_j-\rho_{lj}\sigma_j\sqrt{2\log n}+\Op,\quad
                i=1,...,r,\\
    z_j^{(i)l}&=\mu_j+\rho_{lj}\sigma_j\sqrt{2\log n}+\Op,\quad
                i=n-r+1,...,n.
  \end{split}
\end{equation}
Using \eqref{eq:42} and \eqref{eq:43}, from \eqref{eq:39}, \eqref{eq:40} and \eqref{eq:41}, we obtain that 
  \begin{align}\label{eq:44}
    \Omega_{jj}=
    &4r\log n\sigma_j^2\sum_{l=1}^p\rho_{lj}^2
      +O_P(\sqrt{\log n}),\\
    \Omega_{j_1j_2}=
    &4r\log n\sigma_{j_1}\sigma_{j_2}
      \sum_{l=1}^p\rho_{lj_1}\rho_{lj_2}+O_P(\sqrt{\log n})\label{eq:45}\\
    v_j=&\Op,\label{eq:53}
  \end{align}
respectively. 
From \eqref{eq:44}, \eqref{eq:45} and \eqref{eq:53}, we have
  \begin{align}\label{eq:46}
    \bOmega=4r\log n \bsigma\brho^2\bsigma+O_P(\sqrt{\log n})
    \quad\text{ and }\quad \mathbf{v}=\Op.
  \end{align}
  The variance,
  \begin{align}\label{eq:20}
    \Var(\hat\bbeta^{\mathrm{D}}|\X)=\sigma^2
    \begin{bmatrix}
      k & \mathbf{v}\tp \\
      \mathbf{v} & \bOmega
    \end{bmatrix}^{-1} &=\frac{\sigma^2}{c}\begin{bmatrix}
      1 & -\mathbf{v}\tp \bOmega^{-1}\\
      -\bOmega^{-1}\mathbf{v} &
      c\bOmega^{-1}+\bOmega^{-1}\mathbf{v}\mathbf{v}\tp \bOmega^{-1}
    \end{bmatrix},
  \end{align}
  where {$c=k-\mathbf{v}\tp \bOmega^{-1}\mathbf{v}=k+O_P(1/\log n)$ and the second equality is from $\eqref{eq:46}$}. Note that from \eqref{eq:46} $\bOmega^{-1}=O_P(1/\log n)$, so
  \begin{align*}
    \bOmega^{-1}-(4r\log n\bsigma\brho^2\bsigma)^{-1}
     &
      =\bOmega^{-1}(4r\log n\bsigma\brho^2\bsigma-\bOmega)(4r\log n\bsigma\brho^2\bsigma)^{-1}\\
     &=O_P\left(\frac{1}{\log n}\right)O_P\big(\sqrt{\log n}\big)
       O\left(\frac{1}{\log n}\right)
       =O_P\left\{\frac{1}{(\log n)^{3/2}}\right\}.
  \end{align*}
  Thus
  \begin{align}\label{eq:21}
    \bOmega^{-1}=\frac{1}{4r\log n}(\bsigma\brho^2\bsigma)^{-1}
    +O_P\left\{\frac{1}{(\log n)^{3/2}}\right\}.
  \end{align}
  Combining \eqref{eq:19}, \eqref{eq:20} and \eqref{eq:21}, and using that $k=2rp$
  \begin{align*}
    \Var(\hat\bbeta^{\mathrm{D}}|\X)=\sigma^2
    \begin{bmatrix}
      \frac{1}{k}+O_P\left(\frac{1}{\log n}\right)
      & O_P\left(\frac{1}{\log n}\right) \\[3mm]
       O_P\left(\frac{1}{\log n}\right)
      & \qquad \frac{1}{4r\log n}(\bsigma\brho^2\bsigma)^{-1}
      +O_P\left\{\frac{1}{(\log n)^{3/2}}\right\}
    \end{bmatrix}.
  \end{align*}
  \end{proof}

\subsubsection{Proof of equation \eqref{eq:14} in Theorem~\ref{thm:5}}
\begin{proof}
When $\z_i\sim LN(\bmu,\bSigma)$. Let $z_{ij}=\exp{(U_{ij})}$ with $\mathbf{U}_{i}=(U_{i1}, ..., U_{ip})\tp\sim N(\bmu,\bSigma)$. From~\eqref{eq:42},
\begin{equation}\label{eq:47}
  \begin{split}
    z_{(i)j}=\exp(U_{(i)j})
    &=\exp(-\sigma_j\sqrt{2\log n})\Op=\op
      ,\quad i=1,...,r,\\
    z_{(i)j}=\exp(U_{(i)j})
    &=\exp(\sigma_j\sqrt{2\log n})\{e^{\mu_j}+\op\}
      ,\quad i=n-r+1,...,n.
    \end{split}
  \end{equation}
Without loss of generality, assume that $\rho_{lj}\ge0$, $l,j=1,...,p$. From~\eqref{eq:43},
\begin{equation}\label{eq:48}
  \begin{split}
    z_j^{(i)l}=\exp(U_j^{(i)l})
    &=\exp(-\rho_{lj}\sigma_j\sqrt{2\log n})\Op=\op
      ,\quad i=1,...,r,\\
    z_j^{(i)l}=\exp(U_j^{(i)l})
    &=\exp\{\sigma_j\sqrt{2\log n}
    -(1-\rho_{lj})\sigma_j\sqrt{2\log n}+\mu_j+\Op\}\\
    &=\exp(\sigma_j\sqrt{2\log n})\op
      ,\quad i=n-r+1,...,n.
  \end{split}
\end{equation}
Using \eqref{eq:47} and \eqref{eq:48}, from \eqref{eq:39}, \eqref{eq:40} and \eqref{eq:41}, we obtain that 
  \begin{align}
    \Omega_{jj}=
    &r\exp(2\sigma_j\sqrt{2\log n})\{e^{2\mu_j}+\op\},\label{eq:49}\\
    \Omega_{j_1j_2}=
    &2r\exp\Big\{(\sigma_{j_1}+\sigma_{j_2})\sqrt{2\log n}\Big\}\op,
      \label{eq:50}\\
    v_j=&r\exp(\sigma_j\sqrt{2\log n})\{e^{\mu_j}+\op\}.\label{eq:51}
  \end{align} 
From \eqref{eq:19}, \eqref{eq:49}-\eqref{eq:51}, for $\A_n=\diag\Big\{1, \exp\big(\sigma_1\sqrt{2\log n}\big), ..., \exp\big(\sigma_p\sqrt{2\log n}\big)\Big\}$,
\begin{equation}\label{eq:52}
  \A_n^{-1}(\X^*_{\mathrm{D}})\tp\X^*_{\mathrm{D}}\A_n^{-1}=
  \A_n^{-1}\begin{bmatrix}
    k & \mathbf{v}\tp \\
    \mathbf{v} & \bOmega
  \end{bmatrix}\A_n^{-1}
  =\begin{bmatrix}
    k & r\mathbf{v}_1\tp \\
    r\mathbf{v}_1 & r\B_5,
  \end{bmatrix}+\op
\end{equation}
where $\mathbf{v}_1=(e^{\mu_1}, ..., e^{\mu_p})\tp$ and  $\B_5=\diag(e^{2\mu_1}, ..., e^{2\mu_p})$. From \eqref{eq:52},
\begin{align*}
  \Var(\A_n\hat\bbeta^{\mathrm{D}}|\X)
  =\sigma^2\A_n\{(\X^*_{\mathrm{D}})\tp\X^*_{\mathrm{D}}\}^{-1}\A_n
  &=\sigma^2\begin{bmatrix}
    k & r\mathbf{v}_1\tp \\
    r\mathbf{v}_1 & r\B_5,
  \end{bmatrix}^{-1}+\op\\
  &=\frac{2\sigma^2}{k}\begin{bmatrix}
    1 & -\mathbf{u}\tp \\
    -\mathbf{u} & p\bLambda+\mathbf{u}\mathbf{u}\tp,
  \end{bmatrix}+\op.
\end{align*}
  \end{proof}

\subsection{Proof of results in Table~\ref{tab:1}}
{ When the covariate has a $t$ distribution, from Theorem~\ref{thm:4}, for simple linear model, the variance of the estimator of $\beta_1$ using the D-OPT IBOSS approach is of the same order as $(z_{(n)1}-z_{(1)1})^{-2}$. From Theorems 2.1.2 and 2.9.2 of \cite{galambos1987asymptotic}, we obtain that $z_{(n)1}-z_{(1)1}\asymp_P n^{1/\nu}$. Thus, the variance is of the order $n^{-2/\nu}$.} 

For the full data approach, the variance of the estimator of $\beta_1$
is of the same order as $(\sumn z_{i1}^2)^{-1}$. When $z_1$ has a $t$
distribution with degrees of freedom $\nu>2$, from Kolmogorov's strong
law of large numbers (SLLN), $\sumn z_{i1}^2=O(n)$ almost surely. If
$\nu\le2$, $\Exp[\{z_{i1}^2\}^{1/(2/\nu+\alpha)}]<\infty$ for any
$\alpha>0$. Thus, from Marcinkiewicz-Zygmund SLLN \citep[Theorem 2 of
Section 5.2 of][]{ChowTeicher2003},
$\sumn z_{ij}^2=o(n^{2/\nu+\alpha})$ almost surely for any
$\alpha>0$. This shows that the order of $(\sumn z_{i1}^2)^{-1}$ is
slower than $n^{-(2/\nu+\alpha)}$ for any $\alpha>0$.
  
For the UNI approach, the lower bound for the variance of the estimator of $\beta_1$ is of the same order as $n(\sumn z_{i1}^2)^{-1}$, which is of order $O(1)$ when $\nu>2$ and is
slower than $n^{2/\nu-1+\alpha}$ for any $\alpha>0$ when $\nu\le2$.

For the intercept $\beta_0$, the variance of the estimator is of the same order as the inverse of the sample size used in each method.


\newpage
\setcounter{equation}{0}
\setcounter{page}{1}
\setcounter{section}{0}
\setcounter{figure}{0}
\renewcommand{\theequation}{S.\arabic{equation}}
\renewcommand{\thesection}{S.\arabic{section}}
\renewcommand{\thefigure}{S.\arabic{figure}}

\begin{center}\Large
  Supplementary Material\\ for ``Information-Based Optimal Subdata Selection for Big Data Linear Regression''
\end{center}

\spacingset{1.45}

We present additional numerical results about the performance of the IBOSS method. 
\section{Predictive performance}
In this section, we investigate the performance of IBOSS in predicting the mean response for a given setting of covariates. We focus on the mean squared prediction error (MSPE), 
  \begin{align}\label{eq:s1}
    \text{MSPE}
    =E[\{E(y_{new})-\hat{y}_{new}\}^2]
    =E[\{\x_{new}\tp(\hat{\bbeta}-\bbeta)\}^2].
  \end{align}
Note that the mean squared prediction error for predicting a future response is 
  \begin{align}
    E\{(y_{new}-\hat{y}_{new})^2\}
    &=E[\{y_{new}-E(y_{new})\}^2]+E[\{E(y_{new})-\hat{y}_{new}\}^2]
    =\sigma^2+\text{MSPE},
  \end{align}
and the variance of $y_{new}$, $\sigma^2$, cannot be reduced by choosing a better subdata or a larger subdata sample size $k$. Thus it is reasonable to focus on the MSPE in \eqref{eq:s1} to evaluate the performance of IBOSS. For prediction, the estimation of $\beta_0$ is also important, so we use $\hat\beta_0^{Da}=\bar y-\bar\z\tp\hat{\bbeta}^{\mathrm{D}}_1$ as indicated in the paper. 

We use the same five cases considered in the paper to generate full
data sets. In addition, we consider another case, Case 6, in which the covariates are from a multivariate $t$ distribution with degrees of freedom $\nu=1$. This is a case often used in evaluating the performance of the LEV method. 
\begin{enumerate}[{Case} 6., leftmargin=*]
  \addtolength{\itemsep}{-0.5\baselineskip}
\item $\z_i$'s have a multivariate $t$ distribution with degrees of freedom $\nu=1$, i.e., $\z_i\sim t_1(\mathbf{0}, \bSigma)$.
\end{enumerate}
For each case, we implement different methods to obtain parameter estimates, and
then generate a new sample of size 5,000 to calculate the MSPEs. The
simulation is repeated 1,000 times and empirical MSPEs are calculated.
Figure~\ref{fig:s1} presents plots of the log$_{10}$ of the MSPEs
against log$_{10}(n)$. For prediction, the relative performance of IBOSS compared with
other methods are similar to that of parameter estimation. That
is, the D-OPT IBOSS method uniformly dominates the subsampling-based
methods UNI and LEV, and its advantage is more significant if the
tail of the covariate distribution is heavier. Specifically for Case 6, it is seen that the performance of D-OPT IBOSS is almost identical to that of the full data approach, and LEV significantly outperforms the UNI. 
      
\begin{figure} 
  \centering
  \begin{subfigure}{0.49\textwidth}
    \includegraphics[width=\textwidth]{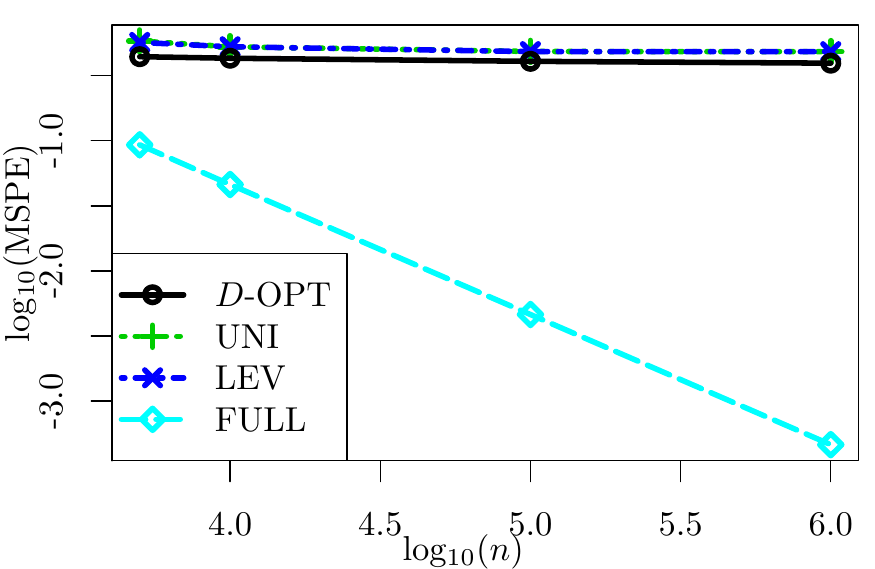}\\[-10mm]
    \caption{Case 1: $\z_i$'s are normal.}
  \end{subfigure}
  \begin{subfigure}{0.49\textwidth}
    \includegraphics[width=\textwidth]{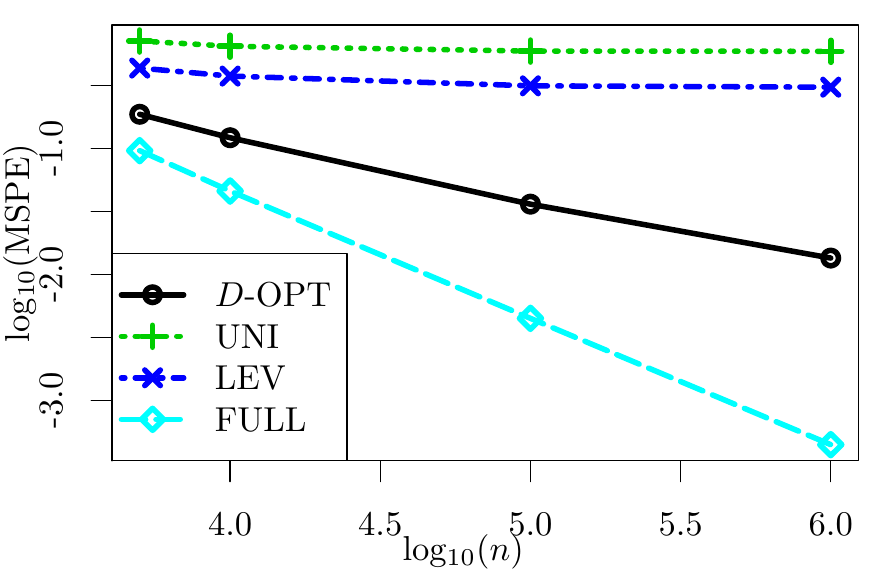}\\[-10mm]
    \caption{Case 2: $\z_i$'s are lognormal.}
  \end{subfigure}\\[3mm]
  \begin{subfigure}{0.49\textwidth}
    \includegraphics[width=\textwidth]{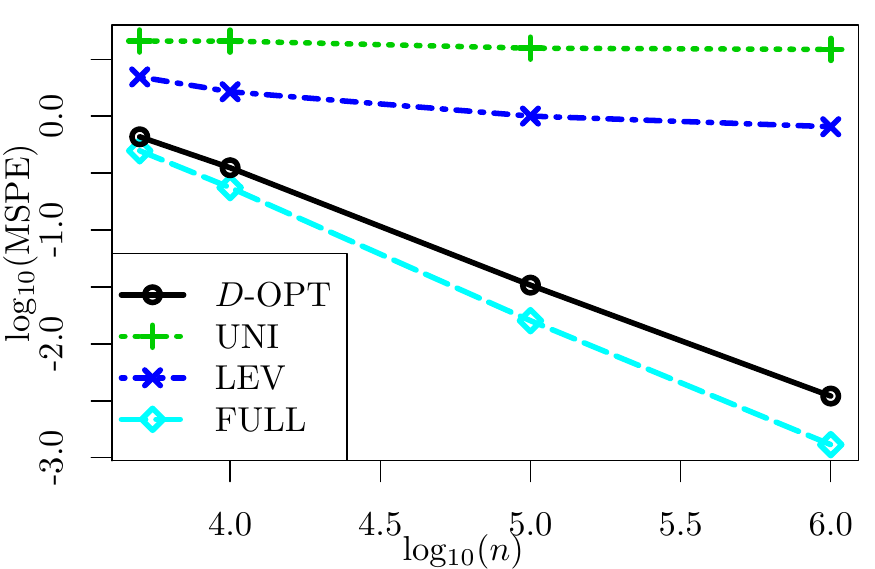}\\[-10mm]
    \caption{Case 3: $\z_i$'s are $t_2$.}
  \end{subfigure}
  \begin{subfigure}{0.49\textwidth}
    \includegraphics[width=\textwidth]{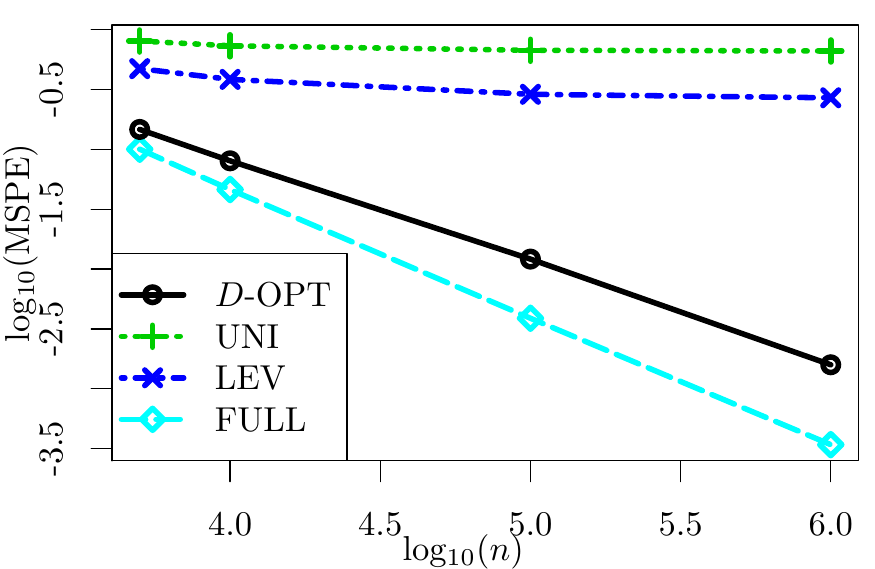}\\[-10mm]
    \caption{Case 4: $\z_i$'s are a mixture.}
  \end{subfigure}\\[3mm]
  \begin{subfigure}{0.49\textwidth}
    \includegraphics[width=\textwidth]{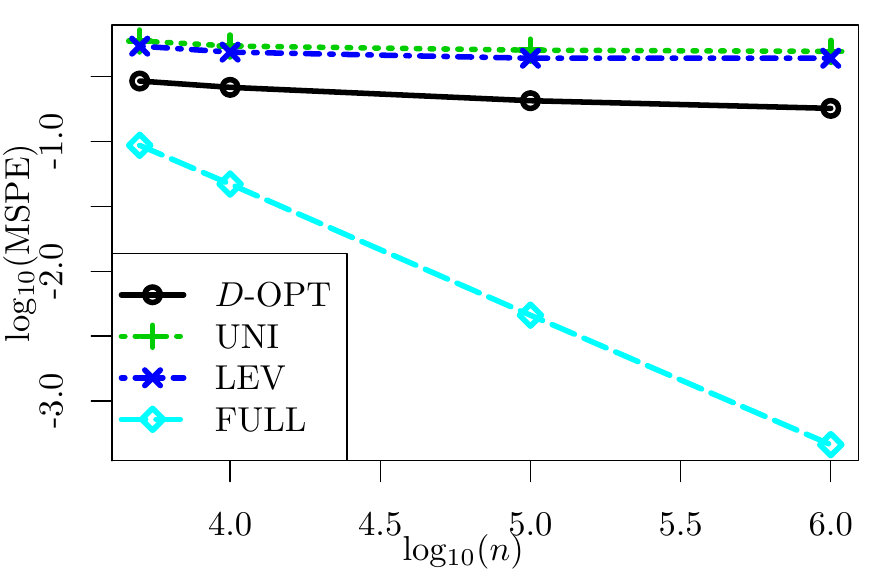}\\[-10mm]
    \caption{Case 5: $\z_i$'s include interaction terms.}
  \end{subfigure}
  \begin{subfigure}{0.49\textwidth}
    \includegraphics[width=\textwidth]{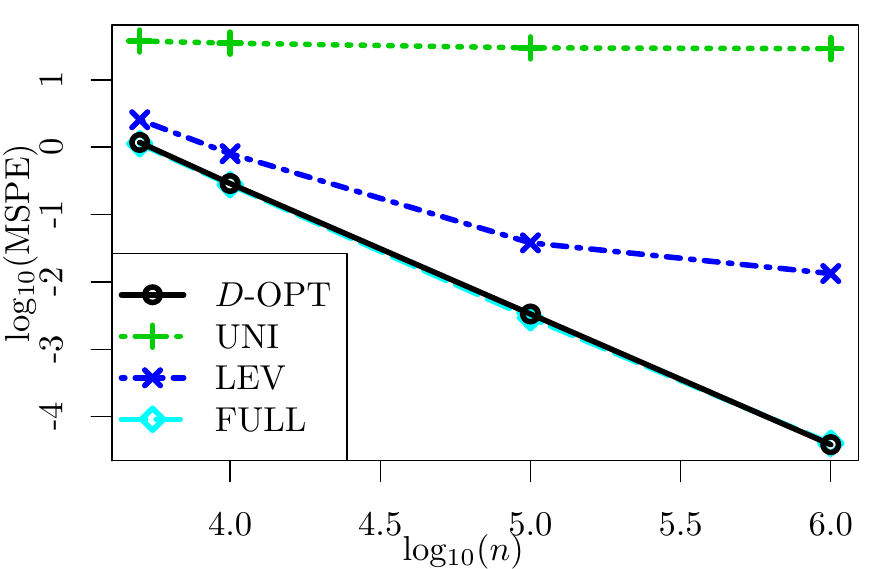}\\[-10mm]
    \caption{Case 6: $\z_i$'s are $t_1$.}
  \end{subfigure}
  \caption{MSPEs for predicting mean responses for six
    different distributions of the covariates $\z_i$. The subdata
    size $k$ is fixed at $k=1000$ and the full data size $n$
    changes. Logarithm with base 10 is taken of $n$ and MSPEs for
    better presentation of the figures.}
  \label{fig:s1}
\end{figure}

\section{Column permutation}
In this section, we provide numerical results accessing the effect of column permutation on the IBOSS method. To differentiate the effect of each column in the covariate matrix, we change the covariance matrix $\bSigma$ such that $\Sigma_{ij}=0.5^{|i-j|}$ if $i\neq j$, and $\Sigma_{ij}=1+3(i-1)/p$ if $i=j$, $i,j=1, ..., 50$. With this setup, the correlation structure for the covariates is unexchangeable and variances for different columns are different. Using this covariance matrix, we generate covariates $\z_i$'s according to Case 5 in Section 5.1 of the paper. The IBOSS method is applied with the original order of covariate columns as well as with a single random permutation of covariate columns. Results are presented in Figure~\ref{fig:s5}. It is seen that the performances of IBOSS for the two approaches are very similar. This agrees with the theoretical results. 

\begin{figure} 
  \centering
  \begin{subfigure}{0.49\textwidth}
    \includegraphics[width=\textwidth]{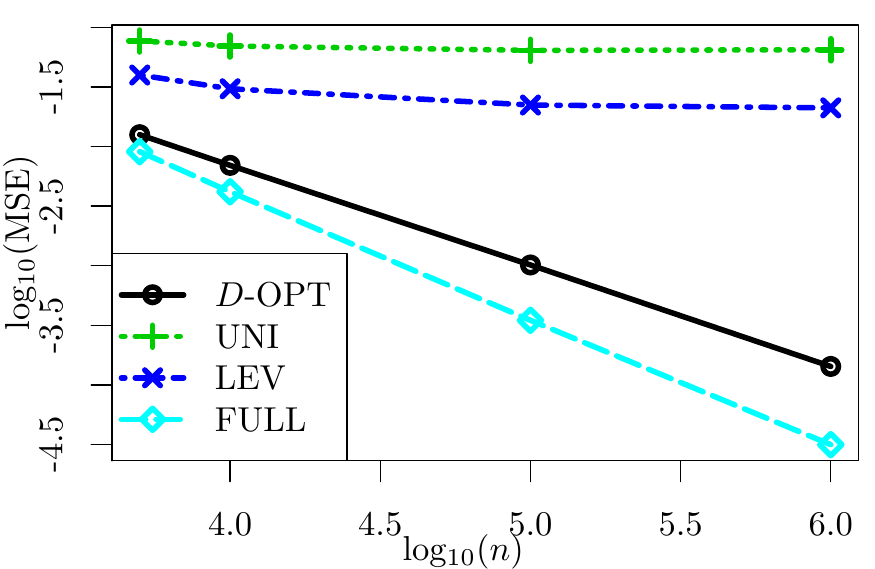}\\[-10mm]
    \caption{Original order, slope parameter}
  \end{subfigure}
  \begin{subfigure}{0.49\textwidth}
    \includegraphics[width=\textwidth]{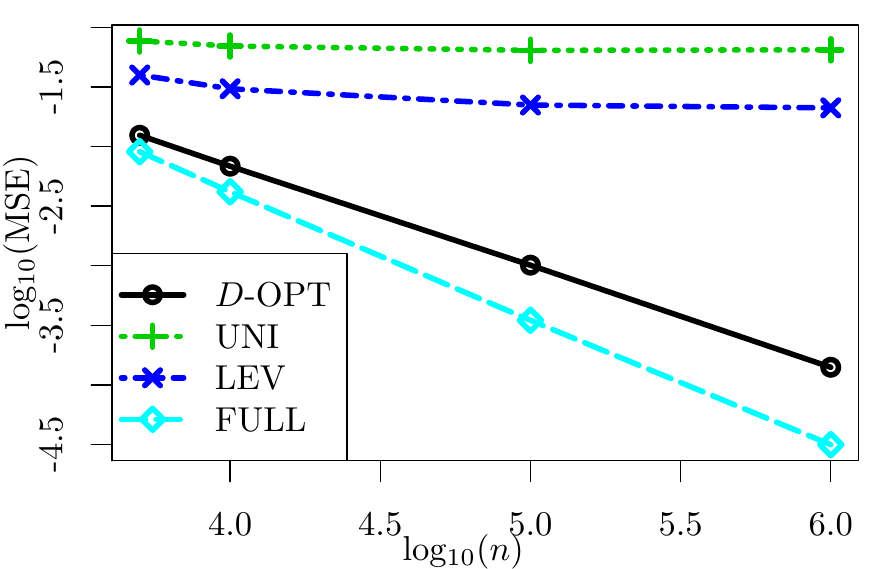}\\[-10mm]
    \caption{Shuffled order, slope parameter}
  \end{subfigure}\\[3mm]
  \begin{subfigure}{0.49\textwidth}
    \includegraphics[width=\textwidth]{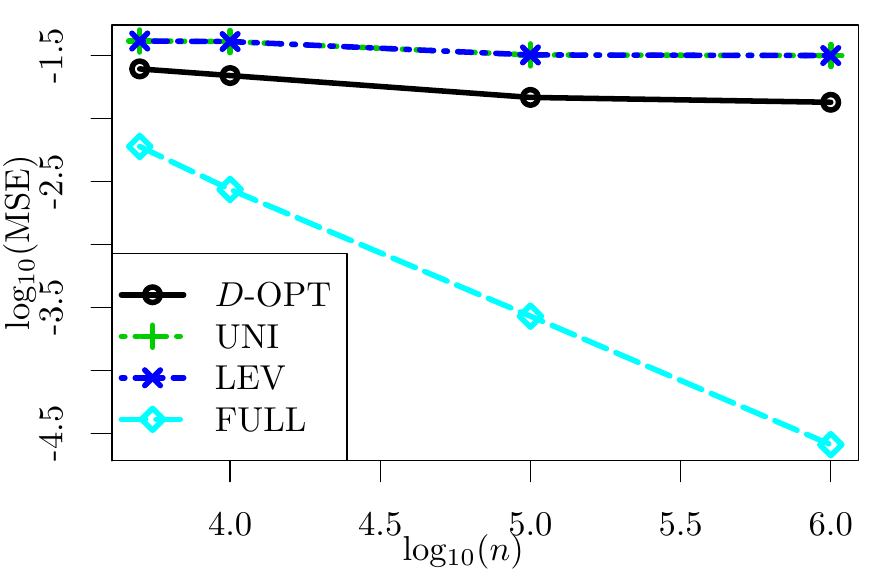}\\[-10mm]
    \caption{Original order, intercept parameter}
  \end{subfigure}
  \begin{subfigure}{0.49\textwidth}
    \includegraphics[width=\textwidth]{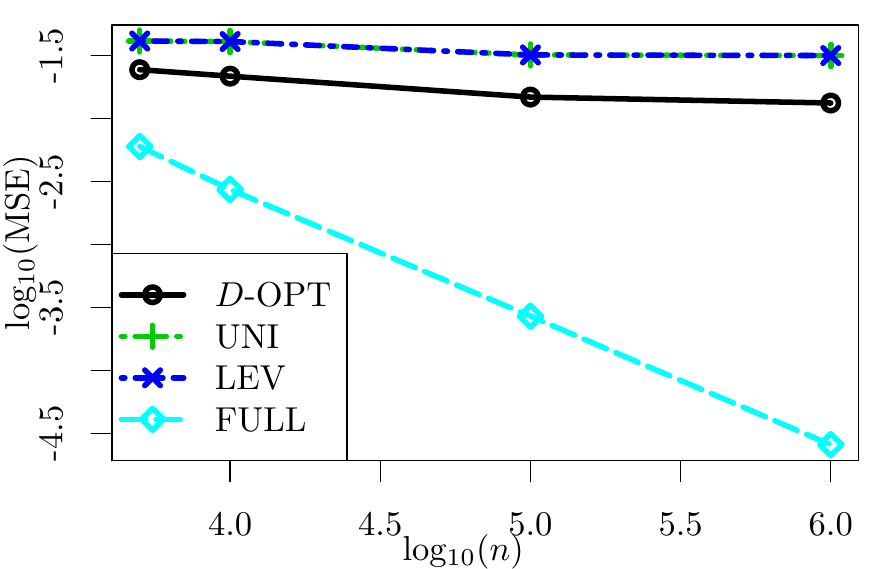}\\[-10mm]
    \caption{Shuffled order, intercept parameter}
  \end{subfigure}
  \caption{MSEs for estimating the slope parameter (top panel) and the intercept parameter (bottom panel) with different orders of the covariate columns. The left panel presents results with the original order of covariate columns and the right panel presents results with the randomly shuffled order of covariate columns. The subdata size $k$ is fixed at $k=1000$ and the full data size $n$
    changes. Logarithm with base 10 is taken of $n$ and MSEs for
    better presentation of the figures.}
  \label{fig:s5}
\end{figure}

\section{Interaction model}
In this section, we consider a case that the true model contains all the main effects and all the pairwise interaction terms. However, only the main effects are used in selecting subdata. 
Data are generated from the following linear model,
\begin{align}
  y_i&=\beta_0+\sum_{j=1}^{10}z_{ij}\beta_j
       +\sum_{j_1\neq j_2}^{10}z_{ij_1}z_{ij_2}\beta_{j_1j_2}+\varepsilon_i,\
       \quad i=1,...,n,
\end{align}
where the true value of regression coefficients are
$\beta_j=\beta_{j_1j_2}=1$ for $j,j_1,j_2=1,...,10$, and $\varepsilon_i$'s are
i.i.d. $N(0,9)$. Two different distributions are considered to generate covariates $\z_i$'s: one is a multivariate
normal distribution $\z_i\sim N(\mathbf{0}, \bSigma_{10\times10})$ and the other is a multivariate
lognormal distribution $\z_i\sim LN(\mathbf{0}, \bSigma_{10\times10})$,
where $\bSigma_{10\times10}$ is a 10 by 10 covariance matrix with $\Sigma_{ij}=0.5^{I(i\neq j)}$, for $i, j=1, ..., 10$. In selecting subdata, only the main effects are used. The interaction terms are not used in subdata selection but are used in parameter estimation.

Figure~\ref{fig:s2} presents the MSEs for estimating the slope parameters, which are calculated from 1000 iterations of the simulation. It is seen that IBOSS is still the most efficient method among subdata-based methods for both of the distributions.

\begin{figure}[H]
  \centering
  \begin{subfigure}{0.49\textwidth}
    \includegraphics[width=\textwidth]{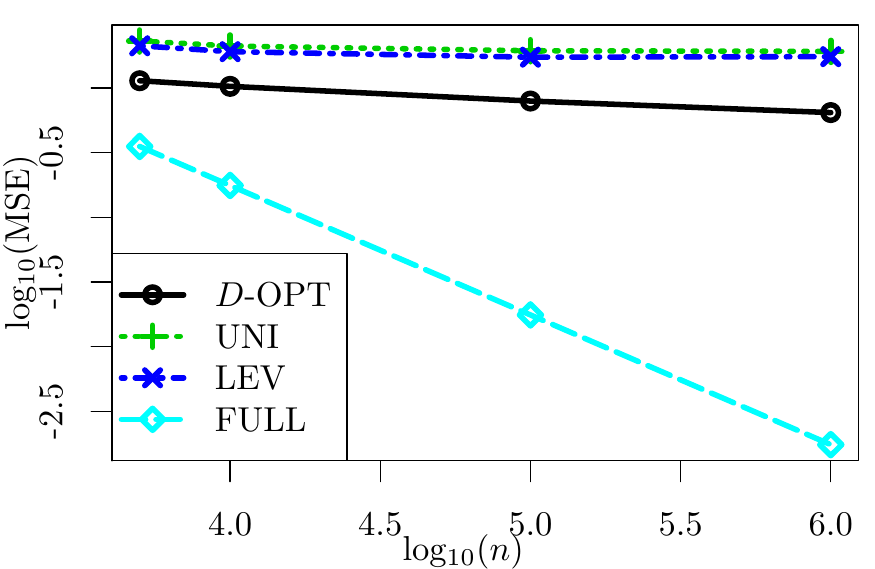}\\[-10mm]
    \caption{Case 1: $\z_i$'s are normal.}
  \end{subfigure}
  \begin{subfigure}{0.49\textwidth}
    \includegraphics[width=\textwidth]{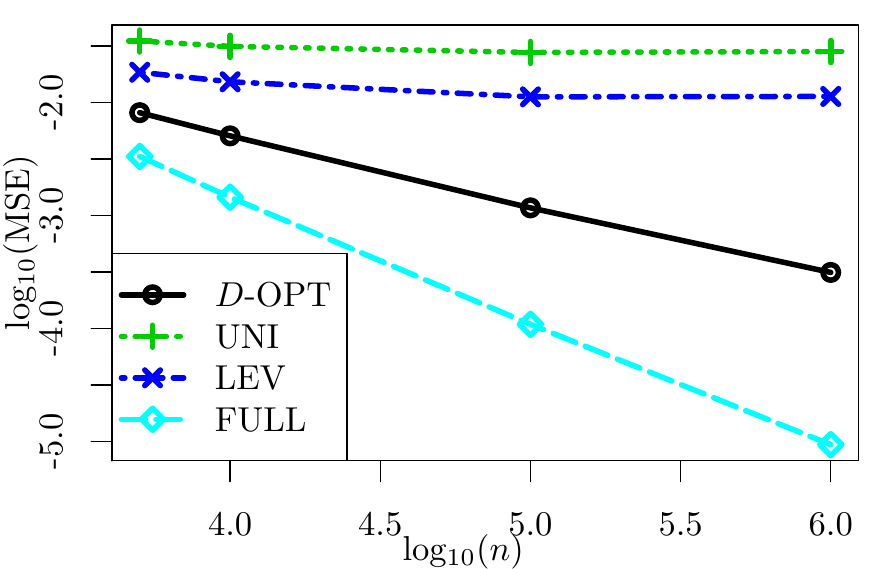}\\[-10mm]
    \caption{Case 2: $\z_i$'s are lognormal.}
  \end{subfigure}
  \caption{MSEs for estimating the slope parameter for two
    different distributions of the covariates $\z_i$. The subdata
    size $k$ is fixed at $k=1000$ and the full data size $n$
    changes. Logarithm with base 10 is taken of $n$ and MSEs for
    better presentation of the figures.}
  \label{fig:s2}
\end{figure}

\section{Nonlinear relationships}
In this section, we consider the scenario that true relationships between the response and the covariates are nonlinear, and transformations cannot linearize the relationships, i.e., a finite-dimensional linear model cannot be correct. We consider the following two models

\begin{align}
  y_i&=\beta_0+\sum_{j=1}^{p-1}z_{ij}\beta_j
       +\frac{3e^{z^{(t)}_{ip}}}{1+e^{z^{(t)}_{ip}}}+\varepsilon_i,\
       \quad i=1,...,n,\tag{WM1}\label{WM:1}\\
  y_i&=\beta_0+\sum_{j=1}^{p-1}z_{ij}\beta_j
       +30\log\Big(1+e^{z^{(t)}_{ip}}\Big)+\varepsilon_i,\
       \quad i=1,...,n,\tag{WM2}\label{WM:2}
\end{align}
where $z^{(t)}_{ip}=z_{ip}I(z_{ip}\le100)+100I(z_{ip}>100)$. 
Covariates and parameter setups are the same as those of Case 4 for the mixture distribution. Although full data are generated from nonlinear model~\eqref{WM:1} or \eqref{WM:2}, the linear main effects model is used for subdata selection and analysis. 

Figure~\ref{fig:s3} presents plots of the log$_{10}$ of the MSEs of estimating the slope parameter and the intercept parameter against log$_{10}(n)$, and plots of the log$_{10}$ of the MSPEs of predicting the mean response. It is seen that, including the full data approach, no method dominates others and larger sample sizes do not necessarily mean more accurate results. When the underlying model is incorrect, the problem is very complicated and there is no simple answer to which method will produce satisfactory results. We present the numerical studies here to show that IBOSS does not always produce the worst results for this scenario, but we have no intention to state that the IBOSS works better than other methods.

\begin{figure}  [H]
  \centering
  \begin{subfigure}{0.49\textwidth}
    \includegraphics[width=\textwidth]{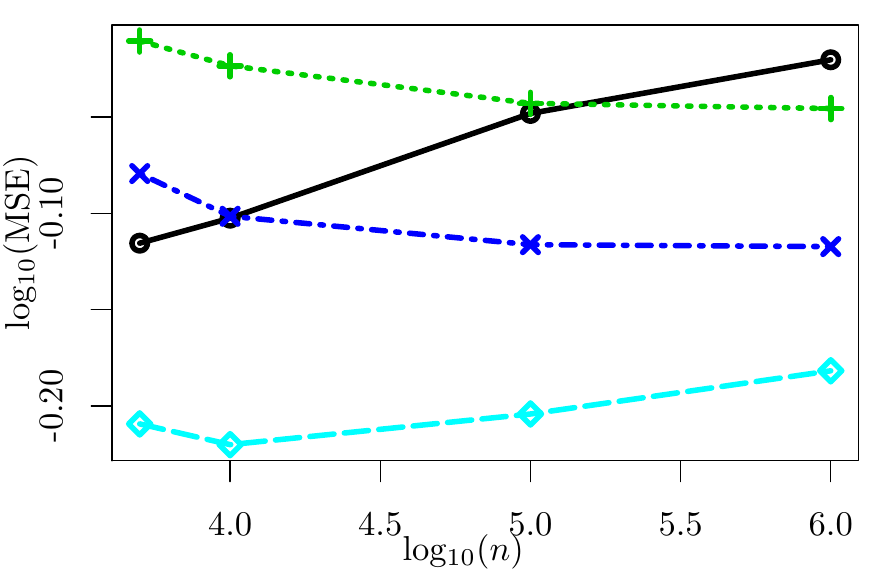}\\[-10mm]
    \caption{Model~\eqref{WM:1}, slope parameter}
  \end{subfigure}
  \begin{subfigure}{0.49\textwidth}
    \includegraphics[width=\textwidth]{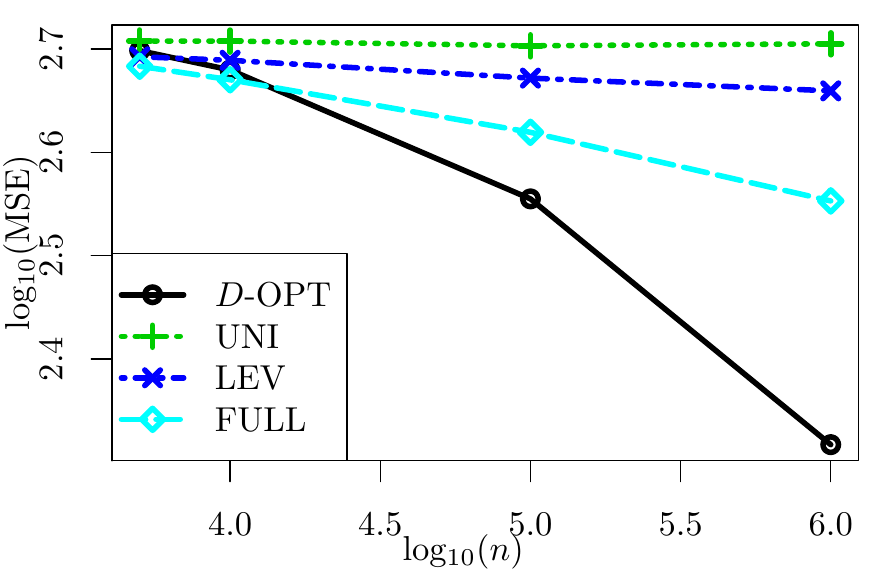}\\[-10mm]
    \caption{Model~\eqref{WM:2}, slope parameter}
  \end{subfigure}\\[3mm]
  \begin{subfigure}{0.49\textwidth}
    \includegraphics[width=\textwidth]{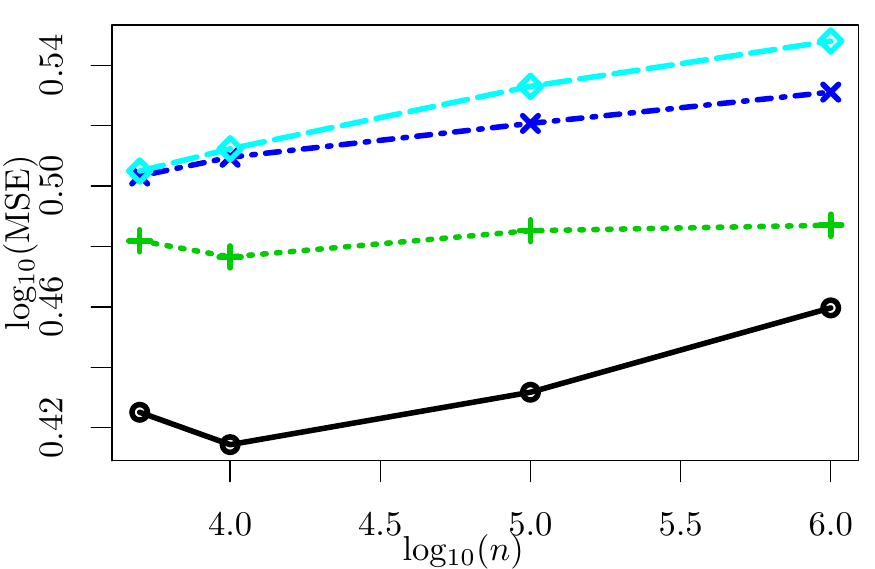}\\[-10mm]
    \caption{Model~\eqref{WM:1}, intercept parameter}
  \end{subfigure}
  \begin{subfigure}{0.49\textwidth}
    \includegraphics[width=\textwidth]{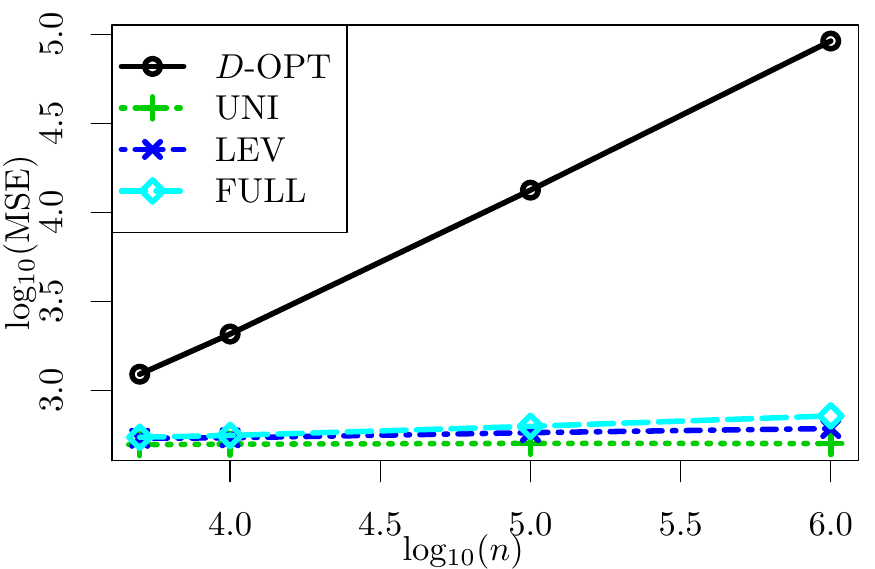}\\[-10mm]
    \caption{Model~\eqref{WM:2}, intercept parameter}
  \end{subfigure}\\[3mm]
  \begin{subfigure}{0.49\textwidth}
    \includegraphics[width=\textwidth]{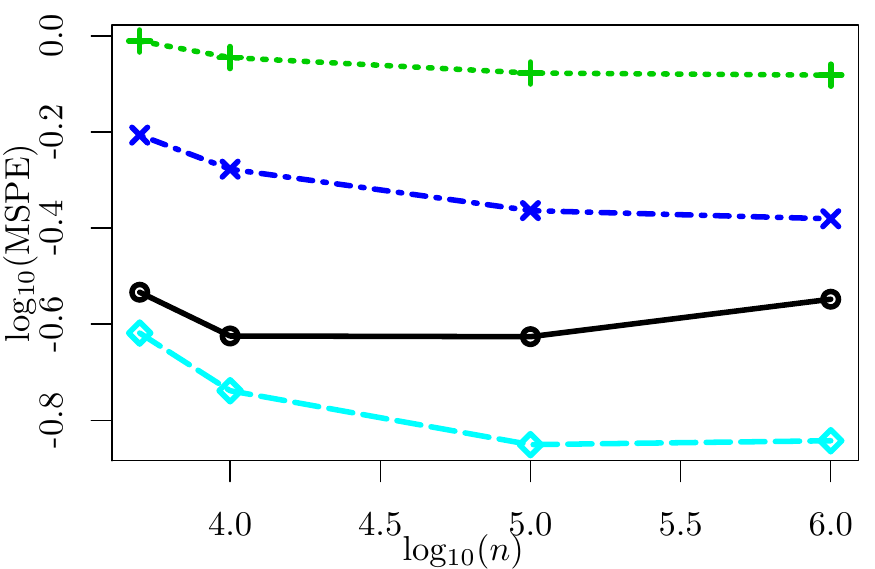}\\[-10mm]
    \caption{Model~\eqref{WM:1}, prediction}
  \end{subfigure}
  \begin{subfigure}{0.49\textwidth}
    \includegraphics[width=\textwidth]{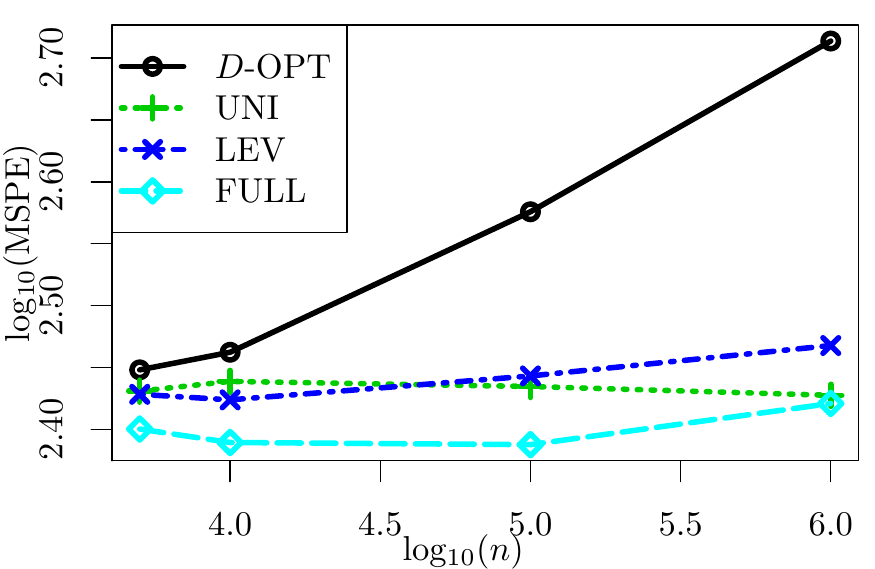}\\[-10mm]
    \caption{Model~\eqref{WM:2}, prediction}
  \end{subfigure}
  \caption{MSEs for estimating the slope parameter (top row), MSEs for estimating the intercept parameter (middle row), and MSPEs for predicting the mean response (bottom row) when true models are nonlinear. The left column is for model~\eqref{WM:1} and the right column is for model~\eqref{WM:2}. The subdata
    size $k$ is fixed at $k=1000$ and the full data size $n$
    changes. Logarithm with base 10 is taken of $n$ and MSEs for
    better presentation of the figures.}
  \label{fig:s3}
\end{figure}

\section{Accuracy-cost tradeoff of the IBOSS method}
In this section, we provide additional results showing the accuracy-cost tradeoff of the IBOSS method. Full data of size $n=5\times10^6$ are generated using the same setup of Case 1. The IBOSS method is implemented with subdata sample sizes of $k=10^2, 10^3,10^4,10^5$ and $10^6$, and the average CPU times and MSEs are calculated from 100 repetitions of the simulation. Results are reported in Figure~\ref{fig:s4}. 
It is seen that as the required CPU time increases, the MSE decreases, which indicates a clear tradeoff between computational cost and estimation accuracy for the IBOSS method. However, as the CPU time increases, the MSE can drop sharply. For example, when the CPU time increases from 6.4976 seconds (corresponding to $k=10^2$) to 7.0839 seconds, the MSE decreases from 13.57091 to 0.00786855. Thus the IBOSS has the advantage to significantly increase the estimation accuracy with little increase in computational cost.

\begin{figure}[H]
  \centering
    \includegraphics[width=\textwidth]{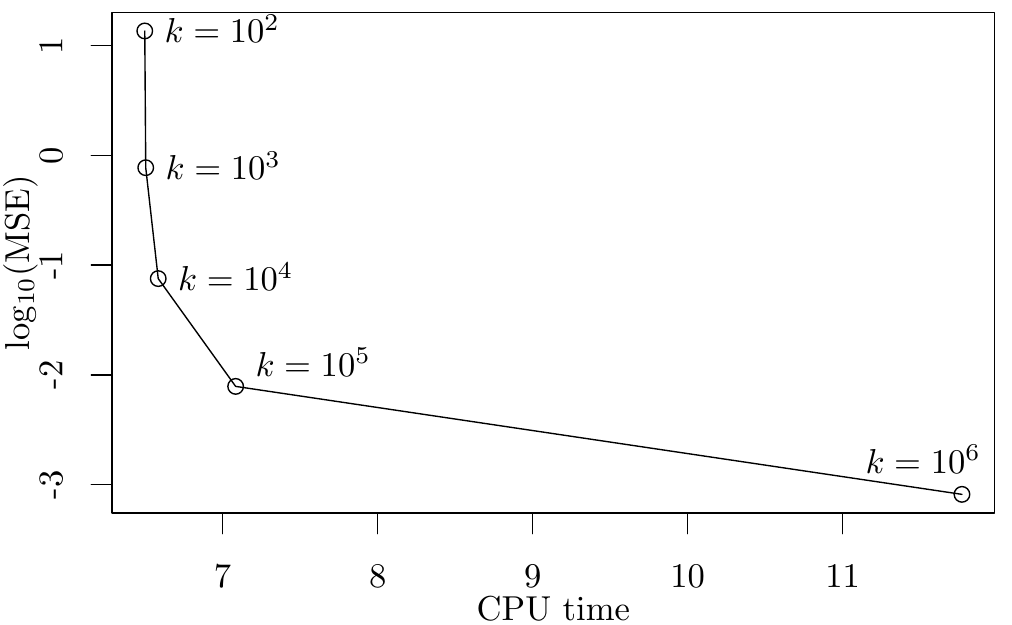}
  \caption{Average CPU times and  MSEs for different subdata sample size $k$ when the covariates are from a multivariate normal distribution. The full data size is set to $n=5\times10^6$ with a dimension $p=50$.}
  \label{fig:s4}
\end{figure}

We perform additional experiments to further investigate the accuracy-cost tradeoff of the IBOSS for both large $n$ and large $p$, and draw comparisons with the performance of repeating the UNI method. Full data are generated with $n=5\times10^5$ and $p=500$, and subdata of sizes $k=10^3, 5\times10^3, 10^4$, $5\times10^4$, and $10^5$ are taken using the IBOSS method or the UNI method. For the UNI method, it is repeated multiple times so that it consumes similar CPU times to the IBOSS method, and the average of the estimates from all repetitions are used as the final estimate. Figure~\ref{fig:s6} presents the results when the covariates are from the multivariate normal distribution (Case 1) and the mixture distribution (Case 4) described in Section 5 of the main paper. 
The average CPU times and MSEs for the slope parameters are calculated from 100 repetitions of the simulation. 
For Case 1 with multivariate normal covariates, the repeated UNI method may produce smaller MSEs compared with the IBOSS method using similar CPU times. However, the differences are not very significant compared with the advantage of the IBOSS method for Case 4, in which the covariate distribution has a heaver tail. 

\begin{figure}[H]
  \centering
  \begin{subfigure}{0.495\textwidth}
    \centering
    \includegraphics[width=\textwidth]{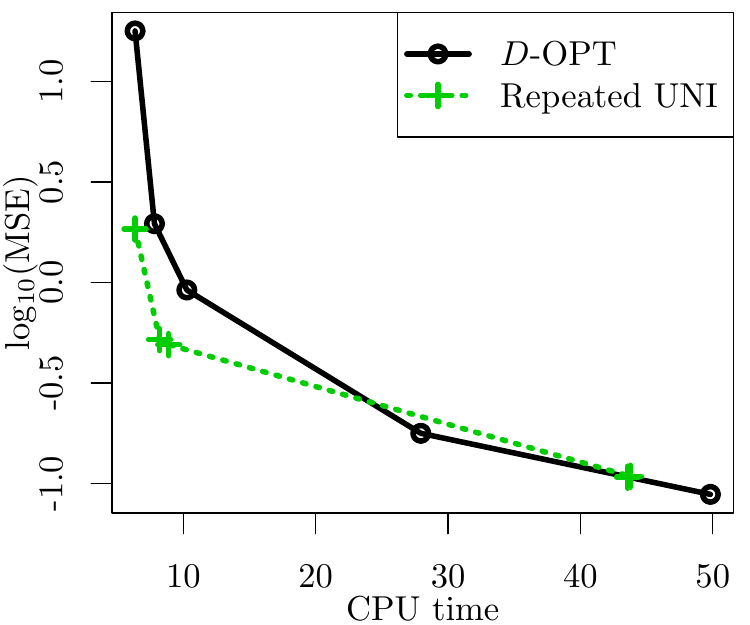}
    \caption{Case 1: $\z_i$'s are normal.}
  \end{subfigure}
  \begin{subfigure}{0.495\textwidth}
    \centering
    \includegraphics[width=\textwidth]{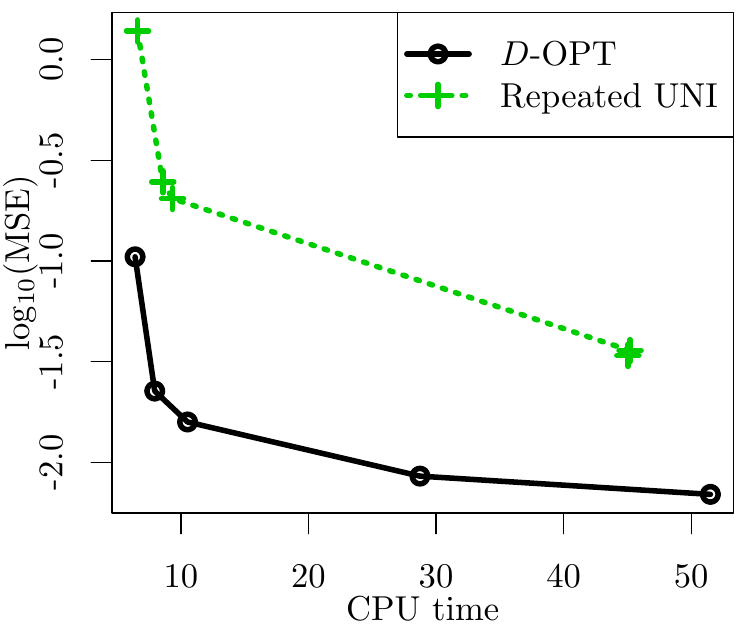}
    \caption{Case 4: $\z_i$'s are a mixture.}
  \end{subfigure}
  \caption{MSEs for different CPU times when the covariates are from a multivariate normal distribution (a) and a mixture distribution (b). The full data size is set to $n=5\times10^5$ with dimension $p=500$.  Subdata sample size are $k=10^3, 5\times10^3, 10^4$, $5\times10^4$, and $10^5$.}
  \label{fig:s6}
\end{figure}

\section{Comparison with the divide-and-conquer method}
In this section, we provide numerical results comparing the IBOSS method and the divide-and-conquer (DC) method proposed in Section 4.3 of \cite{battey2015distributed}. The DC method divide the full data into $S$ subdata sets (The notation $k$ is used in \cite{battey2015distributed}; we use $S$ here because $k$ is used to denote the subdata size.), and the ordinary least squares estimate, say $\hat{\bbeta}_s$, is calculated for each subdata. The DC estimate is the average of $\hat{\bbeta}_s$'s, i.e., $\bar{\bbeta}=S^{-1}\sum_{s=1}^S\hat{\bbeta}_s$. We choose $S=\floor{n^{1/4}}$. In our implementation, if $n/S$ is not an integer, the last subdata will have a sample size of $n-\floor{n/S}*(S-1)$.

Figure~\ref{fig:s7} gives the average CPU times and MSEs for the slope parameters with dimension $p=50$ and different full data size $n$, with choices of $5\times10^3, 10^4, 10^5$, and $10^6$. The average CPU times and MSEs are calculated from 100 repetitions of the simulation. It is seen that the relative performances of estimation efficiency between the IBOSS D-OPT method and the DC method depend on the covariate distribution. The DC method is better when covariates are normally distributed; the IBOSS D-OPT method and the DC method perform similarly when the covariate has a mixture distribution; the IBOSS D-OPT dominates the DC method when the covariate has a $t_1$ distribution. In terms of computational cost in Figure~\ref{fig:s7} (d), the IBOSS D-OPT is more efficient than the DC method especially for large values of $n$. Note that the CPU times for either the DC method or the IBOSS D-OPT method do not depend on the covariate distribution. 

\begin{figure}[H]
  \centering
  \begin{subfigure}{0.49\textwidth}
    \includegraphics[width=\textwidth]{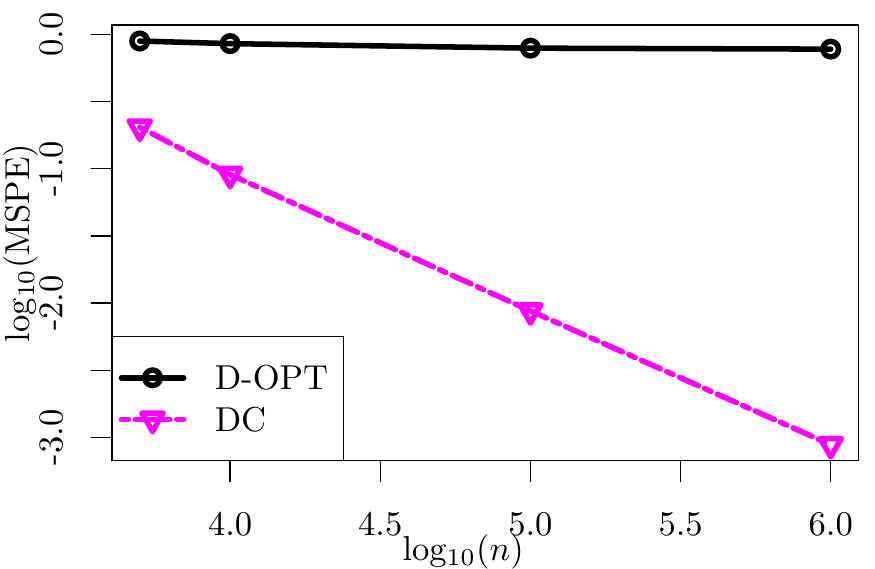}\\[-10mm]
    \caption{Case 1: $\z_i$'s are normal.}
  \end{subfigure}
  \begin{subfigure}{0.49\textwidth}
    \includegraphics[width=\textwidth]{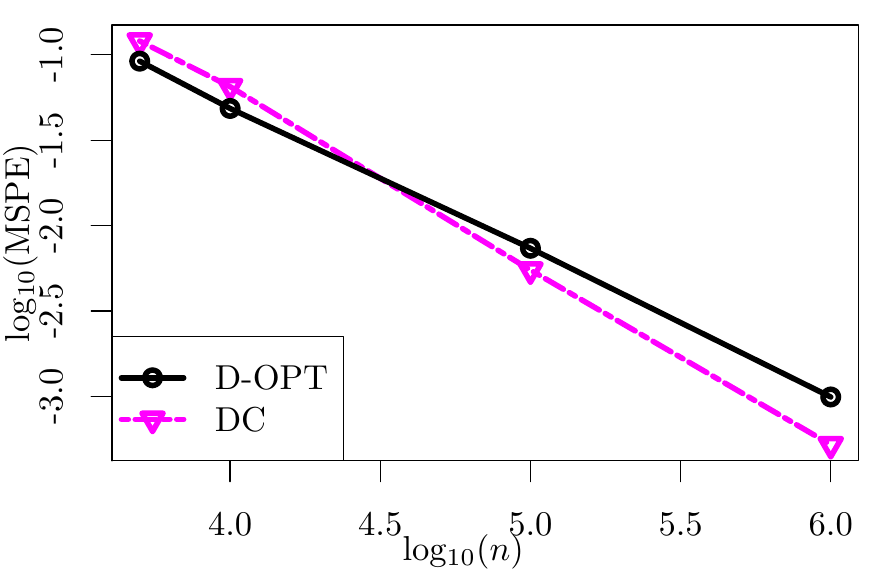}\\[-10mm]
    \caption{Case 4: $\z_i$'s are a mixture.}
  \end{subfigure}\\[3mm]
  \begin{subfigure}{0.49\textwidth}
    \includegraphics[width=\textwidth]{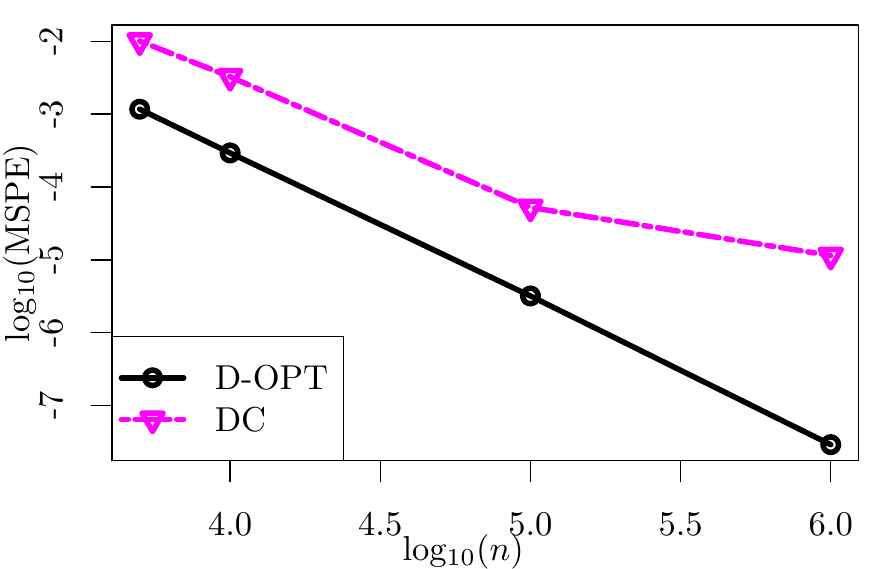}\\[-10mm]
    \caption{Case 6: $\z_i$'s are $t_1$.}
  \end{subfigure}
  \begin{subfigure}{0.49\textwidth}
    \includegraphics[width=\textwidth]{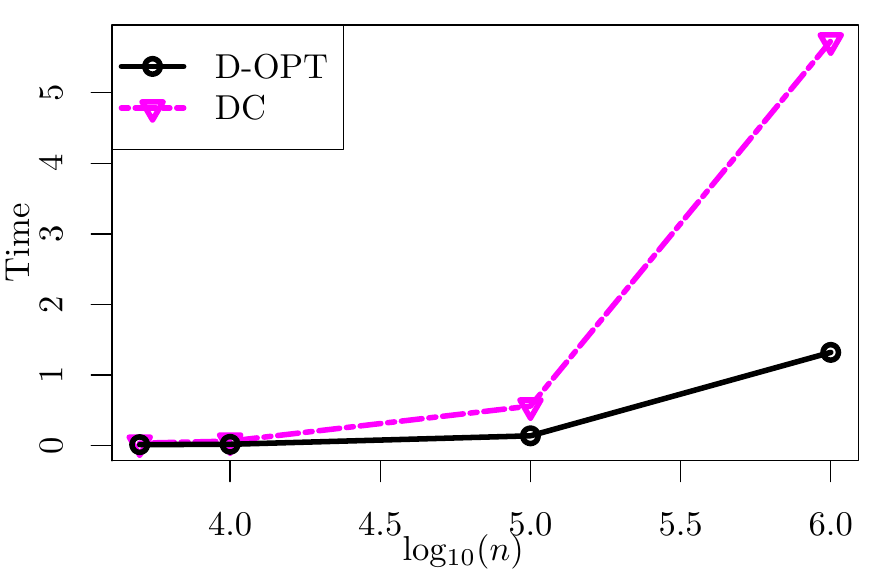}\\[-10mm]
    \caption{CPU times.}
  \end{subfigure}
  \caption{MSEs and CPU times for estimating the slope parameter:
    (a)-(c) give results for MSEs and (d) gives results for CPU
    times. The subdata size $k$ is fixed at $k=1000$ and the full data
    size $n$ changes with fix dimension $p=50$. Logarithm with base 10 is
    taken of $n$ and MSEs for better presentation of the figures.}
  \label{fig:s7}
\end{figure}

To further compare the IBOSS D-OPT method and the DC method with a larger $p$, we increase the dimension to be $p=500$. 
Figure~\ref{fig:s8} gives the average CPU times and MSEs for the slope parameters. Full data are generated with sample sizes $n=5\times10^3, 10^4, 10^5$, and $5\times10^5$. Subdata sample size for the IBOSS method is $k=1000$. It is seen that the relative performances of estimation efficiency between the IBOSS D-OPT method and the DC method depend on the covariate distribution are similar to those with $p=50$. In terms of computational cost in Figure~\ref{fig:s8} (d), the advantage of the IBOSS D-OPT method is more significant compared with the DC method.

\begin{figure}[H]
  \centering
  \begin{subfigure}{0.49\textwidth}
    \includegraphics[width=\textwidth]{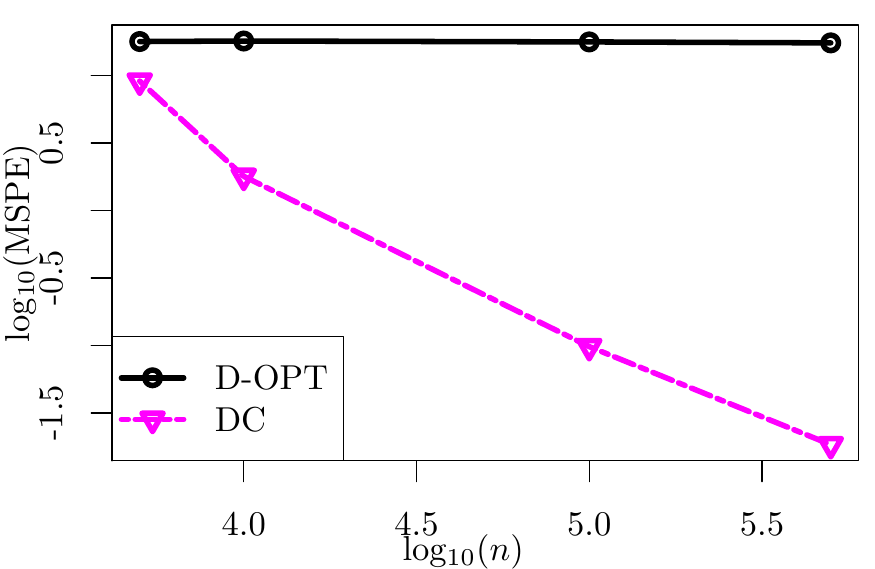}\\[-10mm]
    \caption{Case 1: $\z_i$'s are normal.}
  \end{subfigure}
  \begin{subfigure}{0.49\textwidth}
    \includegraphics[width=\textwidth]{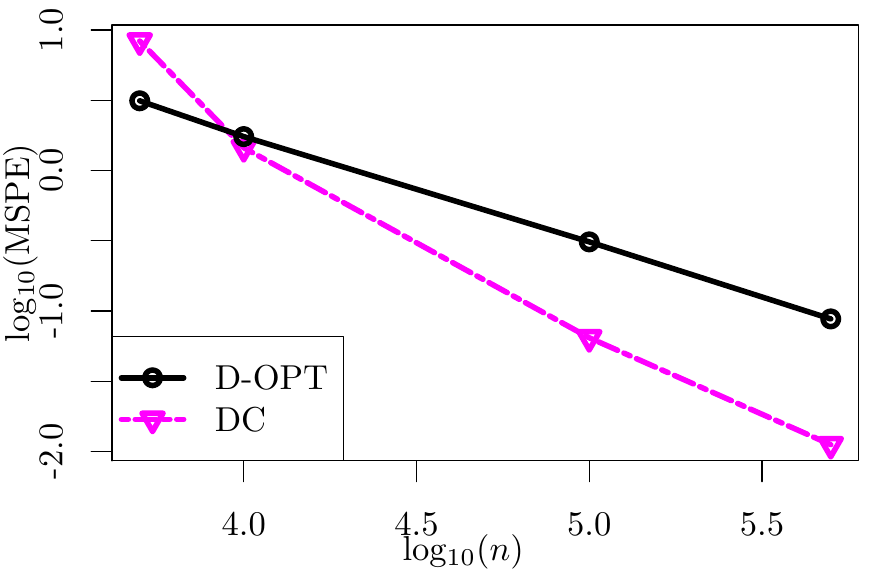}\\[-10mm]
    \caption{Case 4: $\z_i$'s are a mixture.}
  \end{subfigure}\\[3mm]
  \begin{subfigure}{0.49\textwidth}
    \includegraphics[width=\textwidth]{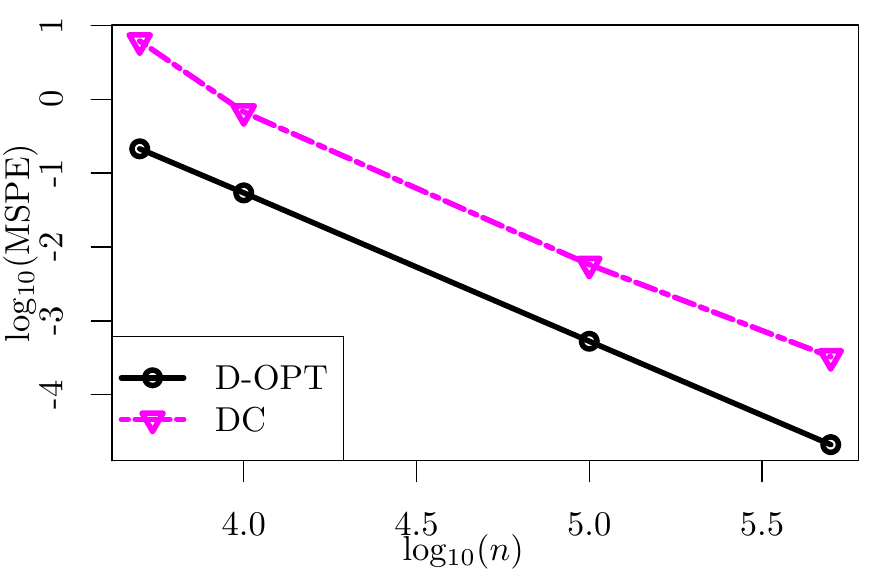}\\[-10mm]
    \caption{Case 6: $\z_i$'s are $t_1$.}
  \end{subfigure}
  \begin{subfigure}{0.49\textwidth}
    \includegraphics[width=\textwidth]{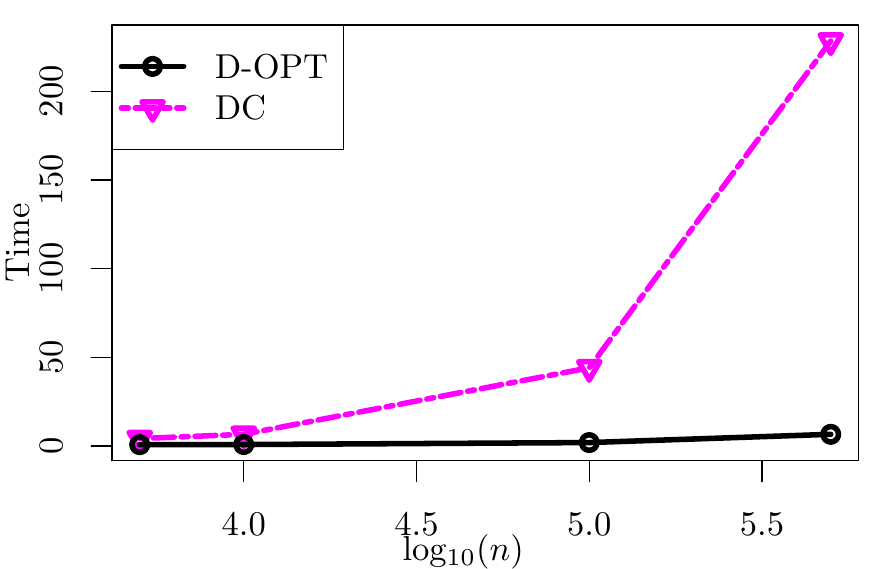}\\[-10mm]
    \caption{CPU times.}
  \end{subfigure}
  \caption{MSEs and CPU times for estimating the slope parameter:
    (a)-(c) give results for MSEs and (d) gives results for CPU
    times. The subdata size $k$ is fixed at $k=1000$ and the full data
    size $n$ changes with fixed dimension $p=500$. Logarithm with base 10 is
    taken of $n$ and MSEs for better presentation of the figures.}
  \label{fig:s8}
\end{figure}

\section{Performance of IBOSS with regularization method}
In this section, we provide numerical results to evaluate the performance of the IBOSS method in application to regularization methods. We use the IBOSS method to select subdata, and then feed it to the elastic net regularization\citep{zou2005regularization} method. 
Full data with dimension $p=60$ are generated for sample sizes $n$, with choices of $5\times10^3, 10^4, 10^5$, and $10^6$. The intercept is set to $\beta_0=1$, while the slope parameter $\bbeta_1$ has a sparse structure with the first 10 element being 0.1 and the rest 50 element being 0. The elastic net method is implemented using the glmnet R package \citep{friedman2010glmnet}. Tuning parameters are selected using the cross validation method provided in the R package.

We calculate the MSPEs based on 100 repetitions of the simulation. In each repetition, we implement different methods to obtain a subdata set of $k=1000$, apply the elastic net to the subdata set to estimate a model, and then use the model to calculate the MSPEs based on a new sample of size 5,000. Figure~\ref{fig:s9} presents the results of the simulation. It is seen that the relative performance of IBOSS compared with
other methods are similar to that of parameter estimation in the main paper. That is, the D-OPT IBOSS method uniformly dominates the subsampling-based
methods UNI and LEV, and its advantage is more significant if the
tail of the covariate distribution is heavier.

We also implement the ridge regression method. The results are similar so we omit them. 
\begin{figure}[H]
  \centering
  \begin{subfigure}{0.49\textwidth}
    \includegraphics[width=\textwidth]{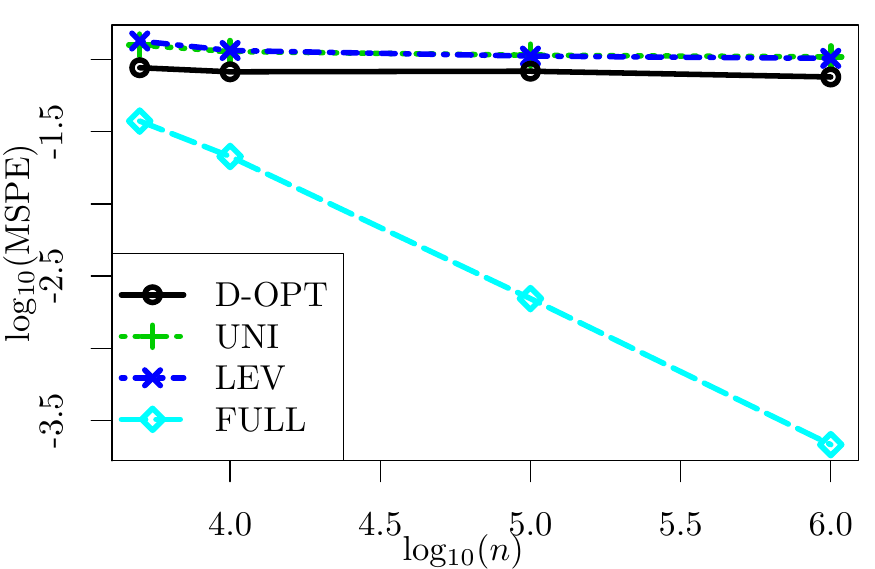}\\[-10mm]
    \caption{Case 1: $\z_i$'s are normal.}
  \end{subfigure}
  \begin{subfigure}{0.49\textwidth}
    \includegraphics[width=\textwidth]{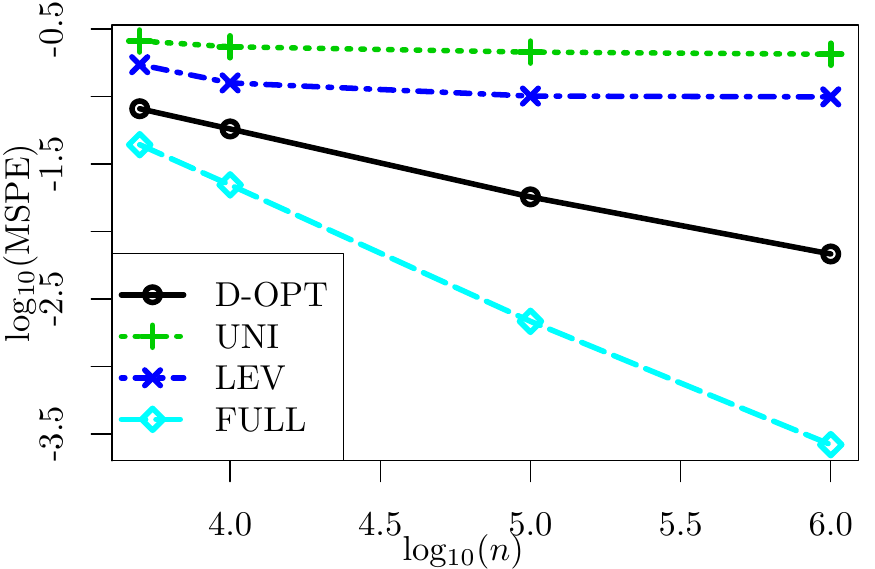}\\[-10mm]
    \caption{Case 2: $\z_i$'s are lognormal.}
  \end{subfigure}\\[3mm]
  \begin{subfigure}{0.49\textwidth}
    \includegraphics[width=\textwidth]{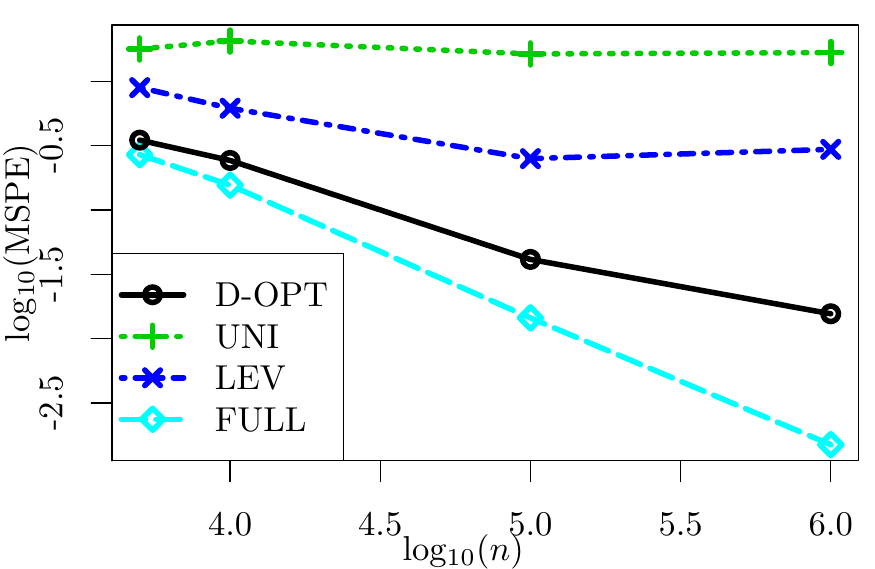}\\[-10mm]
    \caption{Case 3: $\z_i$'s are $t_2$.}
  \end{subfigure}
  \begin{subfigure}{0.49\textwidth}
    \includegraphics[width=\textwidth]{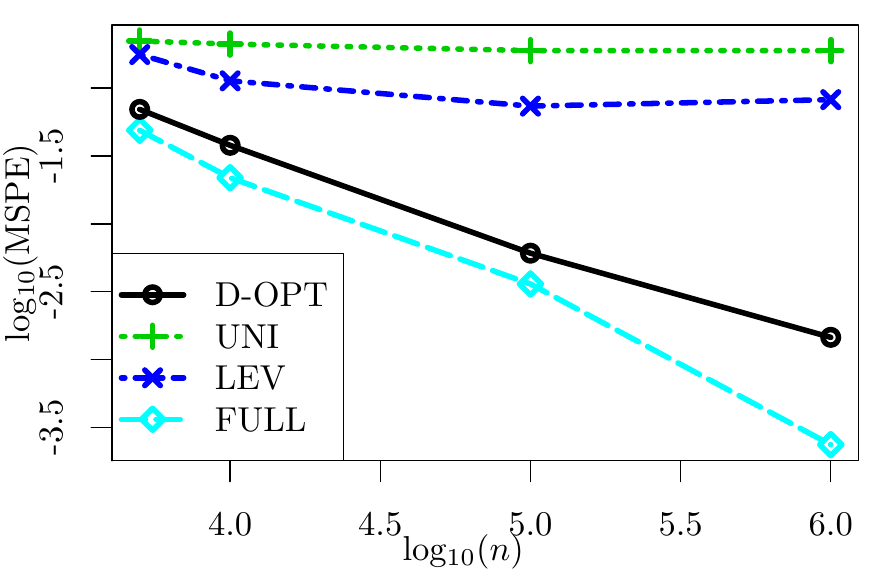}\\[-10mm]
    \caption{Case 4: $\z_i$'s are a mixture.}
  \end{subfigure}
  \caption{MSPEs for predicting mean responses using the elastic net
    method with the subdata of size $k=1000$ selected from the full data. Logarithm with base 10 is taken of the full data sample
    size $n$ and MSPEs for better presentation of the figures.}
  \label{fig:s9}
\end{figure}

\section{Unequal variance}
In this section, we provide a simple numerical study to evaluate the
performance of the IBOSS method when the error term in the linear
model is heteroscedastic. We use same setup in the main paper to
generate the full data except that the standard deviations of the
error terms are different and are generated from the exponential
distribution with rate parameter 1, i.e., the variance for each error term is randomly generated from a squared exponential random variable. Figure~\ref{fig:s10} presents MSE
for estimating the slope parameter. It is seen that the relative
performance of IBOSS compared with other methods are similar to that
of parameter estimation in the main paper. That is, the D-OPT IBOSS
method uniformly dominates the subsampling-based methods UNI and LEV,
and its advantage is more significant if the tail of the covariate
distribution is heavier. Note that when the error terms have unequal variances, transformations are often used to stabilize the variances or weighted least squares are often used instead of the ordinal least squares. These questions are beyond the scope of this paper and we will investigate them in another project.

\begin{figure}[H]
  \centering
  \begin{subfigure}{0.49\textwidth}
    \includegraphics[width=\textwidth]{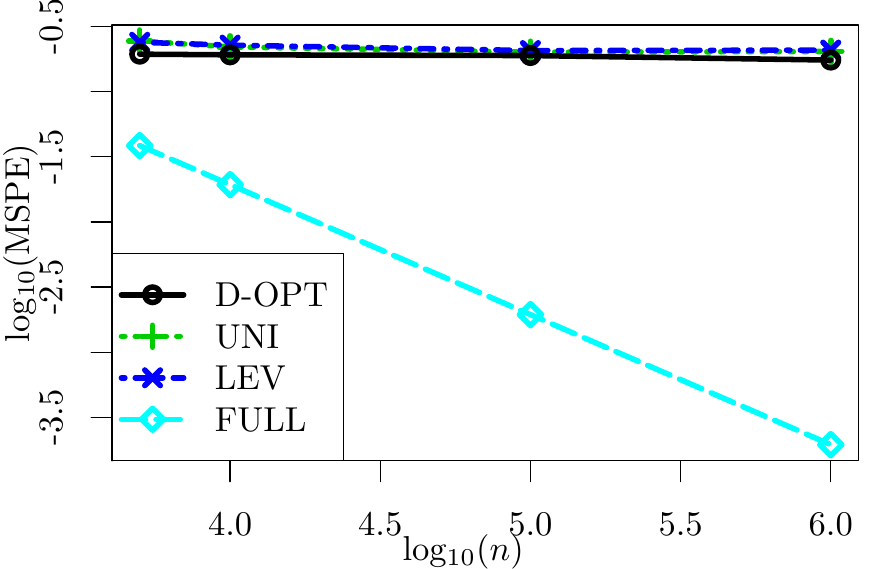}\\[-10mm]
    \caption{Case 1: $\z_i$'s are normal.}
  \end{subfigure}
  \begin{subfigure}{0.49\textwidth}
    \includegraphics[width=\textwidth]{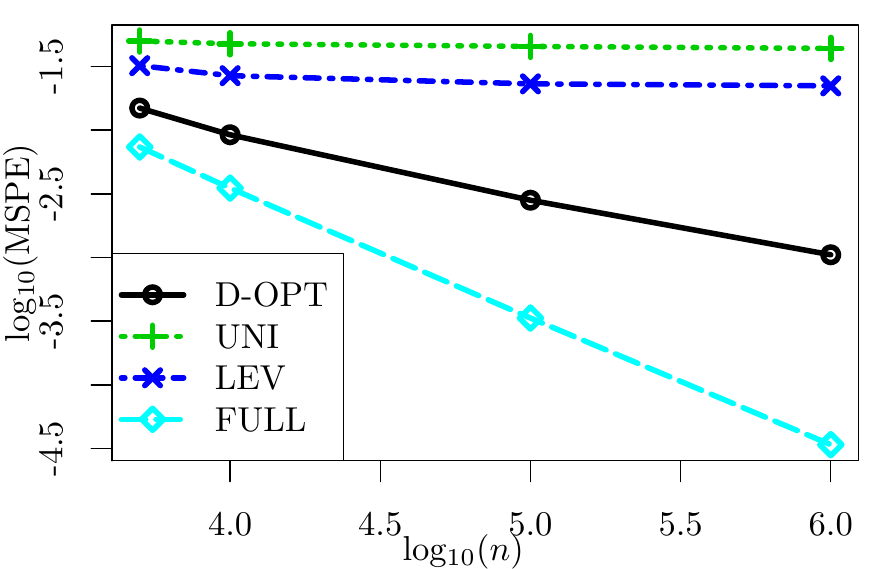}\\[-10mm]
    \caption{Case 2: $\z_i$'s are lognormal.}
  \end{subfigure}\\[3mm]
  \begin{subfigure}{0.49\textwidth}
    \includegraphics[width=\textwidth]{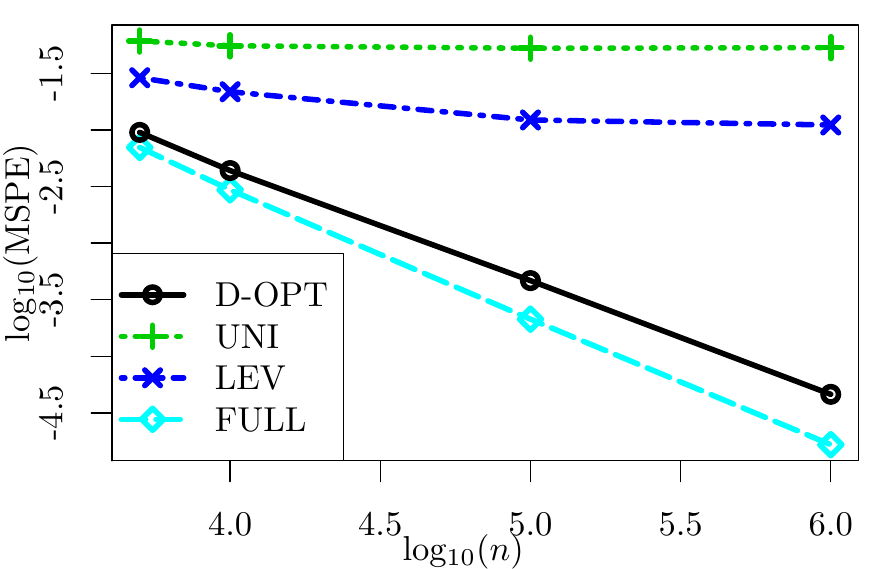}\\[-10mm]
    \caption{Case 3: $\z_i$'s are $t_2$.}
  \end{subfigure}
  \begin{subfigure}{0.49\textwidth}
    \includegraphics[width=\textwidth]{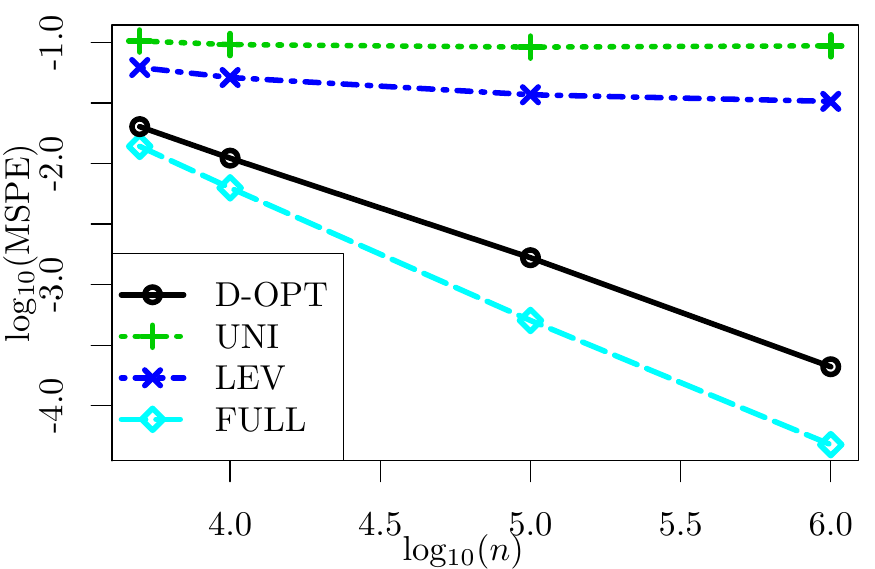}\\[-10mm]
    \caption{Case 4: $\z_i$'s are a mixture.}
  \end{subfigure}
  \caption{MSEs for estimating the slope parameter when the error
    terms are heteroscedastic. The subdata size $k$ is fixed at
    $k=1000$ and the full data size $n$ changes. Logarithm with base
    10 is taken of $n$ and MSEs for better presentation of the
    figures.}
  \label{fig:s10}
\end{figure}

\end{document}